\documentclass[11pt]{amsart}
\usepackage[applemac]{inputenc}
\usepackage[titletoc]{appendix}

\usepackage[backend=biber, style=numeric, date=year, url=false, doi=false, isbn=false, eprint=true, firstinits=true, maxcitenames=5, maxbibnames=5]{biblatex}
\AtEveryBibitem{%
  \ifentrytype{online}{%
  }{%
    \clearfield{eprint}%
    \clearfield{urldate}%
  }%
}
\usepackage{blindtext}
\usepackage{amssymb}
\usepackage{amsfonts}
\usepackage{amsthm}
\usepackage{tikz}
\usepackage{caption}
\usepackage{color}
\usepackage{graphicx}
\usepackage{xcolor}
\usepackage{shadethm}
\usepackage{subfigure}
\usepackage{epigraph}
\usepackage{placeins}

\usepackage{enumerate}
\usepackage{verbatim}
\usetikzlibrary{trees}
\usetikzlibrary{decorations.markings}
\usetikzlibrary{arrows}
\usepackage{mathrsfs}
\usepackage[margin=1.0in]{geometry}
\usepackage{xparse}
\usepackage{diagrams}

\usetikzlibrary{positioning,arrows,patterns}
\usetikzlibrary{decorations.markings}
\usetikzlibrary{calc}
\tikzset{
  photon/.style={decorate, decoration={snake}, draw=black},
  fermion/.style={draw=black, postaction={decorate},decoration={markings,mark=at position .55 with {\arrow{>}}}},
  vertex/.style={draw,shape=circle,fill=black,minimum size=5pt,inner sep=0pt},
particle/.style={thick,draw=black},
particle2/.style={thick,draw=blue},
avector/.style={thick,draw=black, postaction={decorate},
    decoration={markings,mark=at position 1 with {\arrow[black]{triangle 45}}}},
gluon/.style={decorate, draw=black,
    decoration={coil,aspect=0}}
 }
\NewDocumentCommand\semiloop{O{black}mmmO{}O{above}}
{%
\draw[#1] let \p1 = ($(#3)-(#2)$) in (#3) arc (#4:({#4+180}):({0.5*veclen(\x1,\y1)})node[midway, #6] {#5};)
}

\usepackage[mathcal]{euscript}

\usepackage{url}
\usepackage{hyperref}
\hypersetup{colorlinks=true, citecolor=red, linktocpage, linkcolor=blue, urlcolor=cyan} 

\theoremstyle{plain}
\newtheorem{thm}{Theorem}[section]
\newtheorem{lem}[thm]{Lemma}
\newtheorem{prop}[thm]{Proposition}

\theoremstyle{definition}
\newtheorem{defn}[thm]{Definition}
\newtheorem{notation}[thm]{Notation}
\newtheorem{rem}[thm]{Remark}

\newtheorem{ex}[thm]{Example}
\newtheorem{ass}[thm]{Assumption}
\newtheorem*{thm*}{Theorem}
\newtheorem*{lem*}{Lemma}
\newtheorem*{prop*}{Proposition}
\newtheorem*{cor*}{Corollary}
\newtheorem*{exe*}{Exercise}
\newtheorem*{defn*}{Definition}
\newtheorem*{rem*}{Remark}

\theoremstyle{remark}

\newcommand{\R}{\mathbb{R}}
\newcommand{\Z}{\mathbb{Z}}
\newcommand{\A}{\mathcal{A}}

\newcommand{\E}{\mathbb{E}} 
\newcommand{\bbE}{\mathbb{E}} 
\newcommand{\bbX}{\mathbb{X}} 
\newcommand{\calU}{\mathcal{U}}
\newcommand{\calD}{\mathcal{D}}
\newcommand{\calY}{\mathcal{Y}}

\newcommand{\ii}{{\mathrm{i}}}
\newcommand{\dd}{{\mathrm{d}}}

\newcommand{\id}{\mathrm{id}}

\newcommand{\C}{\mathbb{C}}

\DeclareMathOperator{\End}{End}

\newcommand{\T}{\mathsf{T}}

\newcommand{\bbQ}{{\mathbb{Q}}}

\newcommand{\de}{\partial}

\newcommand{\calB}{\mathcal{B}}
\newcommand{\calH}{\mathcal{H}}
\newcommand{\calS}{\mathcal{S}}

\newcommand{\calG}{\mathcal{G}}

\newcommand{\calL}{\mathcal{L}}
\newcommand{\calM}{\mathcal{M}}

\newcommand{\calE}{\mathcal{E}}
\newcommand{\calP}{\mathcal{P}}
\newcommand{\calF}{\mathcal{F}}

\newcommand{\btpsi}{\boldsymbol{\widetilde{\psi}}}
\newcommand{\sfeta}{\boldsymbol{\eta}}

\newcommand{\sfX}{{\mathsf{X}}}

\newcommand{\qtconn}{\nabla_\mathsf{G}}

\def\gpd{\,\lower1pt\hbox{$\longrightarrow$}\hskip-.24in\raise2pt
               \hbox{$\longrightarrow$}\,}

\let\Tilde=\widetilde

\let\Hat=\widehat

\newcommand{\hateta}{\widehat{\boldsymbol\eta}}
\newcommand{\hatX}{\widehat{\mathsf{X}}}

\DeclareMathOperator{\dr}{d}
\DeclareMathOperator{\Map}{Map}

\newcommand{\I}{\mathrm{i}}

\newcommand{\ee}{\textnormal{e}}

\newcommand{\calV}{\mathcal{V}}

\usepackage{todonotes}

\makeatletter

\pgfdeclareshape{genus}{
 \anchor{center}{\pgfpointorigin}
\backgroundpath{
    \begingroup
    \tikz@mode
    \iftikz@mode@fill

         \pgfpathmoveto{\pgfqpoint{-0.78cm}{-.17cm}}
     \pgfpathcurveto %
            {\pgfpoint{-0.35cm}{-.44cm}}
        {\pgfpoint{0.35cm}{-.44cm}}
        {\pgfpoint{.78cm}{-0.17cm}} 
     \pgfpathmoveto{\pgfqpoint{-0.78cm}{-0.17cm}}
     \pgfpathcurveto %
            {\pgfpoint{-0.25cm}{.25cm}}
        {\pgfpoint{.25cm}{.25cm}}
        {\pgfpoint{0.78cm}{-0.17cm}}
        \pgfusepath{fill}
        \fi

        \iftikz@mode@draw
     \pgfpathmoveto{\pgfqpoint{-1cm}{0cm}}
     \pgfpathcurveto %
            {\pgfpoint{-0.5cm}{-.5cm}}
        {\pgfpoint{0.5cm}{-.5cm}}
        {\pgfpoint{1cm}{0cm}}

     \pgfpathmoveto{\pgfqpoint{-0.75cm}{-0.15cm}}
     \pgfpathcurveto %
            {\pgfpoint{-0.25cm}{.25cm}}
        {\pgfpoint{.25cm}{.25cm}}
        {\pgfpoint{0.75cm}{-0.15cm}}
              \pgfusepath{stroke}
        \fi
        \endgroup
    }
    }

    \makeatother

\usepackage[symbols,sort=use]{glossaries}

\makeglossaries

\newglossaryentry{M}
{
  name={\ensuremath{M}},
  description={a finite-dimensional manifold},
  symbol={\ensuremath{M}}
}
\newglossaryentry{S_M}
{
  name={\ensuremath{S_M}},
  description={local action functional},
  symbol={\ensuremath{S_M}}
}
\newglossaryentry{F_M}
{
  name={\ensuremath{F_M}},
  description={space of fields},
  symbol={\ensuremath{F_M}}
}
\newglossaryentry{calF_M}
{
  name={\ensuremath{\calF_M}},
  description={BV space of fields},
  symbol={\ensuremath{\calF_M}}
}
\newglossaryentry{calS_M}
{
  name={\ensuremath{\calS_M}},
  description={BV action functional},
  symbol={\ensuremath{\calS_M}}
}
\newglossaryentry{calF^de_M}
{
  name={\ensuremath{\calF^\de_M}},
  description={BFV space of boundary fields},
  symbol={\ensuremath{\calF^\de_M}}
}
\newglossaryentry{calS^de}
{
  name={\ensuremath{\calS^\de}},
  description={BFV boundary action},
  symbol={\ensuremath{\calS^\de_M}}
}
\newglossaryentry{Delta}
{
  name={\ensuremath{\Delta}},
  description={global BV Laplacian on half-densitites},
  symbol={\ensuremath{\Delta}}
}
\newglossaryentry{psi_M}
{
  name={\ensuremath{\psi_M}},
  description={quantum state},
  symbol={\ensuremath{\psi_M}}
}
\newglossaryentry{V_M}
{
  name={\ensuremath{\calV_M}},
  description={space of residual fields},
  symbol={\ensuremath{\calV_M}}
}
\newglossaryentry{calH_deM}
{
  name={\ensuremath{\calH_{\de M}}},
  description={space of boundary states},
  symbol={\ensuremath{\calH_{\de M}}}
}
\newglossaryentry{Omega_deM}
{
  name={\ensuremath{\Omega_{\de M}}},
  description={BFV boundary operator},
  symbol={\ensuremath{\Omega_{\de M}}}
}
\newglossaryentry{Dens}
{
  name={\ensuremath{\textnormal{Dens}^{\frac{1}{2}}(M)}},
  description={half-densities on a manifold \ensuremath{M}},
  symbol={\ensuremath{\textnormal{Dens}^{\frac{1}{2}}(M)}}
}
\newglossaryentry{calP}
{
  name={\ensuremath{\calP}},
  description={polarization of $\calF^\de_M$},
  symbol={\ensuremath{\calP}}
}
\newglossaryentry{calH^calP_deM}
{
  name={\ensuremath{\calH^\calP_{\de M}}},
  description={full space of boundary states},
  symbol={\ensuremath{\calH^\calP_{\de M})}}
}
\newglossaryentry{calH^calP(princ)_deM}
{
  name={\ensuremath{\calH^{\calP,\textnormal{princ}}_{\de M}}},
  description={principal space of boundary states},
  symbol={\ensuremath{\calH^{\calP,\textnormal{princ}}_{\de M}}}
}
\newglossaryentry{Conf}
{
  name={\ensuremath{\mathsf{Conf}_\Gamma(M)}},
  description={configuration space},
  symbol={\ensuremath{\mathsf{Conf}_\Gamma(M)}}
}
\newglossaryentry{C}
{
  name={\ensuremath{\mathsf{C}_\Gamma(M)}},
  description={compactified configuration space},
  symbol={\ensuremath{\mathsf{C}_\Gamma}(M)}
}
\newglossaryentry{HatcalH_M}
{
  name={\ensuremath{\Hat{\calH}_M}},
  description={the space \ensuremath{\textnormal{Dens}^\frac{1}{2}(\calV_M)\otimes \calH_{\de M}}},
  symbol={\ensuremath{\Hat{\calH}_M}}
}
\newglossaryentry{B}
{
  name={\ensuremath{\calB^\calP_{\de M}}},
  description={leaf space of the polarization $\calP$},
  symbol={\ensuremath{\calB^\calP_{\de M}}}
}
\newglossaryentry{fullstate}
{
  name={\ensuremath{\boldsymbol{\psi}_M}},
  description={full quantum state},
  symbol={\ensuremath{\boldsymbol{\psi}_M}}
}
\newglossaryentry{Omega^princ}
{
  name={\ensuremath{\Omega^{\textnormal{princ}}}},
  description={principal part of the BFV boundary operator},
  symbol={\ensuremath{\Omega^{\textnormal{princ}}}}
}
\newglossaryentry{Omega^X_0}
{
  name={\ensuremath{\Omega^\mathbb{X}_0}},
  description={\ensuremath{\mathbb{X}}-part of the free BFV boundary operator},
  symbol={\ensuremath{\Omega^\mathbb{X}_0}}
}
\newglossaryentry{Omega^E_0}
{
  name={\ensuremath{\Omega^\mathbb{E}_0}},
  description={\ensuremath{\mathbb{E}}-part of the free BFV boundary operator},
  symbol={\ensuremath{\Omega^\mathbb{E}_0}}
}
\newglossaryentry{Omega^E_pert}
{
  name={\ensuremath{\Omega^\mathbb{E}_\textnormal{pert}}},
  description={\ensuremath{\mathbb{E}}-part of the perturbative BFV boundary operator},
  symbol={\ensuremath{\Omega^\mathbb{E}_\textnormal{pert}}}
}
\newglossaryentry{Omega^X_pert}
{
  name={\ensuremath{\Omega^\mathbb{X}_\textnormal{pert}}},
  description={\ensuremath{\mathbb{X}}-part of the perturbative BFV boundary operator},
  symbol={\ensuremath{\Omega^\mathbb{X}_\textnormal{pert}}}
}
\newglossaryentry{fullOmega_deM}
{
  name={\ensuremath{\boldsymbol{\Omega}_{\de M}}},
  description={full BFV boundary operator},
  symbol={\ensuremath{\boldsymbol{\Omega}_{\de M}}}
}
\newglossaryentry{phi}
{
  name={\ensuremath{\varphi}},
  description={formal exponential map},
  symbol={\ensuremath{\varphi}}
}
\newglossaryentry{T}
{
  name={\ensuremath{\mathsf{T}}},
  description={Taylor expansion},
  symbol={\ensuremath{\mathsf{T}}}
}
\newglossaryentry{HatS}
{
  name={\ensuremath{\Hat{S}}},
  description={completed symmetric algebra},
  symbol={\ensuremath{\Hat{S}}}
}
\newglossaryentry{D}
{
  name={\ensuremath{D_\mathsf{G}}},
  description={classical Grothendieck connection},
  symbol={\ensuremath{D_\mathsf{G}}}
}
\newglossaryentry{R}
{
  name={\ensuremath{R}},
  description={vector field in the definition of \ensuremath{D_\mathsf{G}}},
  symbol={\ensuremath{R}}
}
\newglossaryentry{TildeS_Sigma,x}
{
  name={\ensuremath{\Tilde{\calS}_{\Sigma,x}}},
  description={formal globalized action},
  symbol={\ensuremath{\Tilde{\calS}_{\Sigma,x}}}
}
\newglossaryentry{Sigma}
{
  name={\ensuremath{\Sigma}},
  description={worldsheet manifold with boundary in the Poisson Sigma Model},
  symbol={\ensuremath{\Sigma}}
}
\newglossaryentry{hatX}
{
  name={\ensuremath{\hatX}},
  description={lift of \ensuremath{\mathsf{X}} by \ensuremath{\varphi_x}},
  symbol={\ensuremath{\hatX}}
}
\newglossaryentry{hateta}
{
  name={\ensuremath{\hateta}},
  description={lift of \ensuremath{\boldsymbol{\eta}} by \ensuremath{\dd\varphi_x}},
  symbol={\ensuremath{\hateta}}
}
\newglossaryentry{Tildepsi_Sigma,x}
{
  name={\ensuremath{\Tilde{\psi}_{\Sigma,x}}},
  description={principal covariant quantum state},
  symbol={\ensuremath{\Tilde{\psi}_{\Sigma,x}}}
}
\newglossaryentry{btpsi_Sigma,x}
{
  name={\ensuremath{\btpsi_{\Sigma,x}}},
  description={full covariant quantum state},
  symbol={\ensuremath{\btpsi_{\Sigma,x}}}
}
\newglossaryentry{x}
{
  name={\ensuremath{x}},
  description={constant background field},
  symbol={\ensuremath{x}}
}
\newglossaryentry{X}
{
  name={\ensuremath{X}},
  description={Base map of the Poisson Sigma Model},
  symbol={\ensuremath{X}}
} 
\newglossaryentry{superX}
{
  name={\ensuremath{\mathsf{X}}},
  description={superfield version of \ensuremath{X}},
  symbol={\ensuremath{\mathsf{X}}}
} 
\newglossaryentry{eta}
{
  name={\ensuremath{\eta}},
  description={fiber map of the Poisson Sigma Model},
  symbol={\ensuremath{\eta}}
}
\newglossaryentry{superEta}
{
  name={\ensuremath{\boldsymbol{\eta}}},
  description={superfield version of \ensuremath{\eta}},
  symbol={\ensuremath{\boldsymbol{\eta}}}
}
\newglossaryentry{qGBFV}
{
  name={\ensuremath{\nabla_\mathsf{G}}},
  description={quantum Grothendieck BFV operator},
  symbol={\ensuremath{\nabla_\mathsf{G}}}
}
\newglossaryentry{Poisson_mnf}
{
  name={\ensuremath{\mathscr{P}}},
  description={Poisson manifold},
  symbol={\ensuremath{\mathscr{P}}}
}
\newglossaryentry{Poisson_str}
{
  name={\ensuremath{\Pi}},
  description={Poisson structure},
  symbol={\ensuremath{\Pi}}
}
\newglossaryentry{bbX}
{
  name={\ensuremath{\mathbb{X}}},
  description={boundary field part of \ensuremath{\mathsf{X}}},
  symbol={\ensuremath{\mathbb{X}}}
} 
\newglossaryentry{bbE}
{
  name={\ensuremath{\mathbb{E}}},
  description={boundary field part of \ensuremath{\boldsymbol{\eta}}},
  symbol={\ensuremath{\mathbb{E}}}
} 
\newglossaryentry{scrX}
{
  name={\ensuremath{\mathscr{X}}},
  description={fluctuation field part of \ensuremath{\mathsf{X}}},
  symbol={\ensuremath{\mathscr{X}}}
} 
\newglossaryentry{scrE}
{
  name={\ensuremath{\mathscr{E}}},
  description={fluctuation field part of \ensuremath{\boldsymbol{\eta}}},
  symbol={\ensuremath{\mathscr{E}}}
} 
\newglossaryentry{sfx}
{
  name={\ensuremath{\mathsf{x}}},
  description={residual field part of \ensuremath{\mathsf{X}}},
  symbol={\ensuremath{\mathsf{x}}}
} 
\newglossaryentry{sfe}
{
  name={\ensuremath{\mathsf{e}}},
  description={residual field part of \ensuremath{\boldsymbol{\eta}}},
  symbol={\ensuremath{\mathsf{e}}}
} 
\newglossaryentry{calD_G}
{
  name={\ensuremath{\mathcal{D}_\mathsf{G}}},
  description={deformed Grothendieck connection},
  symbol={\ensuremath{\calD_\mathsf{G}}}
} 
\newglossaryentry{A}
{
  name={\ensuremath{A}},
  description={connection term coming from Kontsevich's formality map},
  symbol={\ensuremath{A}}
}
\newglossaryentry{F}
{
  name={\ensuremath{F}},
  description={curvature term coming from Kontsevich's formality map},
  symbol={\ensuremath{F}}
}
\newglossaryentry{barD_G}
{
  name={\ensuremath{\overline{\mathcal{D}}_\mathsf{G}}},
  description={the flat connection \ensuremath{\calD_\mathsf{G}+[\gamma,\enspace]_\star}},
  symbol={\ensuremath{\overline{\calD}_\mathsf{G}}}
} 
\newglossaryentry{F^P}
{
  name={\ensuremath{F^\mathscr{P}}},
  description={Weyl curvature tensor of \ensuremath{\calD_\mathsf{G}}},
  symbol={\ensuremath{F^\mathscr{P}}}
}
\newglossaryentry{barF^P}
{
  name={\ensuremath{\overline{F}^\mathscr{P}}},
  description={Weyl curvature tensor of \ensuremath{\overline{\calD}_\mathsf{G}}},
  symbol={\ensuremath{\overline{F}^\mathscr{P}}}
}
\newglossaryentry{gamma}
{
  name={\ensuremath{\gamma}},
  description={a solution to \ensuremath{F^\mathscr{P}+\calD_\mathsf{G}\gamma+\gamma\star\gamma=0}},
  symbol={\ensuremath{\gamma}}
}
\newglossaryentry{twistedEOmega}
{
  name={\ensuremath{\Tilde{\boldsymbol{\Omega}}^{\E,\gamma}_{\de\Sigma}}},
  description={twisted \ensuremath{\E}-part of the full BFV boundary operator},
  symbol={\ensuremath{\Tilde{\boldsymbol{\Omega}}^{\E,\gamma}_{\de\Sigma}}}
}
\newglossaryentry{tildeXOmega}
{
  name={\ensuremath{\Tilde{\boldsymbol{\Omega}}^{\mathbb{X}}_{\de\Sigma}}},
  description={\ensuremath{\mathbb{X}}-part of the globalized full BFV boundary operator},
  symbol={\ensuremath{\Tilde{\boldsymbol{\Omega}}^{\mathbb{X}}_{\de\Sigma}}}
}
\newglossaryentry{tildeEOmega}
{
  name={\ensuremath{\Tilde{\boldsymbol{\Omega}}^{\mathbb{E}}_{\de\Sigma}}},
  description={\ensuremath{\mathbb{E}}-part of the globalized full BFV boundary operator},
  symbol={\ensuremath{\Tilde{\boldsymbol{\Omega}}^{\mathbb{E}}_{\de\Sigma}}}
}
\newglossaryentry{twistedpsi}
{
  name={\ensuremath{\Tilde{\boldsymbol{\psi}}^{\gamma}_{\Sigma,x}}},
  description={twisted full covariant quantum state},
  symbol={\ensuremath{\Tilde{\boldsymbol{\psi}}^{\gamma}_{\Sigma,x}}}
}
\newglossaryentry{modifiedformalglobalizedaction}
{
  name={\ensuremath{\mathbb{S}_{\Sigma,x}}},
  description={modified formal globalized action},
  symbol={\ensuremath{\mathbb{S}_{\Sigma,x}}}
}
\newglossaryentry{gammaaction}
{
  name={\ensuremath{\calS_{\Sigma,\gamma}}},
  description={\ensuremath{\gamma}-action term},
  symbol={\ensuremath{\calS_{\Sigma,\gamma}}}
}
\newglossaryentry{omegaaction}
{
  name={\ensuremath{\calS_{\Sigma,\omega}}},
  description={\ensuremath{\omega}-action term},
  symbol={\ensuremath{\calS_{\Sigma,\omega}}}
}
\newglossaryentry{calE}
{
  name={\ensuremath{\calE}},
  description={the deformed bundle \ensuremath{\Hat{S}T^*\mathscr{P}[[\varepsilon]]}},
  symbol={\ensuremath{\calE}}
}
\newglossaryentry{epsilon}
{
  name={\ensuremath{\varepsilon}},
  description={formal deformation parameter},
  symbol={\ensuremath{\varepsilon}}
}
\newglossaryentry{star_commutator}
{
  name={\ensuremath{[\enspace,\enspace]_\star}},
  description={star commutator},
  symbol={\ensuremath{[\enspace,\enspace]_\star}}
}
\newglossaryentry{star}
{
  name={\ensuremath{\star}},
  description={Kontsevich's star product},
  symbol={\ensuremath{\star}}
}
\newglossaryentry{hbar}
{
  name={\ensuremath{\hbar}},
  description={reduced Planck's constant},
  symbol={\ensuremath{\hbar}}
}
\newglossaryentry{scrC}
{
  name={\ensuremath{\mathscr{C}}},
  description={set of all corner points of \ensuremath{\Sigma}},
  symbol={\ensuremath{\mathscr{C}}}
}
\newglossaryentry{nabla^gamma_G}
{
  name={\ensuremath{\nabla^\gamma_\mathsf{G}}},
  description={twisted quantum Grothendieck BFV operator},
  symbol={\ensuremath{\nabla^\gamma_\mathsf{G}}}
}

\newglossaryentry{calH_C}
{
  name={\ensuremath{\calH_C}},
  description={space of corner states},
  symbol={\ensuremath{\calH_C}}
}
\newglossaryentry{ext_state_space}
{
  name={\ensuremath{\Hat{\calH}^\mathscr{C}_{\de\Sigma,x}}},
  description={extended state space},
  symbol={\ensuremath{\Hat{\calH}^\mathscr{C}_{\de\Sigma,x}}}
}
\newglossaryentry{tot_ext_state_space}
{
  name={\ensuremath{\Hat{\calH}^\mathscr{C}_{\de\Sigma,\textnormal{tot}}}},
  description={total extended state space},
  symbol={\ensuremath{\Hat{\calH}^\mathscr{C}_{\de\Sigma,\textnormal{tot}}}}
}
\newglossaryentry{tildenabla^gamma_G}
{
  name={\ensuremath{\Tilde{\nabla}^\gamma_\mathsf{G}}},
  description={twisted quantum Grothendieck BFV operator for corners},
  symbol={\ensuremath{\Tilde{\nabla}^\gamma_\mathsf{G}}}
}
\newglossaryentry{Omega_scrC}
{
  name={\ensuremath{\boldsymbol{\Omega}_\mathscr{C}}},
  description={corner contribution of the full BFV boundary operator},
  symbol={\ensuremath{\boldsymbol{\Omega}_\mathscr{C}}}
}
\newglossaryentry{i}
{
  name={\ensuremath{\I}},
  description={imaginary unit},
  symbol={\ensuremath{\I}}
}
\newglossaryentry{omega_M}
{
  name={\ensuremath{\omega_M}},
  description={BV symplectic form},
  symbol={\ensuremath{\omega_M}}
}
\newglossaryentry{omega^de}
{
  name={\ensuremath{\omega^\de}},
  description={BFV symplectic form},
  symbol={\ensuremath{\omega^\de}}
}
\newglossaryentry{omega}
{
  name={\ensuremath{\omega}},
  description={a \ensuremath{\star}-central \ensuremath{\calD_\mathsf{G}}-closed 2-form},
  symbol={\ensuremath{\omega}}
}

\begin{document}

\title[On the Globalization of the PSM in the BV-BFV formalism]{On the Globalization of the Poisson Sigma Model in the BV-BFV formalism}
\author[A. S. Cattaneo]{Alberto S. Cattaneo}
\author[N. Moshayedi]{Nima Moshayedi}
\author[K. Wernli]{Konstantin Wernli}
\address{Institut f\"ur Mathematik\\ Universit\"at Z\"urich\\ 
Winterthurerstrasse 190
CH-8057 Z\"urich}
\email[A. S.~Cattaneo]{cattaneo@math.uzh.ch}
\address{Institut f\"ur Mathematik\\ Universit\"at Z\"urich\\ 
Winterthurerstrasse 190
CH-8057 Z\"urich}
\email[N.~Moshayedi]{nima.moshayedi@math.uzh.ch}
\address{Institut f\"ur Mathematik\\ Universit\"at Z\"urich\\ 
Winterthurerstrasse 190
CH-8057 Z\"urich}
\email[K.~Wernli]{konstantin.wernli@math.uzh.ch}
\thanks{
This research was (partly) supported by the NCCR SwissMAP, funded by the Swiss National Science Foundation, and by the
COST Action MP1405 QSPACE, supported by COST (European Cooperation in Science and Technology).  We acknowledge partial support of SNF grant No. 200020\_172498/1. K. W. acknowledges partial support by the Forschungskredit of the University of Zurich, grant no. FK-16-093.
}
\maketitle

\begin{abstract} 
We construct a formal global quantization of the Poisson Sigma Model in the BV-BFV formalism using the perturbative quantization of AKSZ theories on manifolds with boundary and analyze the properties of the boundary BFV operator. Moreover, we consider mixed boundary conditions and show that they lead to quantum anomalies, i.e. to a failure of the (modified differential) Quantum Master Equation. We show that it can be restored by adding boundary terms to the action, at the price of introducing corner terms in the boundary operator. We also show that the quantum Grothendieck BFV operator on the total space of states is a differential, i.e. squares to zero, which is necessary for a well-defined BV cohomology. 
\end{abstract}

\tableofcontents


\section{Introduction}

\subsection{Motivation}
Symplectic groupoids are an important concept in Poisson and symplectic geometry \cite{W}.
A groupoid is a small category whose morphisms are invertible. 
We denote a groupoid
by $G\rightrightarrows M$, where $M$ is the set of objects and $G$ the set of morphisms.
A Lie groupoid is, roughly speaking, a groupoid where 
$M$ and $G$ are smooth manifolds and all structure maps are smooth. Finally, a symplectic groupoid is a Lie groupoid with a symplectic form $\omega\in\Omega^2(G)$ such that the graph of the multiplication  is a Lagrangian submanifold of $(G,\omega)\times(G,\omega)\times(G,-\omega)$. The manifold of objects $M$ has an induced Poisson
structure uniquely determined by requiring that the source map $G\to M$ is Poisson. A Poisson manifold $M$ that arises this way is called integrable. Not every Poisson manifold is integrable.  \\
The Poisson Sigma Model, \cite{SS1,SS2,I} is a 2-dimensional topological Sigma Model with target a Poisson manifold $\mathscr{P}$. The reduced phase space of the Poisson Sigma Model on an interval with target a Poisson manifold $\mathscr{P}$ is the source simply connected symplectic groupoid of $\mathscr{P}$ if $\mathscr{P}$ is integrable and otherwise a topological groupoid arising by singular symplectic reduction \cite{CF5}. \\
In \cite{C,CC1,CC2} it was shown that the space of 
classical boundary fields
always
has an interesting structure called \textsf{relational symplectic groupoid}. 
An relational symplectic groupoid is, roughly speaking, a groupoid in the ``extended category'' of symplectic manifolds where morphisms
are canonical relations. Recall that a canonical relation between two symplectic manifolds $(M_1,\omega_1)$ to
$(M_2,\omega_2)$ is an immersed Lagrangian submanifold of 
$(M_1,\omega_1)\times(M_2,-\omega_2)$.
The main structure of an relational symplectic groupoid $(\mathcal{G},\omega)$ is then given by an immersed Lagrangian submanifold 
$\mathcal{L}_1$ of $(\mathcal{G},\omega)$, which plays the role of unity, and by
an immersed Lagrangian submanifold $\mathcal{L}_3$ of
 $(\mathcal{G},\omega)\times(\mathcal{G},\omega)\times(\mathcal{G},-\omega)$,
 which plays the role
 of associative multiplication\footnote{There is also an immersed Lagrangian submanifold $\calL_2\in (\calG,\omega)\times (\calG,-\omega)$ representing the identity. The composition of $\calL_3$ with $\calL_2$ also defines an immersed Lagrangian submanifold of $(\calG,\omega)\times (\calG,\omega)\times(\calG,-\omega)$ that induces an equivalence relation and a quotient space which is precisely the symplectic reduction, so the symplectic groupoid $G$ in case $M$ is integrable.}. (In addition, there is also an antisymplectomorphism $\mathcal{I}$ of $\mathcal{G}$ that plays the role of the inversion map.)
 \\ 
The goal of this paper is another step towards deformation quantization of the relational symplectic groupoid through the Poisson Sigma Model, using the BV-BFV\footnote{Here the letters BFV stand for Batalin, Fradkin and Vilkovisky, who introduced what is now known as BV \cite{BV1,BV2} and BFV \cite{BF1,BF2,FV1,FV2} formalisms for gauge fixing.} formalism for the quantization of gauge theories on manifolds with boundary \cite{CMR1, CMR2}. This possible application of the BV-BFV formalism was first discussed in \cite{CMR2}. In \cite{CMW2} we explained how the quantization of the relational symplectic groupoid can be achieved in the case of constant Poisson structures. In \cite{CMW4}, we generalized the methods of formal geometry \cite{B,GK,} used in \cite{BCM,CMW2} to describe the perturbative quantization of any polarized AKSZ theory \cite{AKSZ}, possibly on manifolds with boundary. In that picture, the quantum state of the Poisson Sigma Model\footnote{We consider the Poisson Sigma Model as a perturbation around the trivial Poisson structure, so the moduli space of classical solutions on which $\btpsi$ is defined is identified with the target $\mathscr{P}$.} with target $\mathscr{P}$ is described by a section $\btpsi$ of a certain bundle over $\mathscr{P}$ which is closed with respect to an operator $\nabla_\mathsf{G}$: 
\begin{equation}
\nabla_\mathsf{G}\btpsi = 0. 
\end{equation} This equation is called the \textsf{modified differential Quantum Master Equation}. We will call $\nabla_\mathsf{G}$ the \textsf{quantum Grothendieck BFV operator}. In this paper we apply the results of \cite{CMW4} to the Poisson Sigma Model, and extend them to the more general case when we consider, in addition, boundary pieces with fixed boundary conditions. Typically we allow the different types to occur on different pieces of a single connected component of the boundary of the source manifold $\Sigma$.
We speak of alternating boundary conditions. These boundary conditions are required to define the relational symplectic groupoid on boundary fields of the Poisson Sigma Model.
\subsection{Main results}

Let us summarize the main results of the paper. 
We show that the introduction of alternating boundary conditions introduces a quantum anomaly, i.e. a failure of the closedness of $\btpsi$. In fact, we have:
\begin{prop*}\ref{prop_curv}
Consider the full state $\btpsi_{\Sigma,x}$ defined by $\Tilde{\calS}_{\Sigma,x}$ as in Definition \ref{full_state_2}. Then 
\begin{equation}
\nabla_\mathsf{G}\btpsi_{\Sigma,x}=\exp\left(\frac{\gls{i}}{\hbar}\int_{\partial_0\Sigma}F(R,R,\mathsf{T}\varphi^*_x\Pi)(\mathscr{X})\right)\btpsi_{\Sigma,x},
\end{equation}
where we integrate out the $\mathsf{X}$-fluctuation $\mathscr{X}$, which are the high energy part, in $F$ along $\de_0\Sigma$.
\end{prop*}
Here $F(R,R,\mathsf{T}\varphi^*_x\Pi)$ is defined in Appendix \ref{app:globalization} and is part of Kontsevich's $L_\infty$-morphism, and $\de_0\Sigma$ is a certain boundary component. 
Next, we show that by ``twisting'' the state and the operator $\nabla_\mathsf{G}$ by an appropriate Maurer--Cartan element (see Section \ref{sec:GlobalizedPSM}) the anomaly can be reduced to terms supported at the corners (i.e. points where boundary conditions change). We prove the following theorem: 
\begin{thm*}\ref{thm_corners}
Consider the twisted full state $\btpsi_{\Sigma,x}^\gamma$ defined in Definition \ref{full_covariant_state} and the twisted quantum Grothendieck BFV operator $\nabla^\gamma_\mathsf{G}$ defined in Definition \ref{defn_twisted_conn}.
Then
\begin{equation}
\nabla_{\mathsf{G}}^\gamma\btpsi_{\Sigma,x}^\gamma=\sum_{C\in\mathscr{C}_1}T(C)\btpsi_{\Sigma,x}^\gamma,
\end{equation}
where $T(C)$ are functionals on $\calB^\calP_{\de\Sigma}$ with values in $\Omega^1(\mathscr{P})$, depending only on the values of the fields at the corner point $C$.
\end{thm*}
We show that this twisted operator also squares to zero (Remark \ref{rem_flatness_twisted}).
However, we want again to interpret the state as a closed section with respect to a certain operator that squares to zero. In Section \ref{sec:mdQME} we show that  this can  be done by enlarging the space of states (see Definition \ref{defn_ext_state_space}) and defining a new operator $\Tilde{\nabla}^\gamma_\mathsf{G}$  (see Equation \eqref{eq:conntilde}) on the new bundle of states. We show the following theorem: 
\begin{thm*}[modified differential Quantum Master Equation for alternating boundary conditions (\ref{thm_mdQME_corners})]
Let $\Tilde{\nabla}^\gamma_\mathsf{G}$ be given as in Equation \eqref{eq:conntilde}, and consider the twisted full state $\btpsi_{\Sigma,x}^\gamma$. Then 
\begin{equation}
\label{eq:mdQME_alt_boundary}
\Tilde{\nabla}^\gamma_\mathsf{G}\btpsi_{\Sigma,x}^\gamma=0
\end{equation}
\end{thm*}
We also show that the new operator $\Tilde{\nabla}^\gamma_\mathsf{G}$ again squares to zero:
\begin{thm*}\ref{thm:flatness}
The operator $\Tilde{\nabla}^\gamma_\mathsf{G}$ squares to zero, i.e. $(\Tilde{\nabla}^\gamma_{\mathsf{G}})^2=0$.
\end{thm*}
\subsection{Summary}
Let us give a brief overview of the paper. 
\begin{itemize}
\item{
In Section \ref{sec:BVBFV} we give a very rough review of the classical and quantum BV-BFV formalism. For more details the reader is referred to the literature \cite{CMR1,CMR2}. In particular we recall the Quantum Master Equation and its generalization to manifolds with boundary, called the modified Quantum Master Equation. 
}\\
\item{
In Section \ref{sec:AKSZ} we recall the construction, and the results, of \cite{CMW4}. Most importantly, to apply the quantum BV-BFV formalism one needs to linearize the target around constant maps, which form a part of the moduli space of classical solutions of any polarized AKSZ theory \cite{AKSZ}. For nonlinear targets, this can be done in a covariant way, as one varies the image of the constant map. In a natural way this leads to a family of quantizations parametrized by the target that satisfy a generalization of the modified Quantum Master Equation, that we call the modified differential Quantum Master Equation. This equation can be interpreted as the closedness of the state with respect to the quantum Grothendieck BFV operator  $\qtconn$ that squares to zero. Moreover, under change of gauge choices the state changes by a $\qtconn$-exact term, so that there is a certain cohomology describing the physical states. 
}\\
\item{
In Section \ref{sec:PSM} we recall the classical Poisson Sigma Model and its BV-BFV extension \cite{CMR1}.
}\\
\item{In Section \ref{sec:GlobalizedPSM} we apply the results recalled in Section \ref{sec:AKSZ} to the Poisson Sigma Model, which is an example of an AKSZ theory. In particular, we describe the algebraic structure which is captured in the modified differential Quantum Master Equation and the fact that $\qtconn$ squares to zero. We also describe how to twist the theory by a certain 1-form $\gamma$, which produces a new state and a new operator. 
}\\
\item{
In Section \ref{sec:ABC} we discuss what happens when one combines the globalization of the partition function over constant maps with the alternating or mixed boundary conditions that appear in the construction of the relational symplectic groupoid. In particular, we describe an anomaly that arises from the curvature\footnote{This is related to the curvature that appears in the globalization of Kontsevich's star product, see e.g. \cite{CFT}.} of the deformed Grothendieck connection $\calD_G$, and how it can be cancelled by a quantum counterterm in the action. We also describe how the modified differential Quantum Master Equation gets spoilt by terms that come from the corners where the different boundary conditions meet. 
}\\
\item{
In Section \ref{sec:mdQME} we explain how one can restore the modified differential Quantum Master Equation for the Poisson Sigma Model with alternating boundary conditions. For this one has to extend both the space of operators and the space of states, and we define these extensions in Section \ref{subsec:ex_states} and \ref{subsec:ex_op}. We  prove that there is an extension of the twisted operator for which the state defines a closed section and that the extended operator also squares to zero. 
}\\
\item{
Finally, in Section \ref{sec:Outlook} we explain directions for further research. These are not restricted to the deformation quantization of the relational symplectic groupoid. The methods developed in this paper could help understand both the globalization of other theories with more complicated moduli spaces of classical solutions, and the ``extended'' quantization (in the sense of extended TQFTs) of AKSZ theories on manifolds with corners (and possibly, defects of higher codimension). 
}
\end{itemize}
\vspace{0.3cm}

Various details are discussed in the appendices: 

\begin{itemize}
\item{In Appendix \ref{app:Conf} we recall the compactification of various configuration spaces and their boundary strata.}\\
\item{In Appendix \ref{app:PSM} we recall the globalization of Kontsevich's star product and its connection to the Poisson Sigma Model.}\\
\item{In Appendix \ref{app:prop} we recall the construction of a propagator for the Poisson Sigma Model with changing boundary conditions.}
\end{itemize}

\begin{rem*}
We provide a glossary of the most important symbols at the end of the paper.
\end{rem*}

\subsection*{Acknowledgements}
We thank I. Contreras for helpful comments. Moreover, we want to thank the two referees for providing us with important and helpful comments.

\section{The BV-BFV formalism}
\label{sec:BVBFV}

The BV-BFV formalism is a gauge fixing formalism for gauge theories on manifolds with boundary, both at the classical 
\cite{CMR1} and quantum \cite{CMR2} levels. We briefly  recall the most important ideas. 
Readers already familiar with the BV-BFV formalism  as in \cite{CMR2} can skip this section. Another reference for learning about this formalism is \cite{CattMosh1}.

\subsection{Field theory}
We start with the following definition of a classical field theory. 
\begin{defn}[Classical field theory]
A $d$-dimensional \textsf{classical field theory} associates to every compact $d$-dimensional manifold $\gls{M}$ (possibly with boundary) a space of fields $\gls{F_M}$ and an action functional $\gls{S_M}\colon F_M \to \R$.
\end{defn}
Field theories are usually required to be \textsf{local}. For the purpose of the present paper, the following definition will suffice. When we refer to a ``manifold'' $M$, we implicitly allow $M$ to come equipped with background fields (e.g. a metric) upon which the field theory is allowed to depend\footnote{This in particular will also allow us to consider 2D Yang--Mills theory in this formalism.}.
\begin{defn}[Local field theory]
We say that a field theory $(F_M,S_M)$ is \textsf{local} if there is a fiber bundle $E \to M$ such that $F_M = \Gamma(E)$ and there is an integer $k$ such that
\begin{equation}
S_M(\phi) = \int_M L[j^k(\phi)],\label{eq:Lagrangian_field_theory}
\end{equation}
where $j^k$ denotes the $k$-th jet prolongation, and $L \colon J^kE \to \mathrm{Dens}(M)$ is a function on the $k$-th jet bundle of $E$ with values in densities of $M$. $L$ is called the \textsf{Lagrangian} of the theory. 
\end{defn}
Let $(F_M,S_M)$ be a local  field theory. If $M\not=\varnothing$ and we do not fix any boundary conditions, there is a $1$-form $\alpha_{\de M}\in\Omega^1(F_{\de M})$ (the \textsf{Noether} $1$-form) such that the variation of the action $S_M$ is given by 
\begin{equation}
\delta S_M:=\text{EL}_M+\pi^*_M\alpha_{\de M},
\end{equation}
where $\pi_M\colon F_M\to F_{\de M}$ is the natural surjective submersion from the space of fields $F_M$ onto the space of fields $F_{\de M}$ on the boundary $\de M$. $F_{\de M}$ is given by restrictions of bulk fields and their normal jets to the boundary. We denote by $\text{EL}_M$ the $1$-form\footnote{$\text{EL}_M$ is the term that depends only on the variations of the fields but not on higher jets.} coming from the \textsf{Euler-Lagrange equations} (EL equations). The classical solutions are given by the critical points of $S_M$, i.e. by solutions of $\delta S_M=0$.
One can define a presymplectic form $\omega_{\de M}$ on $F_{\de M}$ by setting $\omega_{\de M}:=\delta\alpha_{\de M}$ (we think of $\delta$ as the de Rham differential on the space of fields). By techniques of symplectic geometry, such as \textsf{symplectic reduction}, one can obtain a symplectic manifold $(F^\de_{\de M},\omega^\de_{\de M})$.
Moreover, this construction is compatible with cutting and gluing \cite{CMR1,CMR3}. It leads to a nice quantum formulation in the guise of path integrals after choosing a suitable polarization \cite{CMR2}. We will discuss these issues in this section. 

\begin{rem}
Note that if $\de M=\varnothing$ we get the usual Euler--Lagrange equations from $\delta S_M=0$.
\end{rem}

\subsection{Finite dimensional BV theory}

Let $M$ be a closed manifold and let $F_M$ denote the space of fields associated to $M$. If we consider a regular\footnote{This means that the Hessian of the Lagrangian is weakly non degenerate.} local field theory $S_M\colon F_M\to\R$ the partition function in the path integral approach is
\begin{equation}
\label{state}
\psi_M=\int_{\phi\in F_M}\ee^{\frac{\I}{\hbar}S_M(\phi)}\mathscr{D}\phi.
\end{equation}
Usually, $F_M$ is infinite-dimensional, and one cannot define\footnote{Only in special situations, i.e. $\dim M = 1$, and some examples discussed in \cite{GJ}.} $\mathscr{D}\phi$. The way out is usually to translate the formal asymptotics as $\hbar \to 0$ of finite-dimensional integrals to the infinite-dimensional case. The terms in the asymptotic expansion are convenienetly labeled by Feynman diagrams \cite{Feynman1949,Feynman1950,P}. If the critical points of the action functional $S_M$ are degenerate, one needs to gauge-fix the theory before one can use the formal asymptotics. The most powerful gauge fixing formalism is the BV formalism. We briefly review its finite-dimensional version. Further references for gauge theories, different gauge fixing formalisms (including BV) and their perturbative quantization are
\cite{Mn,Mn2,R}. \\
The start is the following definition: 

\begin{defn}[BV manifold]
A \textsf{BV manifold} is a triple  $(\calF,\omega,\calS)$, where $\calF$ is a supermanifold with $\mathbb{Z}$-grading, $\omega$ an odd symplectic form of degree $-1$ on $\calF$, and $\calS$ is an even function of degree zero on $\calF$, such that 
\begin{equation}
(\calS,\calS) = 0.
\end{equation}
\end{defn}
Here, following Batalin and Vilkovisky \cite{BV1,BV2}, we denote the Poisson bracket induced by the odd symplectic form with round brackets $(\enspace,\enspace)$.  
\begin{rem}[Grading on $\calF$]
Note that we have two different gradings on $\calF$, the $\Z_2$-grading from the supermanifold structure and an additional $\Z$-grading. In physics, the $\Z$-grading is referred to as \textsf{ghost number} and the parity corresponds to bosonic and fermionic particles. Since we consider only bosonic theories, the $\Z_2$-grading coincides with the reduction of the $\Z$-grading.
\end{rem}
In a Darboux chart $(q^i,p_i)$, we can define the \textsf{BV Laplacian} by
\begin{equation}
\Delta^{\text{loc}}:=\sum_i(-1)^{\vert q^{i}\vert}\frac{\partial^2}{\partial q^{i}\partial p_i}.
\end{equation}
Then we get that $(\Delta^{\text{loc}})^2=0$ and for two functions $f,g$, $\Delta^{\text{loc}}(fg)=\Delta^{\text{loc}} fg\pm f\Delta^{\text{loc}} g\pm (f,g)$. 
This extends to a well-defined global operator $\gls{Delta}$ on half-densities \cite{Khudaverdian2004, Severa2006}. 

Moreover, given a half-density $f$ and a Lagrangian submanifold $\calL \subset \calF$, we can define a \textsf{BV integral} $
\int_{\mathcal{L}}f $
by restricting the half-density to the Lagrangian where it becomes a density and can be integrated. The main result in the Batalin--Vilkovisky formalism is the following theorem. 
\begin{thm}[Batalin--Vilkovisky \cite{BV1}]
\label{thm:BV}
If we assume that the integrals converge, then
\begin{itemize} 
\item{If $f=\Delta g$, then $\int_{\mathcal{L}}f=0$,
}
\item{If $\Delta f=0$ and $(\calL_t)$ is a smoothly varying family of Lagrangians, then $\frac{\dd}{\dd t}\int_{\mathcal{L}_t}f = 0$.
}
\end{itemize}
\end{thm}

\begin{rem}
The second point of Theorem \ref{thm:BV} tells us that if we have an ill-defined integral $\int_{\calL_0}f$ for some Lagrangian submanifold $\calL_0$, but we know that $\Delta f=0$, then we can define the value of the integral by a well-defined one $\int_{\calL_1}f$ for some Lagrangian submanifold $\calL_1$, and this does not depend on the choice of $\calL_1$ as long as we deform it smoothly.
\end{rem}

The replacement of $\calL_0$ by $\calL_1$ is called \textsf{gauge-fixing}. This construction can be extended to any (super)manifold. Moreover, considering $f=\ee^{\frac{\I}{\hbar}S}$, two other conditions arise, which are the Master Equations for the classical and quantum level:
\begin{align}
(S,S)&=0, \\
(S,S)-2\I\hbar\Delta S&=0.
\end{align}
The latter one is equivalent to $\Delta \ee^{\frac{\I}{\hbar}S} = 0$. The former one is the classical limit of the latter one for $\hbar\to 0$, and motivates the definition of BV manifold as given above. 

\subsection{Classical BV-BFV formalism}
\label{classicalBV-BFV}
We now turn to the infinite-dimensional case and review the main definitions of reference \cite{CMR1}. 
We first consider the classical BV formalism in field theory and its extension to manifolds with boundary.
\begin{defn}[BV theory]
A $d$-dimensional \textsf{BV theory} is the association of a BV manifold  $M\mapsto (\gls{calF_M},\gls{omega_M},\gls{calS_M})$ to every closed $d$-manifold $M$. 
\end{defn}
\begin{rem}
These BV manifolds are typically infinite-dimensional. This means that neither the BV Laplacian nor the BV integral are defined (at least not without further work). 
\end{rem}
\begin{defn}[BV extension]
We say that a BV theory $(\calF_M,\omega_M,\calS_M)$ is a \textsf{BV extension} of a local field theory $M \mapsto (F_M,S_M)$ if for all closed $d$-manifolds $M$, we have that the degree 0 part $(\calF_M)_0$ of $\calF_M$ satisfies $(\calF_M)_0 = F_M$ and $\calS_M\big|_{(\calF_M)_0} = S_M$. Moreover, we want $\calF_M$,$\calS_M$ and $\omega_M$ to be local.
\end{defn}
To extend the BV formalism to manifolds with boundary one needs its Hamiltonian counterpart, the BFV formalism \cite{BF1,BF2, FV1, FV2}.

\begin{defn}[BFV manifold]
A \textsf{BFV manifold} is a triple \begin{equation}\calF^\partial:=(\mathcal{F}^\partial,\gls{omega^de},Q^\partial) \end{equation} where $\mathcal{F}^\partial$ is a graded manifold, $\omega^\de$ an even symplectic form of degree $0$, and $Q^\partial$ a degree $1$ cohomological, symplectic vector field on $\calF^\partial$. If $\omega^\partial=\delta\alpha^\partial$ is exact, 
the BFV manifold is called \textsf{exact}.
\end{defn}
\begin{rem}
For degree reasons, $Q^\de$ automatically has a Hamiltonian function that we denote by $\calS^\de$, and call it the \textsf{boundary action}. This is the reason that the boundary action $\calS^\de$ is not included in the data of a BFV manifold. 
\end{rem}
%
 Again, we denote by $\delta$ the de Rham differential on the space of fields. 
The notion of BV theory can be extended to manifolds with boundary as was shown in \cite{CMR1,CMR2}. On the boundary we will use the BFV formalism. The compatibility between the BV formalism and the BFV formalism is captured in the following definition. \begin{defn}[BV-BFV manifold]
 A \textsf{BV-BFV manifold} over a given exact BFV manifold $\calF^\partial=(\mathcal{F}^\partial,\omega^\partial=\delta \alpha^\partial,Q^\partial)$ is a quintuple \begin{equation}
 \calF := (\calF,\omega,\calS,Q,\pi),\end{equation} where \begin{itemize} 
 \item  $\calF$ is a graded manifold, 
 \item $\omega$ is an even symplectic form of degree $0$, 
 \item $\calS$ is an even function of degree $0$, 
 \item $Q$ is a degree $1$ cohomological vector field, 
 \item  $\pi\colon \calF \to \calF^\partial$ is a surjective submersion
 \end{itemize} such that \begin{equation}
 \label{mCME0}
 \iota_{Q}\omega=\delta \calS+\pi^*\alpha^\partial
 \end{equation}and $Q^\partial=\delta\pi Q$ where $\delta\pi$ denotes the differential of $\pi$. \end{defn}
 \begin{rem} 
\label{point} 
If $\calF^\partial$ is a point, we get that $(\calF_M,\omega_M,\calS_M)$ is a BV manifold. The shorthand notation for a BV-BFV manifold is $\pi\colon\calF \to \calF^\de$.
\end{rem}

Note that by Remark \ref{point}, the following notion generalizes the one of a BV theory.
\begin{defn}[BV-BFV theory]
\label{defn:BVBFV_theory}
A $d$-dimensional \textsf{BV-BFV theory} associates 
\begin{itemize}
\item to every closed $(d-1)$-dimensional manifold $\Sigma$ a BFV manifold $\calF^\partial_\Sigma$,
\item to a $d$-dimensional manifold $M$ with boundary $\de M$ a BV-BFV manifold $\pi_M\colon \gls{calF^de_M} \to \calF^\de_{\de M}$.
\end{itemize}
\end{defn}

\begin{rem}
Formally, for the Hamiltonian vector field $Q$ of $\calS$, one can write $(\calS,\calS)=\iota_Q\iota_Q\omega=Q(\calS)$. If we consider a BV-BFV theory for a manifold $M$ with boundary $\de M$, one can prove \cite[Proposition 3.1]{CMR1} that 
\begin{equation}
Q(\calS)=\pi^*(2\gls{calS^de}\iota_{Q^\de}\alpha^\de).
\end{equation}
Together with \eqref{mCME0} this implies that
\begin{equation}
\label{mCME}
\iota_Q\iota_Q\omega=2\pi^*\calS^\de.
\end{equation}
We call \eqref{mCME} the \textsf{modified Classical Master Equation}.

\end{rem}

It was shown in \cite{CMR1} that \textsf{abelian $BF$ theory} is an example of a BV-BFV theory. 
\begin{ex}[Abelian $BF$ theory]
\textsf{Abelian $BF$ theory} is given by the following data: 
To a $d-1$-dimensional manifold $\Sigma$, we associate the BFV manifold 
\begin{align}
\calF^\de_\Sigma &= \Omega^\bullet(\Sigma)[1]\oplus \Omega^\bullet(\Sigma)[d-2]\ni \mathsf{X}\oplus \boldsymbol{\eta}\\
\omega^\de_\Sigma&=\int_\Sigma\delta\mathsf{X}\land \delta\boldsymbol{\eta}\\
Q^\de_\Sigma&=(-1)^d\int_\Sigma\left(\dd\boldsymbol{\eta}\land\frac{\delta}{\delta\boldsymbol{\eta}}+\dd\mathsf{X}\land\frac{\delta}{\delta\mathsf{X}}\right)
\end{align}
The Hamiltonian function of $Q^\de_\Sigma$ is given by $\calS^\de = \int_M \boldsymbol{\eta}\land\dd\mathsf{X}$. \\
To a $d$-dimensional manifold $M$ with boundary $\de M$, we associate the BV-BFV manifold over $\calF^\de_{\de M}$ given by 
\begin{align}
\calF_M &= \Omega^\bullet(M)[1]\oplus \Omega^\bullet(M)[d-2]\ni \mathsf{X}\oplus \boldsymbol{\eta}\\
\omega_M&=\int_M\delta\mathsf{X}\land \delta\boldsymbol{\eta}\\
\calS_M&=\int_M \boldsymbol{\eta}\land\dd\mathsf{X}\\
Q_M&=(-1)^d\int_M\left(\dd\boldsymbol{\eta}\land\frac{\delta}{\delta\boldsymbol{\eta}}+\dd\mathsf{X}\land\frac{\delta}{\delta\mathsf{X}}\right) \\
\end{align}
and the map $\pi\colon \calF \to \calF^\de$ is given by restriction, i.e. $\pi = \iota_{\de M}^*$, where $\iota_{\de M}\colon \de M \to M$ is the inclusion. 
\end{ex}

\begin{defn}[$BF$-like theories]
We say that a BV-BFV theory is \textsf{$BF$-like} if
\begin{align}
\calF_M &= (\Omega^\bullet(M)\otimes V[1])\oplus (\Omega^\bullet (M)\otimes V^*[d-2])\\
\calS_M &= \int_M\left(\langle\boldsymbol{\eta},\dd\mathsf{X}\rangle+\calV(\mathsf{X},\boldsymbol{\eta})\right),
\end{align}
where $V$ is a graded vector space, $\langle\enspace,\enspace\rangle$ denotes the pairing between $V^*$ and $V$, and $\calV$ denotes some density-valued function of the fields $\mathsf{X}$ and $\boldsymbol{\eta}$ whose value $\calV(x)$ at $x \in M$ depends only on $\mathsf{X}(x),\boldsymbol{\eta}(x)$\footnote{In particular, $\calV$ does not depend on derivatives of the fields.},
such that $\calS_M$ satisfies the Classical Master Equation for $M$ without boundary.
\end{defn}

\begin{ex}[Quantum mechanics]
Consider $M$ to be a $1$-dimensional manifold, i.e. $d=1$ and $V=W[-1]$ with $W$ concentrated in degree zero. Denote by $P$ and $Q$ the degree-zero form components of $\mathsf{X}$ and $\boldsymbol{\eta}$, respectively. Choose a volume form $\dd t$ on $M$ and a function $H$ on $T^*W$. Set $\calV(\mathsf{X},\boldsymbol{\eta}):=H(\mathsf{X},\boldsymbol{\eta})\dd t=H(Q,P)\dd t$. Then 
\begin{equation}
\calS_M=\int_M\left(\sum_i P_i\dot{Q}^{i}+H(Q,P)\right)\dd t,
\end{equation}
is the action of classical 
mechanics in the Hamiltonian formalism.
\end{ex}

\begin{rem}
The Poisson Sigma Model, which is the main theory regarded in this paper, is an
example of a $BF$-like AKSZ theory (see Section \ref{sec:AKSZ}).
\end{rem}

\begin{ex}[$BF$-like AKSZ theories \cite{AKSZ}]
Assume we are given a function $\Theta$ on $T^*[d-1](V[1])=V[1]\oplus V^*[d-2]$ that is of degree $d$ such that $\{\Theta,\Theta\}=0$, where $\{\enspace,\enspace\}$ is the canonical Poisson structure on the shifted cotangent bundle. Set $\calV(\mathsf{X},\boldsymbol{\eta})$ to be the top degree part of $\Theta(\mathsf{X},\boldsymbol{\eta})$. 
\end{ex}


\subsection{Quantum BV-BFV formalism}
In \cite{CMR2} the notion of a quantum BV-BFV theory was given and it was shown how to perturbatively quantize a classical BV-BFV theory\footnote{We have to assume certain conditions which are in particular satisfied for $BF$-like theories.}. Let us briefly review this\footnote{We slightly changed the definition of quantum BV-BFV theory so that in principle it does not depend on a classical BV-BFV theory.}. 
\begin{defn}[Quantum BV-BFV theory]
A $d$-dimensional \textsf{quantum BV-BFV theory} associates 
\begin{itemize}
\item To every closed $(d-1)$-dimensional manifold $\Sigma$ a graded $\C[[\hbar]]$-module $\mathcal{H}_{\Sigma}$, 
\item To every $d$-dimensional manifold (possibly with boundary) $M$ a finite-dimensional BV manifold $\mathcal{V}_M$, a degree 1 coboundary operator $\gls{Omega_deM}$ on $\gls{calH_deM}$  and a homogeneous element\footnote{Typically, $\psi$ will have degree 0. This is the case when the gauge-fixing Lagrangian (see below) has degree zero, in the sense that its Berezinian bundle has degree zero. This is the case in all examples we consider.} 
\begin{equation}
\gls{psi_M} \in \gls{HatcalH_M} := \mathrm{Dens}^{\frac12}(\gls{V_M}) \otimes \calH_{\de M},
\end{equation}
where $\gls{Dens}$ denotes the space of half-densities on some manifold $M$,
\end{itemize}
such that 
\begin{equation}(\hbar^2\Delta_{\mathcal{V}_M} + \Omega_{\de M})\psi_M = 0.
\label{mQME}
\end{equation}
\end{defn}

\begin{rem}
The shorthand notation for a quantum BV-BFV theory is 
\begin{equation}
M \mapsto (\widehat{\mathcal{H}}_M,\psi_M,\Delta_{\mathcal{V}_M},\Omega_{\de M}).
\end{equation}
\end{rem}
Let us introduce some terminology: We call $\mathcal{V}_M$ the \textsf{space of residual fields}, $\mathcal{H}_{\de M}$ the \textsf{space of boundary states} and $\psi_M$ the \textsf{quantum state}. $\Delta_{\mathcal{V}_M}$ denotes the canonical BV Laplacian on half-densities on the BV manifold $\mathcal{V}_M$. Recall that $\Delta_{\mathcal{V}_M}^2 =0$.
Hence, $\widehat{\mathcal{H}}_M$ carries the two commuting differentials $\Delta_{\mathcal{V}_M}$ and $\Omega_{\de M}$ which gives it the structure of a bicomplex.  We call $\Omega_{\de M}$ the quantum BFV boundary operator.
The condition \eqref{mQME} is called the \textsf{modified Quantum Master Equation}.
\begin{rem}[Terminology]
The space $\calH_\Sigma$ is called the space of states because it arises as a quantization of the symplectic manifold of boundary fields (see also the discussion in \ref{sec:space_of_states} below). An element of $\calH_\Sigma$ is then called a state. In the absence of residual fields, $\psi_M$ is the state produced by the bulk. It is what is usually called a state in the literature, see e.g. \cite{Witten1989} if $\de M$ has a single connected component\footnote{Note that this is a particular state induced by the bulk and not some choice of vacuum state.}.
In case $M$ is a cylinder, $\psi_M$ is actually an evolution operator that can be viewed as a generalized state (note that we never insist on $\calH_\Sigma$ being a Hilbert space). If the boundary is empty (and there are no residual fields), then $\psi_M$ is what is usually called the partition function. It is in general useful (and often necessary) to make a choice of ``slow'' or ``low energy'' fields, which we prefer to call residual fields, and to integrate on a complement. Then $\psi_M$ will be properly a state only after integrating out the residual fields (which is not always possible, cf. the discussions in \cite[Appendix F]{CMR1}, \cite{BCM}, \cite{IM}), but by abuse of notation we prefer to call it the state anyway.  \end{rem}
\begin{defn}[Equivalence]
We say that two quantum BV-BFV theories $(\widehat{\mathcal{H}}_M,\Delta_{\mathcal{V}_M},\Omega_{\de M},\psi_M)$ and $(\widehat{\mathcal{H}}_M',\Delta_{\mathcal{V}'_M},\Omega'_{\de M},\psi'_M)$ are equivalent if  for every manifold $M$ with boundary $\partial M$ there is a quasi-isomorphism of bicomplexes \begin{equation}I_M\colon (\widehat{\mathcal{H}}_M,\Delta_{\mathcal{V}_M},\Omega_{\de M}) \to  (\widehat{\mathcal{H}}_M',\Delta_{\mathcal{V}'_M},\Omega'_{\de M})\end{equation}
such that $I_M(\psi_M)= \psi'_M$.
\end{defn}
\begin{defn}[Change of data]
\label{change_of_data}
Another equivalence relation among theories is the following: We say that two quantum BV-BFV theories $(\widehat{\mathcal{H}}_M,\Delta_{\mathcal{V}_M},\Omega_{\de M},\psi_M)$ and $(\widehat{\mathcal{H}}_M,\Delta_{\mathcal{V}'_M},\Omega'_{\de M},\psi'_M)$
are related by \textsf{change of data} if there is an operator $\tau$ of degree 0 on $\mathcal{H}_{\de M}$ and an element $\chi_M \in \widehat{\mathcal{H}}_M$ with $\deg(\chi_M)=\deg(\psi_M)-1$ such that
\begin{align}
\begin{split}
\Omega'_{\de M} &= [\Omega_{\de M},\tau] \\
\psi'_M &=  (\hbar^2\Delta_{\calV_M} + \Omega_{\de M})\chi_M - \tau \psi_M
\end{split}
\end{align} 
\end{defn}
Let us now explain how to produce a quantum BV-BFV theory by perturbative quantization of a classical BV-BFV theory. Fix a classical BV-BFV theory $\pi\colon \calF \to \calF^\de$. For simplicity we shall assume that $\calF$ and $\calF^\de$ are always \textsf{vector spaces}, which is sufficient for the present paper. For a general discussion see \cite{CMR2}.
\subsubsection{The space of states}\label{sec:space_of_states}
Consider a $(d-1)$-dimensional manifold $\Sigma$. Then the BV-BFV theory associates to it a symplectic vector space $(\calF^\de_\Sigma,\omega^\de_\Sigma,Q_\Sigma^\de)$.   Morally, we want to construct $\mathcal{H}_{\Sigma}$ and $\Omega_{\Sigma}$ as a geometric quantization of this symplectic vector space. More precisely, the construction proceeds as follows.
We require the data of a polarization\footnote{We have only considered the case of real polarizations so far.} $\calP$ of this symplectic vector space. For our purposes, a splitting 
\begin{equation}
\calF^\de_\Sigma \cong \mathcal{B}^\calP_{\Sigma} \oplus \mathcal{K}^\calP_{\Sigma}
\end{equation}
of $\calF^\de_\Sigma$ into Lagrangian subspaces is sufficient. Here $\mathcal{K}^\calP_{\Sigma}$ is thought of as the Lagrangian distribution on $\calF^\de_\Sigma$ and $\mathcal{B}^\calP_\Sigma$ is identified with the leaf space of the polarization. Given a polarization $\gls{calP}$ the associated space of states $\calH_{\de M}$ is a certain space of functionals  on $\mathcal{B}^\calP_\Sigma$. We will discuss the space of states for $BF$-like theories in \ref{sec:BF_like}.

\subsubsection{Splitting the space of fields} To define the quantum state we proceed with the following constructions. Consider a $d$-manifold $M$ (possibly with boundary) and the associated BV-BFV manifold $(\calF_M,\omega_M,\calS_M,Q_M,\pi_M)$ over the exact BFV manifold $(\calF^\de_{\de M},\omega^\de_{\de M}=\delta\alpha^\de_{\de M},Q^\de_{\de M})$. Then, choosing a polarization $\calP$ on $\de M$, we choose a splitting
\begin{equation}
\calF_M\cong\calB_{\de M}^\calP\oplus\calY,
\end{equation}
where $\calY$ denotes some complement. This splitting is subject to the following assumption\footnote{This assumption forces one to choose singular extensions of boundary fields.}. 
\begin{ass}[\cite{CMR2}]\label{ass:singularsplitting}
There is a weakly symplectic form $\omega_\calY$ on $\calY$ such that $\omega_M$ is the extension of $\omega_\calY$ to $\calF_M$. 
\end{ass} Formally, we can think of $\gls{B}$ as the space of \textsf{boundary fields} and $\calY$ the space of \textsf{bulk fields}. Depending on the boundary polarization, we split $\calY$ into residual fields and some complement, i.e. we choose a splitting 
\begin{equation} 
\calY\cong\calV^\calP_M\oplus \calY'
\end{equation}
subject to the following assumption\footnote{This assumption is rather strong but can be slightly relaxed to the notion of \textsf{hedgehog fibration}.} 
\begin{ass}\label{ass:symplecticproduct}
We assume the following hold:
\begin{enumerate}
\item $\calV^\calP_M, \calY'$ are BV manifolds, 
\item $\calV^\calP_M$ is finite-dimensional
\item $\omega_\calY = \omega_{\calV^\calP_M} + \omega_{\calY'}$.
\end{enumerate}
\end{ass}
We call the complement $\calY'$ the space of \textsf{fluctuation fields}. 
Residual fields are also called \textsf{low energy fields} or \textsf{slow fields} and fluctuation fields are also called \textsf{high energy fields} or \textsf{fast fields}.  Typically we choose $\calV^\calP_M$ as the solutions of $\delta \calS_M^{0}=0$ modulo gauge transformations, where $\calS_M^0$ denotes the quadratic part of the action $\calS_M$. This is the minimal choice, and is typically called the space of \textsf{zero modes}. Other choices are related by the equivalence relations above. 
\begin{defn}
A splitting 
\begin{equation}
\label{split}
\calF_M \cong \calB_{\de M}^\calP \oplus \calV^\calP_M\oplus \calY'
\end{equation}
is called \textsf{good} if it satisfies Assumptions \ref{ass:singularsplitting} and \ref{ass:symplecticproduct}.
\end{defn}
\begin{rem}[Connection to Atiyah's TQFT formulation]
From the point of view of topological quantum field theories (TQFTs) as functors $\textbf{Cob}_n\to \textbf{Vect}_\C$ from the $n$-cobordism category (objects are $(n-1)$-manifolds bounding an $n$-manifold and morphisms are exactly the bounding $n$-manifolds connecting the objects) to the category of vector spaces over the complex numbers, it is clear that the quantum state should depend on the bulk. This can be seen by using the fact that the state represents exactly the bounding manifold between the objects and thus a morphism of the cobordism category. This also makes sense for manifolds without boundary, in which case the state is given by a partition function $Z\colon \C\to\C$, where as a morphism in $\textbf{Cob}_n$ it represents any closed $n$-manifold, seen as a bounding manifold connecting the empty $(n-1)$-manifold, i.e. as a morphism $\varnothing\to\varnothing$.
\end{rem}

\subsubsection{The quantum state in $BF$-like theories}
\label{sec:BF_like}
The quantum state in $BF$-like theories is defined perturbatively in terms of Feynman graphs by considering integrals defined on the configuration space of these graphs. In $BF$-like theories there are two preferred polarizations, namely the $\frac{\delta}{\delta \mathsf{X}}$- and $\frac{\delta}{\delta\boldsymbol{\eta}}$-polarization. We specify a polarization by splitting the boundary $\de M$ of the manifold $M$ into two parts $\de_1M$ and $\de_2M$, where we choose the $\frac{\delta}{\delta\boldsymbol{\eta}}$-polarization on $\de_1M$ and the $\frac{\delta}{\delta\mathsf{X}}$-polarization on $\de_2M$. We denote the $\mathsf{X}$-leaf by $\mathbb{X}\in\calB^{\frac{\delta}{\delta\boldsymbol{\eta}}}_{\de M}$ and the $\boldsymbol{\eta}$-leaf by $\E\in \calB^\frac{\delta}{\delta\mathsf{X}}_{\de M}$.

For $BF$-like theories, the polarization determines the first splitting as 
\begin{align}
\calB^\calP_{\de M} &= (\Omega^\bullet(\de_1M)\otimes V[1]) \oplus (\Omega^\bullet(\de_2M) \otimes V^*[d-2]) \\
\calY &= (\Omega^{\bullet}(M,\de_1M)\otimes V[1]) \oplus (\Omega^\bullet(M,\de_2M)\otimes V^*[d-2])
\end{align} 
The minimal space of residual fields is isomorphic to
\begin{equation}
\calV^\calP_M\cong(H^\bullet(M,\de_1M)\otimes V[1] )\oplus (H^\bullet(M,\de_2M)\otimes V^*[d-2]),
\end{equation}
for some graded vector space $V$. A good splitting is then determined by a splitting of the complex of de Rham forms with relative boundary conditions into a subspace $\calV^\calP_M$ isomorphic to cohomology and a complement $\calY'$ in a way compatible with the symplectic structure. One possibility to do so is to use a Riemannian metric and embed the cohomology as harmonic forms. \\
Before we can introduce the quantum state we need to introduce 
the concept of \textsf{composite fields}, which we denote by square brackets $[\enspace]$, e.g. for a boundary field $\mathbb{A}$ we will write $[\mathbb{A}^{i_1}\dotsm \mathbb{A}^{i_k}]$. They can be understood as a \textsf{regularization} of higher functional derivatives: the higher functional derivative $\frac{\delta^k}{\delta\mathbb{A}^{i_1}\dotsm \delta\mathbb{A}^{i_k}}$ gets replaced by a first order functional derivative $\frac{\delta}{\delta[\mathbb{A}^{i_1}\dotsm \mathbb{A}^{i_k}]}$. Concretely, this corresponds to introducing additional boundary vertices as in Figure \ref{fig:composite_field_vertices}. 

\begin{rem}
In fact, this concept will not be needed for the definition of the principal part of the quantum state.
We will use this concept to define the full part of the quantum state where we need to make sure that it will be compatible with the quantum BFV boundary operator, where higher functional derivatives do indeed appear as we will see.
\end{rem}

\begin{defn}[Regular functional]
\label{regular_fun}
A \textsf{regular functional} on the space of base boundary fields is a linear combination of expressions of the form 
\begin{multline}
\label{regular_functional}
\int_{\mathsf{C}_{m_1}(\de_1M)\times \mathsf{C}_{m_2}(\de_2M)}L^{J_1^1...J_1^{\ell_1}J_2...J_2^{\ell_2}...}_{I_1^1....I_1^{r_1}I_2^1...I_2^{r_2}...}\land \pi_1^*\prod_{j=1}^{r_1}\left[\mathbb{X}^{I_1^j}\right]\land\dotsm \land \pi_{m_1}^*\prod_{j=1}^{r_{m_1}}\left[\mathbb{X}^{I_{m_1}^j}\right]\land\\
\land\pi_1^*\prod_{j=1}^{\ell_1}\left[\mathbb{E}_{J_1^j}\right]\land\dotsm\land \pi_{m_1}^*\prod_{j=1}^{\ell_{m_2}}\left[\mathbb{E}_{J_{m_2}^j}\right],
\end{multline}
where $I_i^j$ and $J_i^j$ are (target) multi-indices and $L^{J_1^1...J_1^{\ell_1}J_2...J_2^{\ell_2}...}_{I_1^1....I_1^{r_1}I_2^1...I_2^{r_2}...}$ is a smooth differential form on the direct product of compactified configuration spaces (see Appendix \ref{app:Conf}) $\mathsf{C}_{m_1}(\de_1M)\times \mathsf{C}_{m_2}(\de_2M)$ depending on residual fields. A regular functional is called \textsf{principal} if all multi-indices have length one.
\end{defn}

\begin{defn}[Full space of boundary states]
\label{space_of_boundary_states}
The \textsf{full space of boundary states} $\gls{calH^calP_deM}$ is given by the linear combinations of regular functionals of the form \eqref{regular_functional}.
\end{defn}

\begin{defn}[Principal space of boundary states]
We define the \textsf{principal space of boundary states} $\gls{calH^calP(princ)_deM}$ as the subspace of $\calH^\calP_{\de M}$, where we only consider principal regular functionals.
\end{defn}
The state is defined in terms of Feynman graphs and rules. We briefly explain what these terms mean in the BV-BFV context (for perturbations of abelian $BF$ theory).
\begin{defn}[($BF$) Feynman graph]
A ($BF$) \textsf{Feynman graph} is an oriented graph with three types of vertices $V(\Gamma) = V_{\textnormal{bulk}}(\Gamma) \sqcup V_{\de_1} \sqcup V_{\de_2}$, called bulk vertices and type 1 and 2 boundary vertices, such that 
\begin{itemize}
\item bulk vertices can have any valence, 
\item type 1 boundary vertices carry any number of incoming half-edges (and no outgoing half-edges), 
\item type 2 boundary vertices carry any number of outgoing half-edges (and no incoming half-edges),
\item multiple edges and loose half-edges (leaves) are allowed but not short loops (tadpoles).
\end{itemize}
A \textsf{labeling} of a Feynman graph is a function from the set of half-edges to $\{1,\ldots,\dim V\}$.
\end{defn}
\begin{defn}[Principal graph]
A Feynman graph is called \textsf{principal} if all boundary vertices (type 1 and type 2) are univalent or zero valent.  
\end{defn}
For a set $S$ and a manifold $M$, the open configuration space of $S$ in $M$ is 
\begin{equation}
\mathsf{Conf}_S(M) := \{\iota\colon S\hookrightarrow M|\iota \text{ injection}\}.
\end{equation}
Let $\Gamma$ be a Feynman graph and $M$ a manifold with boundary $\de M = \de_1 M \sqcup \de_2M$ and denote 
\begin{equation}
\gls{Conf}:= \mathsf{Conf}_{V_{bulk}}(M) \times \mathsf{Conf}_{V_{\de_1}}(\de_1M) \times \mathsf{Conf}_{V_{\de_2}}(\de_2M)
\end{equation}
The Feynman rules are a map that associate to a Feynman graph $\Gamma$ a differential form $\omega_{\Gamma} \in \Omega^\bullet(\mathsf{Conf}_\Gamma(M))$. 
\begin{defn}[($BF$) Feynman rules]
Let $\Gamma$ be a labeled Feynman graph, and choose a configuration $\iota\colon V(\Gamma) \to \mathsf{Conf}(\Gamma)$ (that respects the decompositions). We decorate the graph according to the following rules (called \textsf{Feynman rules}):
\begin{itemize}
\item{Bulk vertices in $M$ decorated by ``vertex tensors'' \begin{equation}
\calV^{j_1\ldots j_t}_{i_1 \ldots i_s}:=\frac{\de^{s+t}}{\de \mathsf{X}_{i_1}\dotsm \de \mathsf{X}_{i_s}\de\boldsymbol{\eta}^{j_1}\dotsm \de\boldsymbol{\eta}^{j_t}}\big\vert_{\mathsf{X}=\boldsymbol{\eta}=0}\calV(\mathsf{X},\boldsymbol{\eta}),
\end{equation}
where $s,t$ are the out- and in- valencies of the vertex and $i_1,\ldots,i_s$ and $j_1,\ldots,j_t$ are the labels of the out (resp. in-)oriented half-edges.
}
\item{ Boundary vertices $v \in V_{\de_1}(\Gamma)$ with  incoming half-edges labeled $i_1,\ldots,i_k$ and no out-going half-edges are decorated by a composite field $[\mathbb{X}_{i_1}\ldots\mathbb{X}_{i_k}]$ evaluated at the point (vertex location) $\iota(v)$ on $\partial_1M$.
}
\item{Boundary vertices $v \in V_{\de_2}(\Gamma)$ on $\partial_2M$ with  outgoing half-edges labeled $j_1,\ldots,j_l$ are decorated by $[\E^{j_1}\ldots\E^{j_l}]$ evaluated at the point on $\de_2M$.
}
\item{Edges between vertices $v_1,v_2$ are decorated with the propagator $\zeta(\iota(v_1),\iota(v_2))\cdot \delta_j^{i}$, where $\zeta$ is the propagator induced by $\calL\subset\calY'$, the chosen gauge-fixing Lagrangian.
}
\item{Loose half-edges (leaves) attached to a vertex $v$ and labeled $i$ are decorated with the residual fields $\mathsf{x}_i$ (for out-orientation), $\mathsf{e}^{i}$ (for in-orientation) evaluated at the point $\iota(v)$. 
}
\end{itemize}
We denote the differential forms given by the decorations collectively by $\omega_d$. The differential form $\omega_{\Gamma}$ at $\iota$ is then defined by multiplying all decorations and summing over all labelings: 
\begin{equation}
\label{eq:omega_gamma}
\omega_{\Gamma} = \sum_{\substack{\text{labelings}\\ \text{of }\Gamma}}\hspace{0.1cm}\prod_{\substack{\text{decorations} \\\text{$d$ of } \Gamma}} \omega_d
\end{equation}
\end{defn}
The Feynman rules are summarized in Figures \ref{fig:FeynmanRules0} and \ref{fig:composite_field_vertices}.
\begin{rem}[Configuration spaces]
We will work with the Fulton--MacPherson/Axelrod--Singer compactification of configuration spaces on manifolds with boundary and corners (FMAS compactification, see Appendix \ref{app:Conf}). It is a non-trivial analytic statement (proven first by Axelrod and Singer \cite{AS2}) that the propagator, \textsf{a priori} defined only on the open configuration space $\mathsf{Conf}_2(M)$, extends to the  compactification $\mathsf{C}_2(M)$. It follows that also $\omega_\Gamma$, for all Feynman graphs $\Gamma$, extends to the compactification $\gls{C}$ of $\mathsf{Conf}_\Gamma(M)$. Since integrals remain unchanged by adding strata of lower codimension, this immediately proves that all integrals in Equation \eqref{eq:def_state_pert} below are finite. Moreover, the combinatorics of the stratification can be used for various computations using Stokes' theorem. 
\end{rem}
\begin{defn}[Principal quantum state]
\label{principal_part}
Let $M$ be a manifold, possibly with boundary. 
Given a $BF$-like BV-BFV theory $\pi_M\colon\calF_M \to \calF^\de_{\de M}$, a polarization $\calP$ on $\calF^\de_{\de M}$, a good splitting $\calF_M = \mathcal{B}^\calP_{\de M}\oplus \calV^\calP_M \oplus \calY'$, and a gauge-fixing Lagrangian $\calL \subset \calY'$, we define the \textsf{principal part of the quantum state} by the  formal power series 
\begin{equation}
\psi_M(\mathbb{X},\E;\mathsf{x},\mathsf{e}):=T_M\exp\left(\frac{\I}{\hbar}\sum_{\Gamma}\frac{(-\I\hbar)^{\textnormal{loops}(\Gamma)}}{\vert\textnormal{Aut}(\Gamma)\vert}\int_{\mathsf{C}_\Gamma(M)}\omega_\Gamma(\mathbb{X},\E;\mathsf{x},\mathsf{e})\right),\label{eq:def_state_pert}
\end{equation}
where $\omega_\Gamma$ is given as in \eqref{eq:omega_gamma} and where we denote for an element $\mathsf{X}\oplus \boldsymbol{\eta}\in \calF_M$ the split by 
\begin{align}
\mathsf{X}&= \mathbb{X}\oplus \mathsf{x}\oplus \mathscr{X},\\
\boldsymbol{\eta}&=\E\oplus \mathsf{e}\oplus \mathscr{E}.
\end{align}
Here the sum is taken over all \textsf{connected}, oriented, \textsf{principal} BF Feynman graphs $\Gamma$, $\textnormal{Aut}(\Gamma)$ denotes the set of all automorphisms of $\Gamma$, and $\textnormal{loops}(\Gamma)$ denotes the number of all loops of $\Gamma$. 
\end{defn}
  The coefficient $T_M$ is related to the Reidemeister torsion of $M$, but its precise nature is irrelevant for the purpose of a present paper. For a definition see \cite{CMR1}.
  \begin{rem} The formal power series \eqref{eq:def_state_pert} is our definition of the  
  formal perturbative expansion of the BV integral
\begin{equation}
\psi_M = \int_{\calL \subset \calY'}\ee^{\frac{\I}{\hbar}\calS_M[(\mathsf{X},\boldsymbol{\eta})]} \in \Hat{\calH}_M^\calP:=\Hat{\calH}_{\de M}^\calP\otimes \textnormal{Dens}^\frac{1}{2}(\calV_M^\calP).\end{equation}
It was observed in \cite{CMR2} that, given a good splitting of the form \eqref{split}, one can decompose the action as  
\begin{equation}
\calS^\calP_M := \widehat{\calS}_{M,0} + \widehat{\calS}_{M,\textnormal{pert}} + \calS^\textnormal{res} + \calS^{\textnormal{source}}
\end{equation} 
with 
\begin{align}
\widehat{\calS}_{M,0} &:= \int_M \langle\mathscr{E}, \dd \mathscr{X}\rangle\\
 \widehat{\calS}_{M,\textnormal{pert}} &:= \int_M \calV(\mathscr{X},\mathscr{E}) \\
 \calS^\textnormal{res} &:= (-1)^{d-1}\left(\int_{\de_1 M}\langle \E, \mathsf{x}\rangle +\int_{\de_2 M}\langle \mathbb{X}, \mathsf{e}\rangle \right) \\
 \calS^{\textnormal{source}} &:= (-1)^{d-1}\left(\int_{\de_1 M}\langle \E, \mathscr{X}\rangle +\int_{\de_2 M}\langle \mathbb{X}, \mathscr{E}\rangle \right)
\end{align}
In that way we can rewrite 
\begin{equation}
\psi_M = T_M\left \langle \ee^{\frac{\ii}{\hbar}(\calS^\textnormal{res} + \calS^{\textnormal{source}})}\right\rangle
\end{equation}
where $\langle \enspace \rangle$ denotes the expectation value with respect to the bulk theory ($\widehat{\calS}_{M,0} + \widehat{\calS}_{M,\textnormal{pert}}$), i.e. formally
\begin{equation}
\left\langle \ee^{\frac{\ii}{\hbar}(\calS^\textnormal{res} + \calS^{\textnormal{source}})}\right\rangle =  \int_{\calL \subset \calY'}\ee^{\frac{\I}{\hbar}(\widehat{\calS}_{M,0}[(\mathsf{X},\boldsymbol{\eta})]+\widehat{\calS}_{M,\textnormal{pert}}[(\mathsf{X},\boldsymbol{\eta})] )}\ee^{\frac{\I}{\hbar}(\calS^{\textnormal{res}}[(\mathsf{X},\boldsymbol{\eta})]+{\calS}^{\textnormal{source}}[(\mathsf{X},\boldsymbol{\eta})])}.
\end{equation}
\end{rem}

\begin{rem}
Note that we sum over connected graphs, such that the sum is given by the \textsf{effective action}. 
\end{rem}

%

\begin{figure}
\centering
\subfigure[Interaction vertex]{
\centering
\begin{tikzpicture}
\node[vertex] (o) at (0,0) {};
\node[coordinate, label=below:{$i_1$}] at (30:1) {$i_1$}
edge[fermion] (o);
\node[coordinate, label=above:{$i_2$}] at (60:1) {$i_2$}
edge[fermion] (o);
\node[coordinate, label=below:{$i_s$}] at (145:1) {$i_s$}
edge[fermion] (o);
\draw[dotted] (90:0.5) arc (90:130:0.5);
\node[coordinate, label=below:{$j_1$}] (j1) at (-30:1) {};
\node[coordinate, label=below:{$j_2$}] (j2) at (-60:1) {};
\node[coordinate, label=below:{$j_t$}] (j3) at (-145:1) {};
\draw[dotted] (-90:0.5) arc (-90:-145:0.5); 
\draw[fermion] (o) -- (j1);
\draw[fermion] (o) -- (j2);
\draw[fermion] (o) -- (j3);
\node[coordinate,label=right:{$\leadsto\quad\calV^{j_1\ldots j_t}_{i_1 \ldots i_s}$}] at (2,0) {};
\end{tikzpicture}
}
%
\hspace{2cm}
\subfigure[Residual fields]{
\centering
\begin{tikzpicture}
\node[circle,draw,inner sep=1pt] (x) at (0,1) {$\mathsf{x}^i$};
\node[coordinate, label=below:{$i$}] (x2) at (30:2) {$i_1$}
edge[fermion] (x);
\node[coordinate, label=below:{$j$}] (e2)at (-30:2) {$i_1$};
\node[circle,draw,inner sep=1pt] (e) at (0,-1) {$\mathsf{e}_j$}
edge[fermion] (e2);
\node[coordinate, label={${}$}] at (4,0) {};
\end{tikzpicture}
}
\hspace{3cm}
\subfigure[Boundary vertices]{
\begin{tikzpicture}
\draw (-1,1) -- (1,1);
\node[vertex, label=above:{$\mathbb{X}$}] (x) at (0,1) {};
\node[coordinate] (b1) at (0,0) {}; 
\draw[fermion] (x) -- (b1);
\draw (1,-1) -- (3,-1);
\node[vertex, label=below:{$\mathbb{E}$}] (e) at (2,-1) {};
\node[coordinate] (b2) at (2,0) {}; 
\draw[fermion] (b2) -- (e);

\end{tikzpicture}
\hspace{2cm}
}
\caption{Summary of Feynman graphs and rules}\label{fig:FeynmanRules0}
\end{figure}
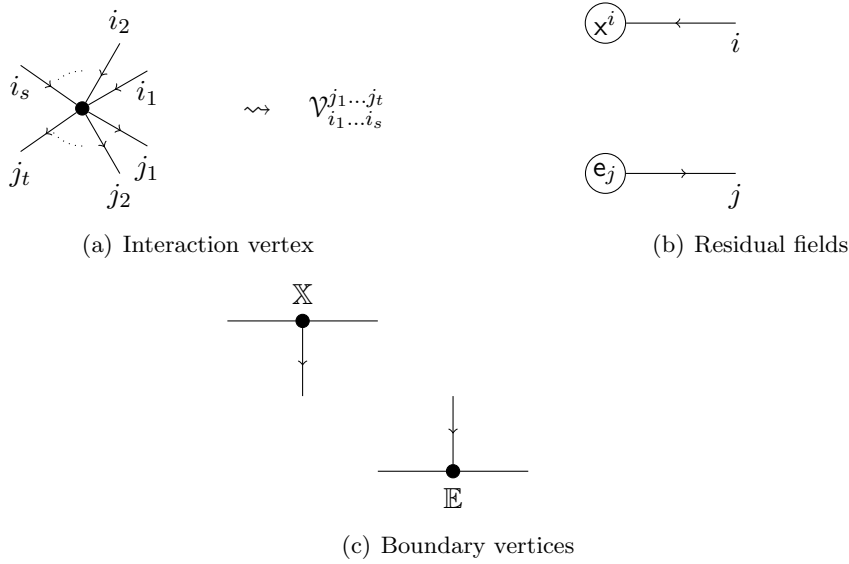

\begin{figure}[ht]
\subfigure[Boundary vertex on $\de_1\Sigma$] {
\begin{tikzpicture}
\draw (-3,0) -- (3,0); 
\node[vertex] (o) at (0,0) {};
\node[coordinate, label=below:{$[\mathbb{X}^{i_1}\cdots\mathbb{X}^{i_k}]$}] at (o.south) {};
\node[coordinate, label=below:{$i_1$}] (b1) at (30:1) {$i_1$};
\node[coordinate, label=above:{$i_2$}] (b2) at (60:1) {$i_2$};
\node[coordinate, label=below:{$i_k$}] (b3) at (145:1) {$i_k$};
\draw[dotted] (90:0.5) arc (90:130:0.5); 
\draw[fermion] (o) -- (b1);
\draw[fermion] (o) -- (b2);
\draw[fermion] (o) -- (b3);
\end{tikzpicture}
}
\hspace{2cm}
\subfigure[Boundary vertex on $\de_2\Sigma$] {
\begin{tikzpicture}
\draw (-3,0) -- (3,0); 
\node[vertex] (o) at (0,0) {};
\node[coordinate, label=below:{$[\mathbb{E}_{i_1}\cdots\mathbb{E}_{i_k}]$}] at (o.south) {};
\node[coordinate, label=below:{$i_1$}] at (30:1) {$i_1$}
edge[fermion] (o);
\node[coordinate, label=above:{$i_2$}] at (60:1) {$i_2$}
edge[fermion] (o);
\node[coordinate, label=below:{$i_k$}] at (145:1) {$i_k$}
edge[fermion] (o);
\draw[dotted] (90:0.5) arc (90:130:0.5);
\end{tikzpicture}
}
\caption{Composite field vertices.}\label{fig:composite_field_vertices}
\end{figure}
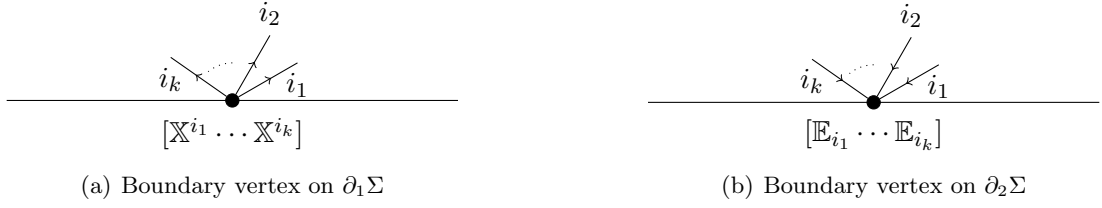 

Using composite fields, one can construct the \textsf{bullet product} on the full space of states as in \cite{CMR2}. For instance, the bullet product of $\int_{\de_1M}u_i\land \mathbb{X}^{i}$ and $\int_{\de_1M}v_i\land \mathbb{X}^{i}$ is 
\begin{multline}
\label{bullet_prod}
\int_{\de_1M}u_i\land \mathbb{X}^{i}\bullet \int_{\de_1M}v_j\land\mathbb{X}^j:=\\
(-1)^{\vert \mathbb{X}^{i}\vert (d-1+\vert v_j\vert)+\vert u_i\vert(d-1)}\left(\int_{\mathsf{C}_2(\de_1M)}\pi_1^*u_i\land \pi_2^*v_j\land \pi_1^*\mathbb{X}^{i}\land \pi_2^*\mathbb{X}^j+\int_{\de_1M}u_i\land v_j\land [\mathbb{X}^{i}\mathbb{X}^j]\right),
\end{multline}
where $u$ and $v$ are smooth differential forms depending on the bulk and residual fields.

\begin{rem}
Consider an operator $\int_{\de_1M}F^{ij}\frac{\delta^2}{\delta\mathbb{X}^{i}\delta\mathbb{X}^j}$. Such an operator is (by definition) interpreted as $\int_{\de_1M}F^{ij}\frac{\delta}{\delta[\mathbb{X}^{i}\mathbb{X}^j]}$, so one gets 
\begin{equation}
\int_{\de_1M}F^{ij}\frac{\delta^2}{\delta\mathbb{X}^{i}\delta\mathbb{X}^j}\left(\int_{\de_1M}u_{i}\land\mathbb{X}^{i}\bullet \int_{\de_1M}v_j\land \mathbb{X}^j\right)=\int_{\de_1M}u_iv_jF^{ij},
\end{equation}
in accordance with our naive expectation.
\end{rem}

\begin{defn}[Full quantum state]
\label{full_covariant_state1}
Let $M$ be a manifold, possibly with boundary. 
Given a $BF$-like BV-BFV theory $\pi_M\colon\calF_M \to \calF^\de_{\de M}$, a polarization $\calP$ on $\calF^\de_{\de M}$, a good splitting $\calF_M \cong \mathcal{B}^\calP_{\de M}\oplus \calV^\calP_M \oplus \calY'$, and a gauge-fixing Lagrangian $\calL \subset \calY'$, we define the \textsf{full quantum state} (similarly as in \eqref{eq:def_state_pert}) by the formal
power series 
\begin{equation}
\gls{fullstate}(\mathbb{X},\E;\mathsf{x},\mathsf{e}):=T_M\exp\left(\frac{\I}{\hbar}\sum_{\Gamma}\frac{(-\I\hbar)^{\textnormal{loops}(\Gamma)}}{\vert\textnormal{Aut}(\Gamma)\vert}\int_{\mathsf{C}_\Gamma(M)}\omega_\Gamma(\mathbb{X},\E;\mathsf{x},\mathsf{e})\right),\label{eq:def_state_pert_full}
\end{equation}
where we also sum over graphs as in Figure \ref{fig:composite_field_vertices} representing composite fields.
\end{defn}
\begin{rem}
The full state can be interpreted as an expectation value with help of the bullet product: 
\begin{equation}
\boldsymbol{\psi}_M = T_M \left\langle \ee_{\bullet}^{\frac{\ii}{\hbar}(\calS^\textnormal{res} + \calS^{\textnormal{source}})}\right\rangle
\end{equation}
where $\ee_\bullet$ denotes the exponential with respect to the bullet product. 
\end{rem}

\subsubsection{The BFV boundary operator}
We want to define the quantum BFV boundary operator for $BF$-like theories according to \cite{CMR2}. Similarly to the state, we will express at first its principal part and then extend it to a regularization using the notion of composite fields. The quantum BFV boundary operator is constructed as a quantization of the BFV action such that Theorem \ref{thm:CMR2} below holds.

\begin{defn}[Principal part of the BFV boundary operator]
The \textsf{principal part} of the BFV boundary operator is given by 
\begin{equation}
\gls{Omega^princ}=\underbrace{\gls{Omega^X_0}+\gls{Omega^E_0}}_{=:\Omega_0}+\underbrace{\gls{Omega^X_pert}+\gls{Omega^E_pert}}_{=:\Omega_{\textnormal{pert}}^{\textnormal{princ}}},
\end{equation}
where 
\begin{align}
\Omega_0^\mathbb{X}&:=(-1)^d\I\hbar\int_{\de_1M}\left(\dd \mathbb{X}\frac{\delta}{\delta\mathbb{X}}\right),\\
\Omega_0^\E &:=(-1)^d\I\hbar\int_{\de_2 M}\left(\dd\mathbb{E}\frac{\delta}{\delta\mathbb{E}}\right),
\end{align}
\begin{multline}
\Omega_{\textnormal{pert}}^\mathbb{X}:=\sum_{n,k\geq 0}\sum_{\Gamma_1'}\frac{(\I\hbar)^{\textnormal{loops}(\Gamma_1')}}{\vert \textnormal{Aut}(\Gamma'_1)\vert}\int_{\de_1M}\left(\sigma_{\Gamma_1'}\right)_{i_1....i_n}^{j_1...j_k}\land\\
\land \mathbb{X}^{i_1}\land \dotsm \land\mathbb{X}^{i_n} \left((-1)^d\I\hbar\frac{\delta}{\delta\mathbb{X}^{j_1}}\right)\dotsm \left((-1)^d\I\hbar\frac{\delta}{\delta\mathbb{X}^{j_k}}\right),
\end{multline}
\begin{multline}
\Omega_{\textnormal{pert}}^\mathbb{E}:=\!\sum_{n,k\geq 0}\sum_{\Gamma_2'}\frac{(\I\hbar)^{\textnormal{loops}(\Gamma_2')}}{\vert \textnormal{Aut}(\Gamma'_2)\vert}\int_{\de_2M}\left(\sigma_{\Gamma_2'}\right)_{i_1....i_n}^{j_1...j_k}\land\\ 
\land\mathbb{E}^{i_1}\land\dotsm\land \mathbb{E}^{i_n}\left((-1)^d\I\hbar\frac{\delta}{\delta\mathbb{E}^{j_1}}\right)\dotsm \left((-1)^d\I\hbar\frac{\delta}{\delta\mathbb{E}^{j_k}}\right),
\end{multline}
where, for $\mathbb{F}_1=\mathbb{X}$ and $\mathbb{F}_2=\E$ and $\ell\in\{1,2\}$, $\Gamma_\ell'$ runs over graphs with 
\begin{itemize}
\item{$n$ vertices on $\de_\ell M$ of valence 1 with adjacent half-edges oriented inwards and decorated with boundary fields $\mathbb{F}^{i_1}_\ell,\ldots,\mathbb{F}^{i_n}_\ell$ all evaluated at the point of collapse $p\in \de_\ell M$,}
\item{$k$ outward leaves if $\ell=1$ and $k$ inward leaves if $\ell=2$, decorated with variational derivatives in boundary fields
$$(-1)^d\I\hbar\frac{\delta}{\delta\mathbb{F}^{j_1}_\ell},\ldots,(-1)^d\I\hbar\frac{\delta}{\delta\mathbb{F}^{j_k}_\ell}$$
at the point of collapse,
}
\item{
no outward leaves if $\ell=2$ and no inward leaves if $\ell=1$  (graphs with them do not contribute).}
\end{itemize}
The form $\sigma_{\Gamma_\ell'}$ is obtained as the integral over the compactified configuration space $\Tilde{\mathsf{C}}_{\Gamma_\ell'}(\mathbb{H}^d)$, where $\mathbb{H}^d$ denotes the $d$-dimensional upper half plane, given by
\begin{equation}
\sigma_{\Gamma_\ell'}=\int_{\Tilde{\mathsf{C}}_{\Gamma_\ell'}(\mathbb{H}^d)}\omega_{\Gamma_\ell'},
\end{equation}
with $\omega_{\Gamma_\ell'}$ being the product of limiting propagators at the point $p$ of collapse and vertex tensors. 
\end{defn}

We want to roughly describe the construction of the BFV boundary operator with composite fields (see \cite{CMR2} for a more detailed discussion). First, we need to define the following notion.

On a regular functional as in \eqref{regular_functional}, we get a term $L$ replaced by $\dd L$ plus all the terms corresponding to  the boundary of the configuration space. As $L$ is smooth, its restriction to the boundary is also smooth and can be integrated on the fibers yielding a smooth form on the base configuration space; for example
\begin{equation}
\Omega_0\int_{\de_1M}L_{IJ}\land[\mathbb{X}^{I}]\land [\mathbb{X}^J]=\pm \I\hbar\int_{\de_1M}\dd L_{IJ}\land [\mathbb{X}^{I}]\land [\mathbb{X}^J],
\end{equation}
\begin{multline}
\Omega_0\int_{\mathsf{C}_2(\de_1M)}L_{IJK}\land \pi_1^*([\mathbb{X}^{I}]\land[\mathbb{X}^J])\land\pi_2^*[\mathbb{X}^K]\\=\pm\I\hbar \int_{\mathsf{C}_2(\de_1M)}\dd L_{IJK}\land \pi_1^*([\mathbb{X}^{I}]\land [\mathbb{X}^J])\land \pi_2^*[\mathbb{X}^K]\pm\I\hbar \int_{\de_1M}\underline{L_{IJK}}\land [\mathbb{X}^{I}]\land [\mathbb{X}^J]\land[\mathbb{X}^K],
\end{multline}
with $\underline{L_{IJK}}=\pi^\de_*L_{IJK}$, where $\pi^\de\colon \de\mathsf{C}_2(\de_1M)\to \de_1M$ is the canonical projection.


Notice that for any two regular functionals $S_1$ and $S_2$ we have $$\Omega_0(S_1\bullet S_2)=\Omega_0(S_1)\bullet S_2\pm  S_1\bullet \Omega_0(S_2).$$ The other generators that we allow are products of expressions of the form 
\begin{align}
\int_{\de_1M}&L^J_{I^1...I^r}\left[\mathbb{X}^{I_1}\right]\land\dotsm \land\left[\mathbb{X}^{I_{r}}\right]\frac{\delta^{\vert J\vert}}{\delta [\mathbb{X}^J]}\\
\int_{\de_2M}&L_I^{J^1...J^\ell}\left[\mathbb{E}_{J_1}\right]\land\dotsm \land\left[\mathbb{E}_{J_{\ell}}\right]\frac{\delta^{\vert I\vert}}{\delta [\mathbb{E}_I]}.
\end{align}

\begin{defn}[Full BFV boundary operator]
\label{full_BFV}
The \textsf{full BFV boundary operator} is given by 
\begin{equation}
\gls{fullOmega_deM}:=\Omega_0+\underbrace{\boldsymbol{\Omega}_{\textnormal{pert}}^\mathbb{X}+\boldsymbol{\Omega}_{\textnormal{pert}}^\E}_{\boldsymbol{\Omega}_{\textnormal{pert}}},
\end{equation}
where 
\begin{equation}
\boldsymbol{\Omega}_{\textnormal{pert}}^\mathbb{X}:=\sum_{n,k\geq 0}\sum_{\Gamma_1'}\frac{(\I\hbar)^{\textnormal{loops}(\Gamma_1')}}{\vert \textnormal{Aut}(\Gamma'_1)\vert}\int_{\de_1M}\left(\sigma_{\Gamma_1'}\right)_{I_1....I_n}^{J_1...J_k}\land\left[\mathbb{X}^{I_1}\right]\land \dotsm \land\left[\mathbb{X}^{I_n}\right] \left((-1)^{kd}(\I\hbar)^k\frac{\delta^{\vert J_1\vert+\dotsm +\vert J_k\vert}}{\delta\left[\mathbb{X}^{J_1}\dotsm \mathbb{X}^{J_k}\right]}\right),
\end{equation}
\begin{equation}
\boldsymbol{\Omega}_{\textnormal{pert}}^\mathbb{E}:=\sum_{n,k\geq 0}\sum_{\Gamma_2'}\frac{(\I\hbar)^{\textnormal{loops}(\Gamma_2')}}{\vert \textnormal{Aut}(\Gamma'_2)\vert}\int_{\de_2M}\left(\sigma_{\Gamma_2'}\right)^{I_1....I_n}_{J_1...J_k}\land \left[\mathbb{E}_{I_1}\right]\land\dotsm\land \left[\mathbb{E}_{I_n}\right] \left((-1)^{kd}(\I\hbar)^k\frac{\delta^{\vert J_1\vert+\dotsm +\vert J_k\vert}}{\delta\left[\mathbb{E}_{J_1}\dotsm \mathbb{E}_{J_k}\right]}\right)\end{equation}
where, for $\mathbb{F}_1=\mathbb{X}$ and $\mathbb{F}_2=\E$ and $\ell\in\{1,2\}$, $\Gamma_\ell'$ runs over graphs with 
\begin{itemize}
\item{$n$ vertices on $\de_\ell M$, where vertex $s$ has valence $\vert I_s\vert\geq 1$, with adjacent half-edges oriented inwards and decorated with boundary fields $[\mathbb{F}^{I_1}_\ell],\ldots,[\mathbb{F}^{I_n}_\ell]$ all evaluated at the point of collapse $p\in \de_\ell M$,}
\item{$\vert J_1\vert+\dotsm +\vert J_k\vert$ outward leaves if $\ell=1$ and $\vert J_1\vert+\dotsm +\vert J_k\vert$ inward leaves if $\ell=2$, decorated with variational derivatives in boundary fields
$$(-1)^d\I\hbar\frac{\delta}{\delta[\mathbb{F}^{J_1}_\ell]},\ldots,(-1)^d\I\hbar\frac{\delta}{\delta[\mathbb{F}^{J_k}_\ell]}$$
at the point of collapse,
}
\item{
no outward leaves if $\ell=2$ and no inward leaves if $\ell=1$ (graphs with them do not contribute).}
\end{itemize}
The form $\sigma_{\Gamma_\ell'}$ is obtained as the integral over the compactified configuration space $\Tilde{\mathsf{C}}_{\Gamma_\ell'}(\mathbb{H}^d)$, where $\mathbb{H}^d$ denotes the $d$-dimensional upper half plane, given by
\begin{equation}
\sigma_{\Gamma_\ell'}=\int_{\Tilde{\mathsf{C}}_{\Gamma_\ell'}(\mathbb{H}^d)}\omega_{\Gamma_\ell'},
\end{equation}
with $\omega_{\Gamma_\ell'}$ being the product of limiting propagators at the point $p$ of collapse and vertex tensors. 
\end{defn}

\begin{figure}[ht]
\begin{tikzpicture}

\draw[thick] (-6,0) -- (6,0); 
\node[vertex] (bdry1) at (-1,0) {};
\node[coordinate, label=below:{}] at (bdry1.south) {};
\node[vertex] (bdry2) at (1,0) {};
\node[coordinate, label=below:{}] at (bdry2.south) {};


\node[vertex] (bulk1) at (-0.5,1) {}; 
\node[] (pi1) at (-0.8,1){};
\node[vertex] (bulk2) at (0.8,1.3) {};
\node[] (pi2) at (1.1,1.3){};
\node[vertex] (bulk3) at (0,0.5) {};
\node[] (pi3) at (0.3,0.5){};


\draw[] (0:2) arc (0:180:2); 
\draw[fermion] (-1,3.5) -- (bulk1);
\draw[fermion] (-0.5,3.5) -- (bulk2);
\draw[fermion] (0,3.5) -- (bulk2);
\draw[fermion] (bulk2) -- (bdry2);
\draw[fermion] (bulk1) -- (bdry1);
\draw[fermion] (bulk1) -- (bulk2);
\draw[fermion] (bulk2) -- (bulk3);
\draw[fermion] (bulk3) -- (bdry1);
\draw[fermion] (bulk3) -- (bulk1);
\draw[fermion] (5,3) -- (bdry2);
\draw[fermion] (5,2) -- (bdry2);
\draw[fermion] (5,1) -- (bdry2);
\draw[fermion] (-5,3) -- (bdry1);
\draw[fermion] (-5,2) -- (bdry1);
\draw[fermion] (-5,1) -- (bdry1);
\end{tikzpicture} 
\caption{Example of a graph collapsing to the boundary with three bulk and two boundary vertices. The semicircle represents the collapsing of the graph.}
\label{fig:collapsing_boundary_op}
\end{figure}
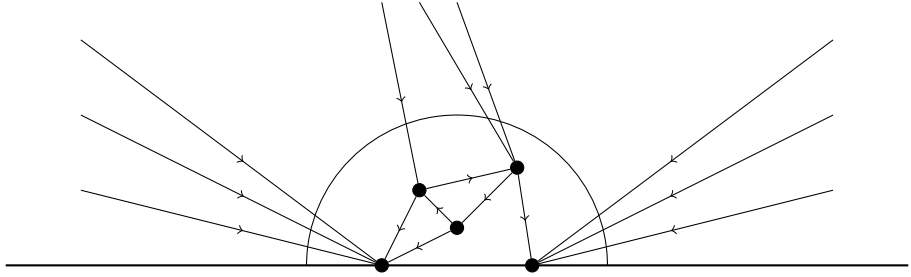

\begin{thm}[\cite{CMR2}]
\label{thm:CMR2}
Let $M$ be a smooth manifold (possibly with boundary). Then the following hold:
\begin{enumerate}
\item{The full covariant state $\boldsymbol{\psi}_M$ satisfies the \textsf{modified Quantum Master Equation}:
\begin{equation}
\label{full_mQME}
(\hbar^2\Delta_{\calV_M}+\boldsymbol{\Omega}_{\de M})\boldsymbol{\psi}_M=0.
\end{equation}
}
\item{The full BFV boundary operator $\boldsymbol{\Omega}_{\de M}$ squares to zero: 
\begin{equation}
\label{flatness_full_boundary_op}
(\boldsymbol{\Omega}_{\de M})^2=0.
\end{equation}
}
\item{A change of propagator or residual fields leads to a theory related by change of data as in \ref{change_of_data}.
}
\end{enumerate}
\end{thm}



\FloatBarrier
\section{Quantization of AKSZ Sigma Models}\label{sec:AKSZ}
In \cite{CMW4} it was shown that one can construct a globalized quantum state in the guise of perturbative quantization for any possibly nonlinear $BF$-like AKSZ Sigma Model \cite{AKSZ} on manifolds with boundary. This is done by considering techniques of formal geometry as in \cite{GK,B} and the BV-BFV formalism. In this section we want to recall the most important concepts of \cite{CMW4}. 

\subsection{AKSZ Sigma Models} 
Let us recall the definition of differential graded symplectic manifolds and AKSZ Sigma Models.

\begin{defn}[Differential graded symplectic manifold]
A \textsf{dg symplectic manifold} of degree $d$ is a graded manifold $\calM$ endowed with a symplectic form $\omega=\dd \alpha$ of degree $d$ and a Hamiltonian function $\Theta$ of degree $d+1$ satisfying $\{\Theta,\Theta\}=0$, where $\{\enspace,\enspace\}$ is the Poisson bracket induced by $\omega$. 
\end{defn}

\begin{rem}
This is sometimes also called a \textsf{Hamiltonian manifold}.
\end{rem}

\begin{defn}[AKSZ Sigma Model]
The \textsf{AKSZ Sigma Model} with target a Hamiltonian manifold $(\calM,\omega=\dd\alpha,\Theta)$ of degree $d-1$ is the BV theory, which associates to a $d$-manifold $\Sigma$ the BV manifold $(\calF_\Sigma,\omega_\Sigma,\calS_\Sigma)$, where\footnote{This is the infinite-dimensional graded manifold adjoint to the Cartesian product (internal morphisms).} $\calF_\Sigma=\Map(T[1]\Sigma,\calM)$, $\omega_\Sigma$ is of the form $\omega_\Sigma=\int_\Sigma\omega_{\mu\nu}\delta \textsf{A}^\mu\land \delta\textsf{A}^\nu$,
and $\calS_\Sigma[\textsf{A}]=\int_\Sigma\left(\alpha_\mu(\textsf{A})\dd \textsf{A}^\mu+\Theta(\textsf{A})\right)$, where $\textsf{A}\in\calF_\Sigma$,
$\omega_{\mu\nu}$ are the components of the symplectic form $\omega$, $\alpha_\mu$ are the components of $\alpha$ and $\textsf{A}^\mu$ are the components of the superfield $\textsf{A}$ in local coordinates. 
\end{defn}

In \cite{CMW4}, we study the following type of AKSZ Sigma Models.

\begin{defn}[Split AKSZ Sigma Model]
We call an AKSZ Sigma Model \textsf{split}\footnote{Note that in $BF$-like AKSZ theories we have a target $T^*[d-1](V[1])=V[1]\oplus V^*[d-2]$, where $V$ is some graded vector space, whereas in the ``split'' case the target is of the form $T^*[d-1]M$ for a graded manifold $M$.}, if the target is of the form 
\begin{equation}
\calM=T^*[d-1]M
\end{equation}
with canonical symplectic structure, where $M$ is a graded manifold.
\end{defn}

\subsection{Formal geometry}
We briefly recall the aspects of formal geometry which are most relevant for this paper. 
Let $M$ be a smooth manifold.

\begin{defn}[Generalized exponential map]
A \textsf{generalized exponential map} is a map $\gls{phi}\colon U\to M$, where $U\subset TM$ is an open neighborhood of the zero section, such that $\varphi(x,0)=x$ and $\dd\varphi(x,0)=\id _{T_xM}$.
\end{defn}

\begin{rem}
For $x\in M$ and $y\in T_xM\cap U$ we write $\varphi(x,y)=\varphi_x(y)$.
\end{rem}

\begin{ex}
An example would be the actual exponential map of a torsion-free linear connection.
\end{ex}

\begin{defn}[Formal exponential map]
A \textsf{formal exponential map} is an equivalence class of generalized exponential maps, where two generalized exponential maps are said to be equivalent if their vertical jets at the zero section agree to all orders.
\end{defn}

For a function $f\in C^\infty(M)$, we can produce a section $\sigma\in\Gamma(\Hat{S}T^*M)$ by defining 
\begin{equation}
\sigma_x:=\mathsf{T}\varphi_x^*f,
\end{equation}
where $\gls{T}$ denotes the Taylor expansion in the fiber coordinates around $y=0$. We denote by $\gls{HatS}$ the completed symmetric algebra. Note that we use any representative of $\varphi$ to define the pullback. We denote this section by $\mathsf{T}\varphi^* f$. Moreover, since it only depends on the jets of the representative, it is independent of the choice of representative.

\begin{defn}[Grothendieck connection]
Given a formal exponential map $\varphi$, we can define the associated \textsf{Grothendieck connection}\footnote{This connection can be extended to a differential on the complex of $\Hat{S}T^*M$-valued differential forms on $\Gamma(\bigwedge^\bullet T^*M\otimes \Hat{S}T^*M)$. Since $\Gamma(\bigwedge^\bullet T^*M\otimes \Hat{S}T^*M)$ is the algebra of functions on the formal graded manifold $\calM:=T[1]M\oplus T[0]M$, the extended differential, which we also denote by  $D_\mathsf{G}$, gives $\calM$ the structure of a differential graded manifold. In particular since $D_\mathsf{G}$ vanishes on the body, we may linearize at each $x\in M$ and get an $L_\infty$-algebra structure on $T_xM[1]\oplus T_xM\oplus T_xM$.} $\gls{D}$ on $\Hat{S}T^*M$, given by $D_\mathsf{G}=\dd+R$, where $\dd$ is the de Rham differential and $\gls{R}\in \Gamma(T^*M\otimes TM\otimes \Hat{S}T^*M)$ is a $1$-form with values in derivations of $\Gamma(\Hat{S}T^*M)$, defined in local coordinates by $R_i\dd x^{i}$ with 
\begin{equation}
\label{vector_field}
R_i(x;y) := \left(\left(\frac{\partial \varphi_x}{\partial y}\right)^{-1}\right)^k_j\frac{\partial\varphi_x^{j}}{\partial x^{i}}\frac{\de }{\de y^k}=: Y^k_i(x;y)\frac{\de}{\de y^k}.
\end{equation}
\end{defn}

\begin{rem}
\label{rem:counterterm}
One can show that \eqref{vector_field} does not depend on the choice of coordinates. 
In fact, we are able to find counterterms for the action to correct for a quantum anomaly since by formal geometry one may resolve functions, differential forms, multivector fields, etc., so actually there is a complex with trivial cohomology in all degrees different from zero. 
In particular, we get
\begin{equation}
H^\bullet_{D_\mathsf{G}}(\Gamma(\Hat{S}T^*M))=H^0_{D_\mathsf{G}}(\Gamma(\Hat{S}T^*M))=\mathsf{T}\varphi^*C^\infty(M)\cong C^\infty(M).
\end{equation}
\end{rem}

\subsection{Globalized BV-BFV Quantization}
\label{subsec:globalized_BV-BFV_quantization}
Now one can use the constructions above to formulate a globalized quantum state, which we call the full covariant state as in \cite{CMW4}. For this we need to extend the action by a formal globalization term, where we also lift the fields as the pullback of the formal exponential map at a constant field $\gls{x}\colon \gls{Sigma} \to M$. This corresponds to linearizing the space of fields $\calF_\Sigma$ around these constant maps. Following \cite{BCM} we give the following definition:

\begin{defn}[Formal globalized action]
\label{formal_globalized_action}
For $(\mathsf{X},\boldsymbol{\eta})\in\calF_\Sigma$, we define the \textsf{formal globalized action} by 
\begin{equation}
\label{formal_glob_action}
\gls{TildeS_Sigma,x}[(\hatX,\hateta)]:=\int_\Sigma\left(\Hat{\boldsymbol{\eta}}_i\land\dd\Hat{\mathsf{X}}^{i}+\mathsf{T}\Tilde{\varphi}_x^*\Theta(\mathsf{X},\boldsymbol{\eta})+Y_i^j(x;\hatX)\hateta_j\land \dd x^{i}\right),
\end{equation}
where $\Tilde{\varphi}_x\colon \Map(T[1]\Sigma,T^*[d-1]T_xM)\to \Map(T[1]\Sigma,\calM)$ denotes the lift of the formal exponential map $\varphi_x$ for $x\in \Sigma$ and $(\hatX,\hateta)$ is the preimage of $(\mathsf{X},\boldsymbol{\eta})$ under this lift.
\end{defn}

\begin{rem}
Note that $\mathsf{X}=\varphi_x(\gls{hatX})$ and $\boldsymbol{\eta}=\left(\dd\varphi_x(\hatX)^*\right)^{-1}\gls{hateta}$.
\end{rem}

\begin{rem}
A similar approach to globalization for closed manifolds was done by Grady--Gwilliam, Costello, Grady--Li--Li \cite{GG,Cost1,GLL}.  Their construction is based on the idea that one can replace the target by an $L_\infty$ equivalent one, whereas the one introduced in \cite{BCM} before was based on the idea of using formal geometry to define a symplectomorphism on a neighborhood of each solution in the space of fields to start the perturbation theory. The two approaches are essentially equivalent. However, in \cite{GG,Cost1,GLL} they only get $BF_\infty$ theories since they start with theories of a particular simple type.
We consider more general theories that do not fit into this framework. Here $BF_\infty$ means that one of the two fields appears at most linearly, but this is not the case in our setting (e.g., in the Poisson Sigma Model for a nonlinear Poisson structure). Moreover, one should work around more general solutions than just the constant ones. In principle, one should do formal geometry on the moduli space of solutions. Note also that this construction can be generalized to non AKSZ models. 
\end{rem}

The Feynman rules corresponding to the formal globalized action as in \eqref{formal_glob_action} are given in Figure \ref{fig:FeynmanRules}. 

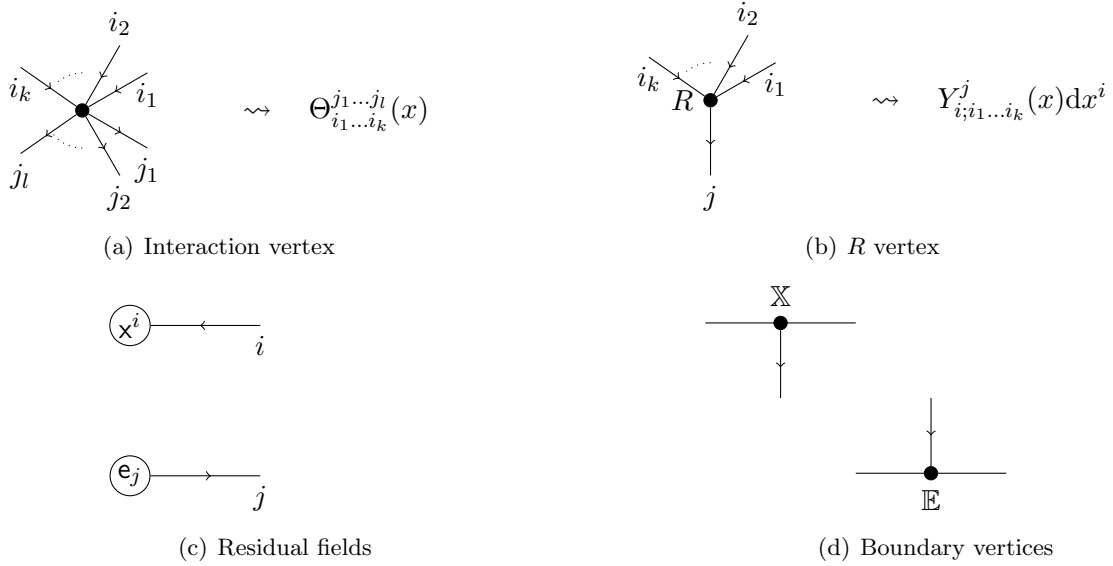
\begin{figure}
\centering
\subfigure[Interaction vertex]{
\centering
\begin{tikzpicture}
\node[vertex] (o) at (0,0) {};
\node[coordinate, label=below:{$i_1$}] at (30:1) {$i_1$}
edge[fermion] (o);
\node[coordinate, label=above:{$i_2$}] at (60:1) {$i_2$}
edge[fermion] (o);
\node[coordinate, label=below:{$i_k$}] at (145:1) {$i_k$}
edge[fermion] (o);
\draw[dotted] (90:0.5) arc (90:130:0.5);
\node[coordinate, label=below:{$j_1$}] (j1) at (-30:1) {};
\node[coordinate, label=below:{$j_2$}] (j2) at (-60:1) {};
\node[coordinate, label=below:{$j_l$}] (j3) at (-145:1) {};
\draw[dotted] (-90:0.5) arc (-90:-145:0.5); 
\draw[fermion] (o) -- (j1);
\draw[fermion] (o) -- (j2);
\draw[fermion] (o) -- (j3);
\node[coordinate,label=right:{$\leadsto\quad\Theta^{j_1\ldots j_l}_{i_1 \ldots i_k}(x)$}] at (2,0) {};
\end{tikzpicture}
}
\hspace{2cm}
\subfigure[$R$ vertex]{
\centering
\begin{tikzpicture}
\node[vertex,label=left:{$R$}] (o) at (0,0) {};
\node[coordinate, label=below:{$i_1$}] at (30:1) {$i_1$}
edge[fermion] (o);
\node[coordinate, label=above:{$i_2$}] at (60:1) {$i_2$}
edge[fermion] (o);
\node[coordinate, label=below:{$i_k$}] at (145:1) {$i_k$}
edge[fermion] (o);
\draw[dotted] (90:0.5) arc (90:130:0.5);
\node[coordinate, label=below:{$j$}] (j1) at (-90:1) {};
\draw[fermion] (o) -- (j1);
\node[coordinate,label=right:{$\leadsto\quad Y^{j}_{i;i_1\ldots i_k}(x)\dd x^i$}] at (2,0) {};
\end{tikzpicture}
}

\hspace{2cm}
\subfigure[Residual fields]{
\centering
\begin{tikzpicture}
\node[circle,draw,inner sep=1pt] (x) at (0,1) {$\mathsf{x}^i$};
\node[coordinate, label=below:{$i$}] (x2) at (30:2) {$i_1$}
edge[fermion] (x);
\node[coordinate, label=below:{$j$}] (e2)at (-30:2) {$i_1$};
\node[circle,draw,inner sep=1pt] (e) at (0,-1) {$\mathsf{e}_j$}
edge[fermion] (e2);
\node[coordinate, label={${}$}] at (4,0) {};
\end{tikzpicture}
}
\hspace{3cm}
\subfigure[Boundary vertices]{
\begin{tikzpicture}
\draw (-1,1) -- (1,1);
\node[vertex, label=above:{$\mathbb{X}$}] (x) at (0,1) {};
\node[coordinate] (b1) at (0,0) {}; 
\draw[fermion] (x) -- (b1);
\draw (1,-1) -- (3,-1);
\node[vertex, label=below:{$\mathbb{E}$}] (e) at (2,-1) {};
\node[coordinate] (b2) at (2,0) {}; 
\draw[fermion] (b2) -- (e);

\end{tikzpicture}
\hspace{2cm}
}
\caption{Summary of Feynman graphs and rules}\label{fig:FeynmanRules}
\end{figure}

\begin{defn}[Principal covariant quantum state]
The \textsf{principal covariant quantum state} $\gls{Tildepsi_Sigma,x}$ is defined as in Definition \ref{principal_part},
using the Feynman rules given in Figure \ref{fig:FeynmanRules} coming from the formal global action $\Tilde{\calS}_{\Sigma,x}$.
\end{defn}

As in the linear case, one needs to define the full covariant state to prove the modified differential Quantum Master Equation.

\begin{defn}[Full covariant quantum state]
\label{full_state_2}
We define the \textsf{full covariant quantum state} $\gls{btpsi_Sigma,x}$ as in Definition \ref{full_covariant_state1}, 
using the Feynman rules in Figure \ref{fig:FeynmanRules} coming from the formal global action $\Tilde{\calS}_{\Sigma,x}$ and additionally with the rules for the boundary vertices as in Figure \ref{fig:composite_field_vertices}.
\end{defn}

One of the main result of \cite{CMW4} is that this state statisfies the globalized version of the modified Quantum Master Equation, which we call the modified differential Quantum Master Equation. It is stated as the following theorem.

\begin{thm}[modified differential Quantum Master Equation for split AKSZ theories]
Consider the full covariant perturbative state $\btpsi_{\Sigma,x}$ as a quantization of an anomaly free and unimodular split AKSZ theory with target $T^*[d-1]\calM$, where $\calM$ is a graded manifold. Then
\begin{equation}
\label{AKSZ_mdQME}
\left(\dr_x -\I\hbar \Delta_{\calV_\Sigma} + \frac{\I}{\hbar} \boldsymbol{\Omega}_{\de \Sigma}\right) \btpsi_{\Sigma,x}=0,
\end{equation}
where we denote by $\dr_x$ the de Rham differential on $M$, the body of the graded manifold $\calM$.
\end{thm}

Note that $\gls{qGBFV}:=\left(\dr_x -\I\hbar \Delta_{\calV_\Sigma} + \frac{\I}{\hbar} \boldsymbol{\Omega}_{\de \Sigma}\right)$ is an operator on the total state space $\Hat{\calH}^\calP_{\Sigma,tot}$. Influenced from the classical case, we call it the \textsf{quantum Grothendieck BFV operator}. Another main result of \cite{CMW4} is the following Theorem.

\begin{thm}
\label{thm:flatness_1}
The quantum Grothendieck BFV operator $\nabla_\mathsf{G}$ squares to zero, i.e.
\begin{equation}
(\nabla_\mathsf{G})^2\equiv 0.
\end{equation}
\end{thm}

\begin{rem}
One can also think of $\nabla_\mathsf{G}$ as a flat connection on the total bunlde of states \cite{CMW4}.
\end{rem}

\section{Review of the Poisson Sigma Model}
\label{sec:PSM}
The Poisson Sigma Model \cite{I,SS1,SS2} is a 2-dimensional topological field theory, with important relation to deformation quantization \cite{K,CF1,CF2}, see also Appendix \ref{app:conn_PSM}, and in particular a special case of an AKSZ Sigma Model. In this section we will very briefly review some aspects of its classical version. 
\subsection{Classical Poisson Sigma Model}
Let us fix a Poisson manifold $(\gls{Poisson_mnf},\gls{Poisson_str})$. 
\begin{defn}[Classical Poisson Sigma Model]
The classical Poisson Sigma Model associates to a smooth, oriented, compact  and connected $2$-manifold $\Sigma$ (usually called the worldsheet)  the space of fields $F_{\Sigma} = \textbf{VBun}(T\Sigma,T^*\mathscr{P})$ of vector bundle maps from $T\Sigma$ to $T^*\mathscr{P}$. An element of $F_\Sigma$ will be identified with a pair $(X,\eta)$ where $\gls{X} \colon \Sigma \to \mathscr{P}$ is the base map and $\gls{eta} \in \Gamma(\Sigma,T^*\Sigma \otimes X^*T^*\mathscr{P})$ is a 1-form on $\Sigma$ with values in $X^*T^*\mathscr{P}$. The action functional is 
\begin{equation}
\label{PSM}
S_\Sigma[(X,\eta)]=\int_\Sigma\left(\langle \eta, \dd X \rangle + \frac{1}{2} \langle\Pi(X), \eta \land\eta \rangle\right),
\end{equation}
where $\langle\enspace , \enspace\rangle $ denotes the pairing between vectors and covectors. 
\end{defn}

\begin{rem}
In local coordinates $x^i$ on $\mathscr{P}$, we can write $\eta = \eta_i\dd x^i$ and $X = (X^1,\ldots, X^n)$. Then the action reads
\begin{equation}
\label{PSMlocal}
S_\Sigma[(X,\eta)]=\int_\Sigma\left(\eta_i\land \dd X^{i}+\frac{1}{2}\Pi^{ij}(X)\eta_i\land \eta_j\right),
\end{equation}
where we use the Einstein summation convention. 
\end{rem}

\subsection{BV-BFV extension}
The Poisson Sigma Model is a \textsf{gauge theory}, in the sense that the Lagrangian is invariant under infinitesimal gauge transformations. More precisely there is a distribution on the space of fields which leaves the action invariant and closes on shell, i.e. once the equations of motions are imposed. In particular, the infinitesimal symmetries for the Poisson Sigma Model are given by the following gauge transformations
\begin{align}
\delta_\beta X^{i}&=\Pi^{ij}(X)\beta_j,\\
\delta_\beta \eta_i&=-\dd \beta_i-\partial_i\Pi^{ij}(X)\eta_j\beta_k,
\end{align}
where $\beta_i$ is an infinitesimal parameter that is a section of $X^*T^*M$. If $\de\Sigma\not=\varnothing$, we also want that $\beta_i$ vanishes on $\de\Sigma$ since one wants $\eta$ to vanish on the boundary. 

Because the gauge symmetries only close on shell, the BRST formalism fails, and one needs to revert to the BV formalism \cite{CF1,CF2} on closed surfaces and to the BV-BFV formalism on surfaces with boundary \cite{CMR1,CMR2}.
The BV extended action and space of fields for the Poisson Sigma Model can be constructed from the AKSZ formalism as discussed in \cite{CF4}. 

\begin{defn}[BV extended Poisson Sigma Model]
\label{BV_PSM}
The BV theory associated to the Poisson Sigma Model is given by the triple $$(\calF_\Sigma,\omega_\Sigma,\calS_\Sigma),$$ where the \textsf{BV space of fields} is given by
\begin{equation}
\calF_{\Sigma} := \Map(T[1]\Sigma, T^*[1]\mathscr{P}) \ni (\sfX,\sfeta)
\end{equation}
with $\gls{superX} \colon T[1]\Sigma \to \mathscr{P}$ a map and $\gls{superEta}$ a section of $\sfX^*T^*[1]\mathscr{P}$, the \textsf{BV action} is given by 
\begin{equation}
\label{eq:BV_action}
\calS_{\Sigma}[(\sfX,\sfeta)] := \int_{T[1]\Sigma}\left(\langle \sfeta, D \sfX \rangle + \frac{1}{2}\langle \Pi(\sfX), \sfeta \wedge \sfeta \rangle\right),
\end{equation} 
where $D=\theta^\mu\frac{\partial}{\partial u_\mu}$ is the differential on $T[1]\Sigma$ for local even coordinates $\{u_\mu\}$ on $\Sigma$ and odd coordinates $\{\theta^\mu\}$, and the \textsf{BV symplectic form} is given by 
\begin{equation}
\label{BVform}
\omega_{\Sigma} := \int_{\Sigma} \delta \sfX \wedge \delta\sfeta.
\end{equation}
\end{defn}

\begin{rem}
In local coordinates on $\mathscr{P}$ we can write
\begin{equation}
\calS_{\Sigma}[(\sfX,\sfeta)] := \int_{\Sigma}\left( \sfeta_i \land\dd \sfX^{i} + \frac{1}{2} \Pi^{ij}(\sfX) \sfeta_i \wedge \sfeta_j\right).
\end{equation}
\end{rem}

On closed surfaces, this action satisfies the Classical Master Equation
\begin{equation}
\label{CME}
(\calS_{\Sigma},\calS_{\Sigma}) = 0. 
\end{equation}
Here $(\enspace,\enspace)$ is the odd Poisson bracket (BV bracket) associated to the odd symplectic form $\omega_\Sigma$.

One can reformulate the Classical Master Equation in terms of the cohomological vector field as $Q_{\Sigma}\left(\calS_{\Sigma}\right) = 0$ where $Q_{\Sigma} = \left(\calS_{\Sigma},\enspace\right)$ in local coordinates on $\mathscr{P}$ is given by
\begin{equation}
\label{Qdef}
Q_{\Sigma}  = \int_{\Sigma}\left( \left(\dd \sfX^i + \Pi^{ij}(\sfX) \wedge \sfeta_j\right)\wedge
\frac{\delta}{\delta \sfX^i} + \left(\dd \sfeta_i + \frac{1}{2}\frac{\partial}{\partial x^i}\Pi^{jk}(\mathsf{X})\sfeta_j\land\sfeta_k\right) \wedge \frac{\delta}{\delta \sfeta_i}\right).
\end{equation}

In the BV-BFV formalism the boundary conditions are left unspecified and hence the Classical Master Equation no longer makes sense. However, one can still define the symplectic form $\omega_\Sigma$ by \eqref{BVform}, the action by \eqref{eq:BV_action} and the vector field $Q_\Sigma$ by \eqref{Qdef}.

\begin{defn}[BV-BFV extended Poisson Sigma Model]
The BV-BFV theory associated to the Poisson Sigma Model is given by associating to a manifold $\Sigma$ with boundary $\partial\Sigma$ the BV-BFV manifold $$(\calF^\partial_{\partial\Sigma},\pi_\Sigma,\omega^\partial_{\partial\Sigma}=\delta\alpha^\partial_{\partial\Sigma},\calS^\partial_{\partial\Sigma},Q^\partial_{\partial\Sigma})$$ over the BV manifold $(\calF_\Sigma,\omega_\Sigma,\calS_\Sigma)$, where 
\begin{align}
\calF^\partial_{\de\Sigma} &:= \Map(T[1]\de\Sigma,T^*[1]\mathscr{P}), \\
\alpha^\de_{\de \Sigma} &:= \int_{\de \Sigma} \sfeta \wedge \delta \sfX, \\
Q^\partial_{\de \Sigma}  &:= \int_{\de\Sigma}\left( \left(\dd \sfX^i + \Pi^{ij}(\sfX) \wedge \sfeta_j\right)\wedge
\frac{\delta}{\delta \sfX^i} + \left(\dd \sfeta_i + \frac{1}{2}\frac{\partial}{\partial x^i}\Pi^{jk}(\mathsf{X})\sfeta_j\land\sfeta_k\right) \wedge \frac{\delta}{\delta \sfeta_i}\right), \\
\calS^\partial_{\de \Sigma} &:=  \int_{\de\Sigma}\left(\langle \sfeta, \dd \sfX \rangle + \frac{1}{2}\langle \Pi(\sfX), \sfeta \wedge \sfeta \rangle\right),
\end{align} 
and the map $\pi_\Sigma\colon \calF_{\Sigma} \to \calF_{\de\Sigma}^\de$ given by restriction of maps. 
\end{defn}


As shown in \cite{CMR1}, these then satisfy the axioms of a BV-BFV theory\footnote{This is automatic for theories which admit an AKSZ formulation.}: 
\begin{align}
Q_{\Sigma}^2 &= 0, \\
\delta\pi_\Sigma(Q_{\Sigma}) &= Q_{\de\Sigma}^\de, \\ 
\iota_{Q_{\Sigma}}\omega_{\Sigma} &= \delta \calS_{\Sigma} + \pi_\Sigma^*\alpha^\de_{\de\Sigma}. \label{mCMEI}
\end{align}
Moreover, $\iota_{Q^{\partial}_{\partial\Sigma}}\omega^\partial_{\partial \Sigma}=\delta\calS^{\partial}_{\partial\Sigma}$. As in Subsection \ref{classicalBV-BFV}, we have the modified Classical Master Equation $\iota_{Q_\Sigma}\iota_{Q_\Sigma}\omega_{\Sigma}=2\pi_\Sigma^*\calS^\de_{\de\Sigma}$.

\section{Globalized BV-BFV Quantization of the Poisson Sigma Model}
\label{sec:GlobalizedPSM}
In this section, we analyze in detail the construction explained in Section \ref{sec:AKSZ} applied to the Poisson Sigma Model.
In particular, we want to describe the BFV boundary operator $\boldsymbol{\Omega}_{\de\Sigma}$ for the Poisson Sigma Model in the case of a worldsheet where we have a single boundary component endowed with a certain polarization (see Figure \ref{simply_pol}). As a preparation for the remainder of the paper, we also discuss how the boundary operator behaves under certain modifications of the formal globalized action. 

\subsection{Globalization at the classical level} 
\label{subsec:pert_quant}
We consider the Poisson Sigma Model action as a perturbation of the quadratic part of the action, 
\begin{equation}
\calS_{0,\Sigma} = \int_{\Sigma} \langle \sfeta , \dd \sfX \rangle .
\end{equation}
Recall that we expand around critical points of $\calS_{0,\Sigma}$, which in particular satisfy $\dd \sfX=0$. Hence the ghost number 0 component of $\sfX$ is a constant map, which we denote by its image $x \in \mathscr{P}$. As discussed in \cite{CF3,BCM,CMW2} and Appendix F of \cite{CMR2}, it makes sense to perform perturbative quantization around points in the moduli space of classical solutions. Since the Euler--Lagrange equations for the Poisson Sigma Model are given by $\dd\mathsf{X}+\Pi(\mathsf{X})\boldsymbol{\eta}=0$, we will perturb around the classical solution $\mathsf{X}=x=const.$ and $\boldsymbol{\eta}=0$ and gauge equivalent solutions.
Hence for the Poisson Sigma Model the appropriate moduli space is given by 
\begin{equation}
\label{modulispace} 
\calM_0=\{(x,0)| x \text{ const map to } \mathscr{P}\} \cong \mathscr{P}.
\end{equation}
In this special case we have $\calM_0 \subset \calF_{\Sigma}$. Instead of fixing a single classical solution $x \in \calM_0$ and expanding around it, we want to vary $x$ itself. As in Subsection \ref{subsec:globalized_BV-BFV_quantization} we consider the fields $\hatX$ and $\hateta$ given by $\mathsf{X}=\varphi_x(\hatX)$ and $\boldsymbol{\eta}=((\dd\varphi_x)^*)^{-1}\hateta$. 

We get a \textsf{formally globalized action} for the Poisson Sigma Model as in Definition \ref{formal_globalized_action} by 
\begin{equation}
\Tilde\calS_{\Sigma,x}[(\hatX,\hateta)]=\underbrace{\int_\Sigma\hateta_i\land\dd_\Sigma\hatX^{i}}_{=:\Hat{\calS}_{0,\Sigma}}+\underbrace{\frac{1}{2}\int_\Sigma(\T\varphi_x^*\Pi)^{ij}\left(\hatX\right)\hateta_i\land\hateta_j}_{=:\Hat{\calS}_{\Pi,\Sigma,x}}+\underbrace{\int_\Sigma Y_i^j\left(x;\hatX\right)\hateta_j\land\dd_{\mathcal{M}_0}x^{i}}_{=:\Hat{\calS}_{\Sigma,x,R}},\label{eq:GlobalPSMAction}
\end{equation}
where we denote by $\dd_{\mathcal{M}_0}$ and $\dd_\Sigma$ the de Rham differentials on $\mathcal{M}_0$ and $\Sigma$ respectively (we only write it once and leave out the indication every time it is clear).

\begin{figure}[ht]
\begingroup%
  \makeatletter%
  \providecommand\color[2][]{%
    \errmessage{(Inkscape) Color is used for the text in Inkscape, but the package 'color.sty' is not loaded}%
    \renewcommand\color[2][]{}%
  }%
  \providecommand\transparent[1]{%
    \errmessage{(Inkscape) Transparency is used (non-zero) for the text in Inkscape, but the package 'transparent.sty' is not loaded}%
    \renewcommand\transparent[1]{}%
  }%
  \providecommand\rotatebox[2]{#2}%
  \ifx\svgwidth\undefined%
    \setlength{\unitlength}{187.22693743bp}%
    \ifx\svgscale\undefined%
      \relax%
    \else%
      \setlength{\unitlength}{\unitlength * \real{\svgscale}}%
    \fi%
  \else%
    \setlength{\unitlength}{\svgwidth}%
  \fi%
  \global\let\svgwidth\undefined%
  \global\let\svgscale\undefined%
  \makeatother%
  \begin{picture}(1,0.84519289)%
    \put(0,0){\includegraphics[width=\unitlength]{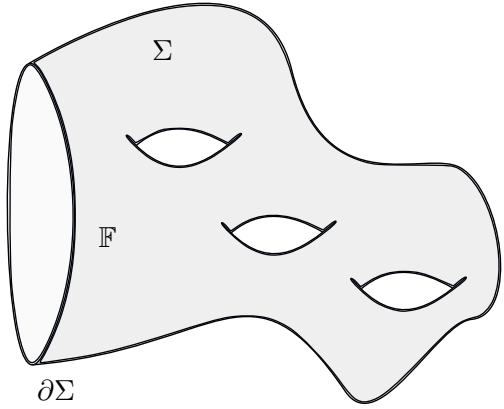}}%
    \put(0.0614748,0.00872313){\color[rgb]{0,0,0}\makebox(0,0)[lb]{\smash{$\de\Sigma$}}}%
    \put(0.18601559,0.31797255){\color[rgb]{0,0,0}\makebox(0,0)[lb]{\smash{$\mathbb{F}$}}}%
    \put(0.29317854,0.69458329){\color[rgb]{0,0,0}\makebox(0,0)[lb]{\smash{$\Sigma$}}}%
  \end{picture}%
\endgroup%

%
%

\caption{Example of a higher genus worldsheet with one connected boundary component and different polarization. Here $\mathbb{F}$ denotes either $\mathbb{X}$ or $\E$.}
\label{simply_pol}
\end{figure}

\subsection{The boundary BFV operator} 
\label{subsec:operator}
In this subsection we want to see how $\boldsymbol{\Omega}_{\de\Sigma}$ is constructed for a formal linearized action but without any globalization term, i.e. for $\Hat{\calS}_{\Sigma,x} = \Hat{\calS}_{0,\Sigma} + \Hat{\calS}_{\Pi,\Sigma,x}$ in the notation of Equation \eqref{eq:GlobalPSMAction}.
We can formulate the boundary operator $\boldsymbol{\Omega}_{\de \Sigma}$ for the Poisson Sigma Model by the usual construction of the collapsing of subgraphs $\Gamma'$ using Definition \ref{full_BFV}
for the non-globalized theory. We briefly review the results of \cite[Section 4.8]{CMR2}, where the boundary operator of the non-globalized theory was computed. Recall the splitting of the space of fields as in \eqref{split}
\begin{align}
\begin{split}
\mathcal{F}_\Sigma&\to\calB_{\partial \Sigma}^\calP\oplus \mathcal{V}^\calP_{\Sigma}\oplus \mathcal{Y}'\\
(\mathsf{X},\boldsymbol{\eta})&\mapsto (\gls{bbX},\gls{bbE})\oplus (\gls{sfx},\gls{sfe})\oplus (\gls{scrX},\gls{scrE}).
\end{split}
\end{align}
We now describe the BFV boundary operator for the different representations\footnote{We call the $\frac{\delta}{\delta\E}$-polarization the $\mathbb{X}$-representation and vice versa.}.

\subsubsection{$\E$-representation}
We look first at the $\E$-representation.
\begin{prop}
\label{E_rep}
In the $\E$-representation, the boundary operator is given by
\begin{equation}
\label{eq:E-rep}
\boldsymbol{\Omega}^\E_{\de\Sigma}=\Omega^\E_0+\boldsymbol{\Omega}_{\textnormal{pert}}^\E,
\end{equation}
where 
\begin{equation}
\boldsymbol{\Omega}_{\textnormal{pert}}^\E=\int_{\de\Sigma} \sum_{I,J,K,R,S}\frac{(-\I\hbar)^{\vert K\vert-\vert I\vert-\vert J\vert+1}}{(\vert K\vert+\vert R\vert+\vert S\vert)!}\partial_K B^{IJ}(\mathsf{T}\varphi_x^*\Pi)[\E_I\E_R][\E_J\E_S]\frac{\delta^{\vert K\vert+\vert R\vert+\vert S\vert}}{\delta[\E_K]\delta[\E_R]\delta[\E_S]},
\end{equation}
where the $B^{IJ}$s are defined as the coefficients in the star product on $C^\infty(\R^n)[[\hbar]]$ by 
\begin{equation}
f\gls{star} g=fg+\sum_{I,J}B^{IJ}\frac{\partial^{\vert I\vert}}{\partial x^{I}}f\frac{\partial^{\vert J\vert}}{\partial x^J}g=fg-\frac{\I\hbar}{2}\sum_{ij}(\mathsf{T}\varphi_x^*\Pi)^{ij}\frac{\partial f}{\partial x^{i}}\frac{\partial g}{\partial x^j}+O(\hbar^2),
\end{equation}
where $I$, $J$ are multi-indices and $i$ and $j$ are indices and $B^{IJ}=0$ if $\vert I\vert=0$ or $\vert J\vert=0$, and $\star$ denotes Kontsevich's star product \cite{K}.
\end{prop}

\begin{proof}[Proof]
Consider a graph $\Gamma'$ with $n$ bulk vertices and $k$ boundary vertices collapsing on the $\E$-boundary. Note that we have $\dim \Tilde{\mathsf{C}}_{\Gamma'}(\mathbb{H}^d)=2n+k-2$, which has to be the same as the form degree of $\omega_{\Gamma'}$ so that the integral 
\begin{equation}
\sigma_{\Gamma'}=\int_{\Tilde{\mathsf{C}}_{\Gamma'}(\mathbb{H}^d)}\omega_{\Gamma'}
\end{equation}
does not vanish. 

Thus we need to have $2n+k-2=2n$, since $n$ is the number of points in the bulk which represent the Poisson tensor, i.e. emitting two arrows that have to remain inside the collapsing subgraph (otherwise the contribution vanishes by the boundary condition on the propagator). Hence we get $k=2$, i.e. the graph has exactly two boundary vertices. We label one boundary vertex by $u_0$ and the other one by $u_1$. Let $L$ be a multiindex labeling the inward leaves of $\Gamma'$. We decompose $L$ as $L=(R,K,S)$, where $R,K,S$ are again multiindices, representing different types of inward leaves. $R$ labels the leaves arriving directly at $u_0$, $S$ labels the leaves arriving directly to $u_1$ and $K$ labels the leaves arriving at some bulk vertices of $\Gamma'$. Moreover, we label by the multiindex $I$ the arrows arriving at $u_0$ from some bulk vertices of $\Gamma'$ and by the multiindex $J$ the arrows arriving at $u_1$ from some bulk vertices of $\Gamma'$ (see Figure \ref{fig:E-rep1}). Since we have exactly two boundary vertices ($k=2$), the graphs when considered without leaves are given by the same graphs as in Kontsevich's star product. If we sum over all graphs having the same multiindices $K,I,J$, we obtain the $K$'th derivative of the $B^{IJ}$ coefficient in the star product, since the limiting propagator coincides with Kontsevich's propagator, and hence we get \eqref{eq:E-rep}.
\end{proof}
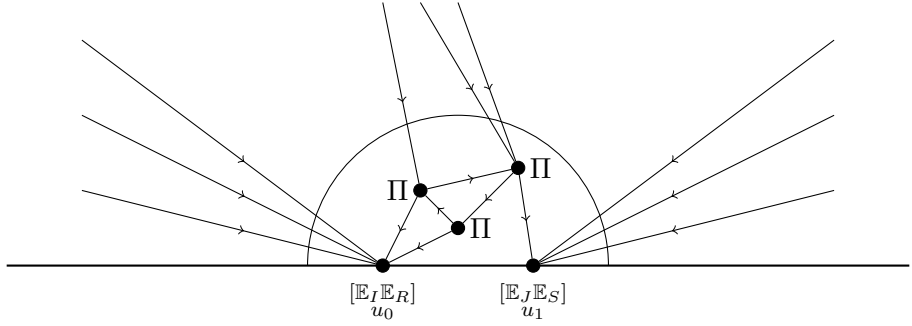
\begin{figure}[ht]
\begin{tikzpicture}

\draw[thick] (-6,0) -- (6,0); 
\node[vertex] (bdry1) at (-1,0) {};
\node[coordinate, label=below:{$\substack{[\E_{I}\E_R]\\ u_0}$}] at (bdry1.south) {};
\node[vertex] (bdry2) at (1,0) {};
\node[coordinate, label=below:{$\substack{[\E_J\E_S]\\ u_1}$}] at (bdry2.south) {};


\node[vertex] (bulk1) at (-0.5,1) {}; 
\node[] (pi1) at (-0.8,1){$\Pi$};
\node[vertex] (bulk2) at (0.8,1.3) {};
\node[] (pi2) at (1.1,1.3){$\Pi$};
\node[vertex] (bulk3) at (0,0.5) {};
\node[] (pi3) at (0.3,0.5){$\Pi$};


\draw[] (0:2) arc (0:180:2); 
\draw[fermion] (-1,3.5) -- (bulk1);
\draw[fermion] (-0.5,3.5) -- (bulk2);
\draw[fermion] (0,3.5) -- (bulk2);
\draw[fermion] (bulk2) -- (bdry2);
\draw[fermion] (bulk1) -- (bdry1);
\draw[fermion] (bulk1) -- (bulk2);
\draw[fermion] (bulk2) -- (bulk3);
\draw[fermion] (bulk3) -- (bdry1);
\draw[fermion] (bulk3) -- (bulk1);
\draw[fermion] (5,3) -- (bdry2);
\draw[fermion] (5,2) -- (bdry2);
\draw[fermion] (5,1) -- (bdry2);
\draw[fermion] (-5,3) -- (bdry1);
\draw[fermion] (-5,2) -- (bdry1);
\draw[fermion] (-5,1) -- (bdry1);
\end{tikzpicture} 
\caption{An example of a subgraph collapsing as in the description. Here we have three incoming arrows to the boundary for the collapsing graph $\Gamma'$ on the right side corresponding to the index $S$, three incoming arrows to the boundary on the left side corresponding to the index $R$, three incoming arrows to $\Gamma'$ corresponding to the index $K$, two incoming arrows to $u_0$ from $\Gamma'$ corresponding to the index $I$ and one incoming arrow to $u_1$ from $\Gamma'$ corresponding to the index $J$.}
\label{fig:E-rep1}
\end{figure}

To analyze the BFV boundary operator, we introduce the notion of certain multiplication operators appearing from collapsing graphs on the boundary endowed with the $\mathbb{E}$-representation. Therefore we give the following definition:

\begin{defn}[Exponential multiplication operator]
\label{mult_op}
The \textsf{exponential multiplication operator} for the boundary field $\mathbb{E}$ is given by the map 
\begin{align}
\ee^{\frac{\I}{\hbar}[\mathbb{E}]y}\colon \Hat{\calH}_{\Sigma}&\to \Hat{\calH}_\Sigma[[y]]\\
\phi&\mapsto \ee^{\frac{\I}{\hbar}[\mathbb{E}]y}\phi:=\sum_{k\geq 0}\left(\frac{\I}{\hbar}\right)^k\sum_{\substack{I\\ \vert I\vert =k}}y^{I}\left(\int_{\de \Sigma}[\mathbb{E}_{I}]\right)\cdot\phi
\end{align}
On the total space $\Hat{\calH}_{\Sigma,\textnormal{tot}}$, the multiplication operator is given by a map 
\begin{equation}
\Hat{\calH}_{\Sigma,\textnormal{tot}}\to \Hat{\calH}_{\Sigma,\textnormal{tot}}\otimes \Hat{S}T^*\mathscr{P}.
\end{equation}
\end{defn}

\begin{rem}
Note that the exponential multiplication operator takes regular functionals to regular functionals. The construction in \cite{CF3}, recalled in Appendix \ref{app:PSM}, yields a bundle $\mathcal{E} = \widehat{S}T^*\mathscr{P}[[\hbar]]$ of $\star$-algebras on $\mathscr{P}$ by applying Kontsevich's deformation quantization in every tangent space.

Thus, we can define a map 
\begin{equation}
\star\colon \Gamma(\Hat{\calH}_{\Sigma,\textnormal{tot}}\otimes \Hat{S}T^*\mathscr{P})\otimes\Gamma(\Hat{\calH}_{\Sigma,\textnormal{tot}}\otimes \Hat{S}T^*\mathscr{P})\to \Gamma(\Hat{\calH}_{\Sigma,\textnormal{tot}}\otimes \Hat{S}T^*\mathscr{P}),
\end{equation}
given by multiplication in $\Gamma(\Hat{\calH}_{\Sigma,\textnormal{tot}})$ and the fiber wise star product in $\Hat{S}T^*\mathscr{P}[[\hbar]]$, i.e. we consider the tensor product of the two algebra bundles $\Hat{\calH}_{\Sigma,\textnormal{tot}}$ and $\Hat{S}T^*\mathscr{P}$ over $\mathbb{C}[[\hbar]]$. 
\end{rem}

\begin{rem}
Note that we can define a map from $\Gamma(\Hat{\calH}_{\Sigma,\textnormal{tot}}\otimes \Hat{S}T^*\mathscr{P})$ to the space of operators, by replacing the fiber coordinates $y^{I}$ by functional derivatives $\frac{\delta}{\delta[\E_{I}]}$. Thus, if we have a section $\sigma$ of $\Hat{\calH}_{\Sigma,\textnormal{tot}}$, we can define the boundary operator $\boldsymbol{\Omega}_{\textnormal{pert}}^\E$ by 
\begin{equation}
\label{boundary_op_exp}
\boldsymbol{\Omega}_{\textnormal{pert}}^\E\sigma=\left(\ee^{\frac{\I}{\hbar}[\E]y}\star \ee^{\frac{\I}{\hbar}[\E]y}\right)\Big|_{y=\frac{\delta}{\delta[\E]}}\sigma.
\end{equation}
Then one can check that \eqref{eq:E-rep} is given by the standard quantization\footnote{Choosing a leaf $b\in\calB^\calP_{\de \Sigma}$ one considers it's conjugated momentum $-\I\hbar \frac{\delta}{\delta b}$.} of the boundary action 
\begin{equation}
\label{boundary_action_exp}
\calS^\de_{\de\Sigma}=\int_{\de\Sigma}\left(\langle\Hat{\boldsymbol{\eta}},\dd \Hat{\mathsf{X}}\rangle+\frac{1}{2}\left[\ee^{\frac{\I}{\hbar}\Hat{\boldsymbol{\eta}}},\ee^{\frac{\I}{\hbar}\Hat{\boldsymbol{\eta}}}\right]_\star(\Hat{\mathsf{X}})\right),
\end{equation}
where $\langle\enspace,\enspace\rangle$ denotes the canoncial pairing of $T_x\mathscr{P}$ with $T^*_x\mathscr{P}$, where $x$ is the constant background field $\Sigma\to\mathscr{P}$, $\gls{star_commutator}$ is the star commutator, and $\star$ is the star product in $T_x\mathscr{P}$. Note that the interesting part here is that we can view the BFV boundary operator as the standard quantization of a deformed boundary action. 
\end{rem}

\begin{rem}The fact that $(\boldsymbol{\Omega}^\mathbb{E}_{\de \Sigma})^2 = 0$ is equivalent to the associativity of the Kontsevich star product. 
\end{rem}

\subsubsection{$\mathbb{X}$-representation}
Next, we consider the $\mathbb{X}$-representation.
\begin{prop}
\label{X_rep}
In the $\mathbb{X}$-representation, the boundary operator is given by
\begin{equation}
\label{eq:X-rep}
\boldsymbol{\Omega}^\mathbb{X}_{\de\Sigma}=\Omega^\mathbb{X}_0+\boldsymbol{\Omega}_{\textnormal{pert}}^\mathbb{X},
\end{equation}
where 
\begin{multline}
\boldsymbol{\Omega}_{\textnormal{pert}}^\mathbb{X}=\sum_{k=0}^\infty\frac{1}{k!}\int_{\de \Sigma} \sum_{L,I_1,\ldots,I_k,R_1,\ldots,R_k}\frac{(-\I\hbar)^{\vert L\vert-(\vert I_1\vert+\dotsm+\vert I_k\vert)+1}}{(\vert L\vert+\vert R_1\vert+\dotsm+\vert R_k\vert)!}\cdot\\
\cdot\partial^L a_{I_1,\ldots,I_k}(\mathsf{T}\varphi_x^*\Pi)\prod_{j=1}^k[\mathbb{X}^{I_j}\mathbb{X}^{R_j}]\frac{\delta^{\vert L\vert+\vert R_1\vert+\dotsm +\vert R_k\vert}}{\delta[\mathbb{X}^L]\delta[\mathbb{X}^{R_1}]\dotsm\delta[\mathbb{X}^{R_k}]},
\end{multline}
where $a_{I_1,\ldots,I_k}$ are given by the sum of the weights over all Feynman graphs with $k$ boundary vertices and $\vert I_j\vert$ outgoing arrows for $1\leq j\leq k$.
\end{prop}

\begin{proof}
In the $\mathbb{X}$-representation there can be arbitrarily many vertices on the boundary, since the arrows emanating from the bulk vertices can now leave the graph. Denote the number of vertices on the boundary by $k$. Then we have a similar construction as for the $\E$-representation, only with the difference that for each boundary vertex we can have arbitrarily many outgoing arrows either out of the collapsing graph $\Gamma'$ (left or right) and arbitrarily many outgoing arrows going into $\Gamma'$. Label the vertices $u$ on the boundary by $1,\ldots,k$.  We denote by $L$ the multiindex labeling the leaves, which emanate from bulk vertices of $\Gamma'$, by $I_j$ the multiindices labeling the arrows which start at $u_j$ and end at some bulk vertices of $\Gamma'$ for $1\leq j\leq k$ and by $R_j$ the multiindices labeling the outward leaves which start at $u_j$ for $1\leq j\leq k$ (see Figure \ref{fig:X-rep1}). Summing over all such graphs $\Gamma'$, we get \eqref{eq:X-rep}.
\end{proof}

\begin{figure}[ht]
\begin{tikzpicture}

\draw[thick] (-6,0) -- (6,0); 
\node[vertex] (bdry1) at (-1,0) {};
\node[coordinate, label=below:{$\substack{[\mathbb{X}^{I_1}\mathbb{X}^{R_1}]\\u_1}$}] at (bdry1.south) {};
\node[vertex] (bdry2) at (1,0) {};
\node[coordinate, label=below:{$\substack{[\mathbb{X}^{I_2}\mathbb{X}^{R_2}]\\ u_2}$}] at (bdry2.south) {};

\node[vertex] (bulk1) at (-0.5,1) {}; 
\node[] (pi1) at (-0.8,1){$\Pi$};
\node[vertex] (bulk2) at (0.8,1.3) {};
\node[] (pi2) at (1.1,1.3){$\Pi$};
\node[vertex] (bulk3) at (0,0.5) {};
\node[] (pi3) at (0.3,0.5){$\Pi$};


\draw[] (0:2) arc (0:180:2); 
\draw[fermion] (bulk1) -- (-1,3.5);
\draw[fermion] (bulk2) -- (-0.5,3.5);
\draw[fermion] (bulk2) -- (0,3.5);
\draw[fermion] (bdry2) -- (bulk2);
\draw[fermion] (bdry1) -- (bulk1);
\draw[fermion] (bulk1) -- (bulk2);
\draw[fermion] (bulk3) -- (bulk2);
\draw[fermion] (bulk3) -- (bulk1);
\draw[fermion] (bdry2) -- (5,3);
\draw[fermion] (bdry2) -- (5,2);
\draw[fermion] (bdry2) -- (5,1);
\draw[fermion] (bdry1) -- (-5,3);
\draw[fermion] (bdry1) -- (-5,2);
\draw[fermion] (bdry1) -- (-5,1);
\end{tikzpicture} 
\caption{An example of a subgraph collapsing as in the description. We consider a term for $k=2$ as before and we label them by $u_1$ and $u_2$. Here we have three outgoing arrows for the collapsing graph $\Gamma'$ on the right side corresponding to the index $R_1$, three outgoing arrows on the left side corresponding to the index $R_1$, three outgoing arrows to $\Gamma'$ corresponding to the index $L$ and one incoming to $\Gamma'$ out of each of the two boundary points corresponding to the indices $I_1$ and $I_2$.}
\label{fig:X-rep1}
\end{figure}
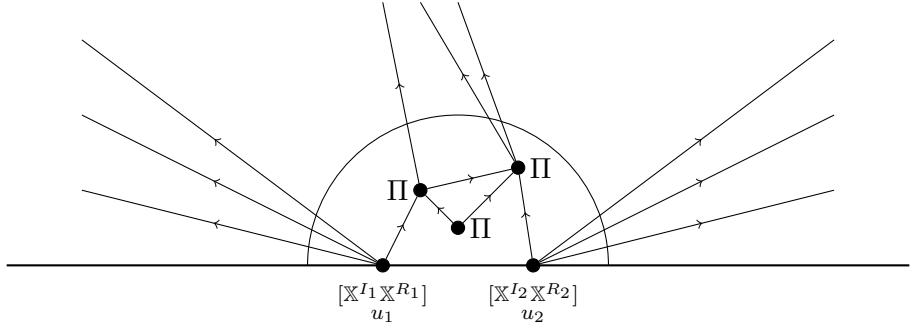

\begin{rem}
The coefficients $a_{I_1,\ldots,I_k}$ can be regarded as the coefficients of an $A_\infty$-algebra \cite{CF6}. The fact that $\boldsymbol{\Omega}_{\de \Sigma}^\mathbb{X}$ squares to zero corresponds to the $A_\infty$-relations. 
\end{rem}

\subsection{The globalized BFV operator}
\label{subsec:BFV_op}
We now give a formulation for $\boldsymbol{\Omega}_{\de \Sigma}$ where we also consider the globalization term $\Hat{\calS}_{\Sigma,x,R}$. Recall that graphically this amounts to introducing new vertices emanating only a single arrow, representing the vector field $R$ as explained in the Feynman rules of Section \ref{sec:AKSZ}. This means that $\boldsymbol{\Omega}_{\de\Sigma}$ now becomes an inhomogeneous form on the Poisson manifold $\mathscr{P}$, since $R$ is a 1-form on $\mathscr{P}$. As before, we distinguish between the $\bbE$- and the $\bbX$-representation. 
\subsubsection{$\bbE$-representation}
\label{subsubsec:E-rep2} 
We start with the $\E$-representation.
\begin{prop}
\label{E_rep2}
In the $\mathbb{E}$-representation, the globalized boundary operator is given by
\begin{equation}
\label{eq:E-rep2}
\gls{tildeEOmega}=\boldsymbol{\Omega}^\mathbb{E}_{(0)}+\boldsymbol{\Omega}_{(1)}^\mathbb{E}+\boldsymbol{\Omega}_{(2)}^\E,
\end{equation}
where 
\begin{align}
\boldsymbol{\Omega}^\mathbb{E}_{(0)}&:=\boldsymbol{\Omega}^\mathbb{E}_{\de \Sigma},\\
\boldsymbol{\Omega}_{(1)}^\mathbb{E}&=\sum_{I,J,L} \int_{\de \Sigma} \frac{(-\I\hbar)^{|L| - |I|}}{(|L|+|J|)!}\de_LA^I(R,\mathsf{T}\varphi_x^*\Pi) [\E_I\E_J] \frac{\delta^{|L|+|J|}}{\delta[\E_ L]\delta[\E_J]}, \label{eq:E-rep2_1}\\
\boldsymbol{\Omega}_{(2)}^\E&= \sum_L \int_{\de \Sigma} \frac{(-\I\hbar)^{|L|+1}}{|L|!}\de_LF(R,R,\mathsf{T}\varphi_x^*\Pi)\frac{\delta^{|L|}}{\delta [\E_L]}, \label{eq:E-rep2_2}
\end{align}
where $A^I$ denotes the sum of weights of graphs with a single boundary vertex, where the incoming arrows at the boundary vertex are labeled by $I$, and $F$ denotes the sum of weights of graphs with no boundary vertices.
\end{prop}

\begin{rem} 
Recall that in the globalization of the Poisson Sigma Model after \cite{CF3}, briefly reviewed in Appendix \ref{app:globalization}, the choice of a formal exponential map on $\mathscr{P}$ induces a Fedosov connection $\gls{calD_G} = \dd + \gls{A}$ on the bundle of $\star$-algebras given by applying Kontsevich formality for $(T_x\mathscr{P},\mathsf{T}\varphi_x^*\mathscr{P})$ for every $x \in \mathscr{P}$. Here $\mathcal{D}_\mathsf{G}$ arises by ``quantizing'' the Grothendieck connection $D_\mathsf{G}$. 
In particular, the graphs appearing in $\boldsymbol{\Omega}^\E_{(1)}$ are exactly the ones appearing in the definition of the connection $1$-form $A$ as in \eqref{eq:defn_A}. 
The connection $\mathcal{D}_\mathsf{G}$ is not flat, $(\mathcal{D}_\mathsf{G})^2\sigma = [F,\sigma]_\star$. The graphs appearing in $\boldsymbol{\Omega}^\E_{(2)}$ are exactly the ones appearing in the Definition of the curvature $2$-form $\gls{F}$ as in \eqref{eq:defn_F}. Note that, by the notation as before, we can also write 
\begin{align}
\label{boundary_op_exp_con}
\boldsymbol{\Omega}_{(1)}^\mathbb{E}&=\left(A(R,\mathsf{T}\varphi^*_x\Pi)(\ee^{\frac{\I}{\hbar}[\E]y})\right)\Big|_{y=\frac{\delta}{\delta[\E]}},\\
\label{boundary_op_exp_curv}
\boldsymbol{\Omega}_{(2)}^\E&=F(R,R,\mathsf{T}\varphi_x^*\Pi)\left(\frac{\delta}{\delta[\E]}\right).
\end{align}
\end{rem}

\begin{proof}[Proof of Proposition \ref{E_rep2}]
We have seen that degree counting implies that there are exactly two boundary vertices in a collapsing graph. Now we have to take the $R$ vertices into account. Consider a collapsing graph with $n$ bulk and $k$ boundary vertices. Then the dimension of the corresponding configuration space is $2n + k -2$. On the other hand, there are now two types of bulk vertices: Suppose there are $m$ vertices labeled by the Poisson bivector field (emitting two arrows) and $r$ vertices labeled by the vector field $R$ (emitting one arrow). Since arrows cannot leave the collapsing graph, the total form degree is $2m + r$, which has to be equal to $2n + k -2$. Since $n = m+r$, this implies that $r + k = 2$. This means there can be either zero, one or two vertices labeled by $R$ with two, one or zero boundary vertices respectively, as shown in Figure \ref{fig:Erepresentationgraphs}. 
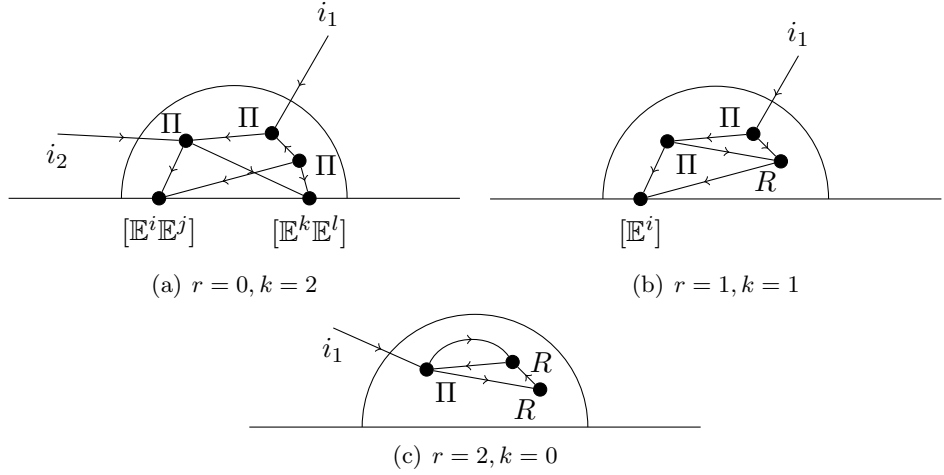
\begin{figure}[ht]
\subfigure[$r=0, k=2$]{
\begin{tikzpicture}

\draw (-3,0) -- (3,0); 
\node[vertex] (bdry1) at (-1,0) {};
\node[coordinate, label=below:{$[\mathbb{E}^{i}\mathbb{E}^{j}]$}] at (bdry1.south) {};
\node[vertex] (bdry2) at (1,0) {};
\node[coordinate, label=below:{$[\mathbb{E}^{k}\mathbb{E}^{l}]$}] at (bdry2.south) {};

\node[vertex] (bulk1) at (30:1) {};
 \node[] (r1) at (20:1.3) {$\Pi$};
\node[vertex] (bulk2) at (60:1) {};
\node[] (r2) at (80:1.1) {$\Pi$};
\node[vertex] (bulk3) at (130:1) {};
\node[] (p1) at (130:1.3) {$\Pi$};


\node[coordinate, label=above:{$i_1$}] (b2) at (60:2.5) {$i_2$};
\node[coordinate, label=below:{$i_2$}] (b3) at (160:2.5) {$i_3$};
\draw[] (0:1.5) arc (0:180:1.5); 
\draw[fermion] (b2) -- (bulk2);
\draw[fermion] (b3) -- (bulk3);
\draw[fermion] (bulk1) -- (bulk2);   
\draw[fermion] (bulk2) -- (bulk3); 
\draw[fermion] (bulk3) -- (bdry1); 
\draw[fermion] (bulk1) -- (bdry2);
\draw[fermion] (bulk1) -- (bdry1);
\draw[fermion] (bulk3) -- (bdry2);
\end{tikzpicture} 
}
\subfigure[$r=1,k=1$]{
\begin{tikzpicture}

\draw (-3,0) -- (3,0); 
\node[vertex] (bdry1) at (-1,0) {};
\node[coordinate, label=below:{$[\mathbb{E}^{i}]$}] at (bdry1.south) {};


\node[vertex] (bulk1) at (30:1) {};
 \node[] (r1) at (20:0.7) {$R$};
\node[vertex] (bulk2) at (60:1) {};
\node[] (r2) at (80:1.1) {$\Pi$};
\node[vertex] (bulk3) at (130:1) {};
\node[] (p1) at (130:0.6) {$\Pi$};


\node[coordinate, label=above:{$i_1$}] (b2) at (60:2.2) {$i_2$};
\draw[] (0:1.5) arc (0:180:1.5); 
\draw[fermion] (b2) -- (bulk2);
\draw[fermion] (bulk2) -- (bulk1);   
\draw[fermion] (bulk3) -- (bulk1);
\draw[fermion] (bulk2) -- (bulk3); 
\draw[fermion] (bulk3) -- (bdry1); 
\draw[fermion] (bulk1) -- (bdry1);
\end{tikzpicture} 
}
\subfigure[$r=2,k=0$]{
\begin{tikzpicture}

\draw (-3,0) -- (3,0); 


\node[vertex] (bulk1) at (30:1) {};
 \node[] (r1) at (20:0.7) {$R$};
\node[vertex] (bulk2) at (60:1) {};
\node[right] (r2) at (bulk2.east) {$R$};
\node[vertex] (bulk3) at (130:1) {};
\node[] (p1) at (130:0.6) {$\Pi$};


\node[coordinate, label=below:{$i_1$}] (b3) at (145:2.3) {$i_3$};
\draw[] (0:1.5) arc (0:180:1.5); 
\draw[fermion] (b3) -- (bulk3);
\draw[fermion] (bulk1) -- (bulk2);  
\draw[fermion] (bulk3) to [bend left=60] (bulk2);
\draw[fermion] (bulk2) -- (bulk3); 
\draw[fermion] (bulk3) -- (bulk1);
\end{tikzpicture} 
}
\caption{Possible graphs in the $\E$-representation. }\label{fig:Erepresentationgraphs}
\end{figure}

The first contribution $r=0$ and $k=2$ is exactly the operator $\boldsymbol{\Omega}^\E_{(0)}$ given in \eqref{eq:E-rep} from the non-globalized case. We get graphs with exactly one boundary vertex labeled by $R$ and graphs with exactly two boundary vertices labeled by $R$.

In the case $r=1$ and $k=1$ we obtain precisely the graphs with a single boundary vertex and a single $R$ bulk vertex (there can be an arbitrary number of vertices labeled by $\Pi$). This proves Equation \eqref{eq:E-rep2_1}.
In the case $r=2$ and $k=0$ we obtain Equation \eqref{eq:E-rep2_2}. 
\end{proof}
\subsubsection{$\bbX$-representation}
Next, we consider the $\mathbb{X}$-representation.
\begin{prop}
\label{X_rep2}
In the $\mathbb{X}$-representation, the globalized boundary operator is given by
\begin{equation}
\label{eq:X-rep2}
\gls{tildeXOmega}=\sum_{j=0}^{\dim(\mathscr{P})}\boldsymbol{\Omega}_{(j)}^\mathbb{X},
\end{equation}
where $\boldsymbol{\Omega}^\mathbb{X}_{(0)}:=\boldsymbol{\Omega}_{\de \Sigma}^\mathbb{X}$ and $\boldsymbol{\Omega}_{(j)}^\mathbb{X}$ is the sum of all graphs with $j$ vertices labeled by $R$ for $1\leq j\leq \dim(\mathscr{P})$.
\end{prop}

\begin{proof}
In the $\bbX$-representation, arrows can leave the collapsing graph, so we cannot do a degree count like in the $\E$-representation; in particular, the number of $R$ vertices in a collapsing graph is only bounded by the dimension of $\mathscr{P}$.
\end{proof}

\subsection{Algebraic structure in the flatness conditon for the BFV operator}
We know from \cite{CMW4} that $(\qtconn)^2 = 0$, and that this is equivalent to $\dd_x\Tilde{\boldsymbol{\Omega}}_{\de \Sigma}+\frac{1}{2}[\Tilde{\boldsymbol{\Omega}}_{\de \Sigma},\Tilde{\boldsymbol{\Omega}}_{\de\Sigma}]=0$. For the Poisson Sigma Model it is interesting to see how this condition can be derived by looking at the explicit structure of $\Tilde{\boldsymbol{\Omega}}_{\de\Sigma}$ as discussed in \ref{subsec:BFV_op}. We again consider the two different representations separately.

\subsubsection{$\E$-represention}
Recall that 
\begin{equation}
\Tilde{\boldsymbol{\Omega}}^\mathbb{E}_{\de\Sigma}=\boldsymbol{\Omega}^\mathbb{E}_{(0)}+\boldsymbol{\Omega}_{(1)}^\mathbb{E}+\boldsymbol{\Omega}_{(2)}^\E,
\end{equation}
where $\boldsymbol{\Omega}^\mathbb{E}_{(j)}$ denotes the part of form degree $j$ for $j\in\{0,1,2\}$. 
\begin{prop}
\label{alg_str}
We have the following equations:
\begin{align}
[\boldsymbol{\Omega}^\E_{(0)},\boldsymbol{\Omega}^\E_{(0)}]&=0,\label{deg0}\\
\dd_x\boldsymbol{\Omega}^\E_{(0)}+[\boldsymbol{\Omega}^\E_{(0)},\boldsymbol{\Omega}^\E_{(1)}]&=0,\label{deg1}\\
\dd_x\boldsymbol{\Omega}^\E_{(1)}+[\boldsymbol{\Omega}^\E_{(0)},\boldsymbol{\Omega}^\E_{(2)}]+\frac12[\boldsymbol{\Omega}^\E_{(1)},\boldsymbol{\Omega}^\E_{(1)}]&=0,\label{deg2}\\
\dd_x\boldsymbol{\Omega}^\E_{(2)}+[\boldsymbol{\Omega}^\E_{(1)},\boldsymbol{\Omega}^\E_{(2)}]&=0,\label{deg3}\\
[\boldsymbol{\Omega}^\E_{(2)},\boldsymbol{\Omega}^\E_{(2)}]&=0\label{deg4},
\end{align}
\end{prop}
\begin{proof}
Proposition \ref{alg_str} follows from general arguments in \cite{CMW4}, but here we give an independent proof. First we look at Equation \eqref{deg1}.

The construction in \cite{CF3}, recalled in Appendix \ref{app:PSM}, yields a bundle $\mathcal{E} = \widehat{S}T^*\mathscr{P}[[\hbar]]$ of $\star$-algebras on $\mathscr{P}$ by applying Kontsevich's deformation quantization in every tangent space. Picking a Grothendieck connection $D_\mathsf{G} = \dd_x + R$ on $\mathscr{P}$, and applying the Kontsevich formality map to $R$, one obtains a connection $\mathcal{D}_\mathsf{G} = \dd_x + A$ on $\mathcal{E}$. In \cite{CF3} it is shown that this connection is a derivation of $\Gamma(\mathcal{E})$, i.e. for $\sigma, \tau \in \Gamma(\mathcal{E})$, we have \begin{equation} \mathcal{D}_\mathsf{G}(\sigma \star \tau) = (\mathcal{D}_\mathsf{G}\sigma) \star \tau + \sigma \star (\mathcal{D}_\mathsf{G}\tau).\label{eq:der_in_proof} \end{equation}
We claim that this equation is equivalent to \eqref{deg1}. This can be done directly by writing out \eqref{deg1} and \eqref{eq:der_in_proof} in coefficients, but it is best seen through Feynman diagrams (after all, $A$ and the star product are defined through Feynman diagrams). First, rewrite \eqref{eq:der_in_proof} into 
\begin{equation}
\dd_x (\sigma \star \tau) - \dd_x\sigma \star \tau - \sigma \star \dd_x\tau = -A(\sigma \star \tau) + (A\sigma) \star \tau + \sigma \star (A\tau).
\end{equation}
The left hand side of this equation is given by applying $\dd_x$ to the coefficients of the star product. Schematically, we represent the diagrammatic content as in Figure \ref{fig:proof_Omega_E}. 
\begin{figure}
\begin{tikzpicture}[scale=0.75]
\begin{scope}[shift={(-10,0)}]
\draw (-2,0) -- (2,0);
\node[vertex] (bdry1) at (-1,0) {};
\node[vertex] (bdry2) at (1,0) {};
\node[coordinate, label=below:{$\sigma$}] at (bdry1.south) {};
\node[coordinate, label=below:{$\tau$}] at (bdry2.south) {};
\node[coordinate, label=below:{$\dd_x$}] at (0,1.5) {};
\draw (0:2) arc (0:180:2);  
\node[coordinate, label={$=$}] at (2.5,0.75) {};
\end{scope}
\begin{scope}[shift={(-5,0)}]
\draw (-2,0) -- (2,0);
\node[vertex] (bdry1) at (-0.75,0) {};
\node[vertex] (bdry2) at (0.75,0) {};
\node[coordinate, label=below:{$\sigma$}] at (bdry1.south) {};
\node[coordinate, label=below:{$\tau$}] at (bdry2.south) {};
\node[coordinate, label={$R$}] at (0,1.25) {};
\draw (0:2) arc (0:180:2);  
\draw (0:1) arc (0:180:1);
\node[coordinate, label={$+$}] at (2.5,0.75) {};
\end{scope}
\draw (-2,0) -- (2,0);
\node[vertex] (bdry1) at (-0.75,0) {};
\node[vertex] (bdry2) at (0.75,0) {};
\node[coordinate, label=below:{$\sigma$}] at (bdry1.south) {};
\node[coordinate, label=below:{$\tau$}] at (bdry2.south) {};
\node[coordinate, label=above:{$R$}] at (bdry1.north) {};
\draw (0:2) arc (0:180:2);  
\draw[shift={(-0.75,0)}] (0:1) arc (0:180:1);
\node[coordinate, label={$+$}] at (2.5,0.75) {};
\begin{scope}[shift={(5,0)}]
\draw (-2,0) -- (2,0);
\node[vertex] (bdry1) at (-0.75,0) {};
\node[vertex] (bdry2) at (0.75,0) {};
\node[coordinate, label=below:{$\sigma$}] at (bdry1.south) {};
\node[coordinate, label=below:{$\tau$}] at (bdry2.south) {};
\node[coordinate, label=above:{$R$}] at (bdry2.north) {};
\draw (0:2) arc (0:180:2);  
\draw[shift={(0.75,0)}] (0:1) arc (0:180:1);
\end{scope}
\end{tikzpicture} 
\caption{Schematics of the diagrammatic content of \eqref{eq:der_in_proof}. $\sigma$ and $\tau$ are arbitrary sections of $\Gamma(\mathcal{E})$. We sum over all possible graphs. By $\dd_x$ we mean that we apply $\dd_x$ to the result of the integration. An $R$ means that there is precisely one vertex labeled by $R$ in every graph.  }\label{fig:proof_Omega_E}
\end{figure}
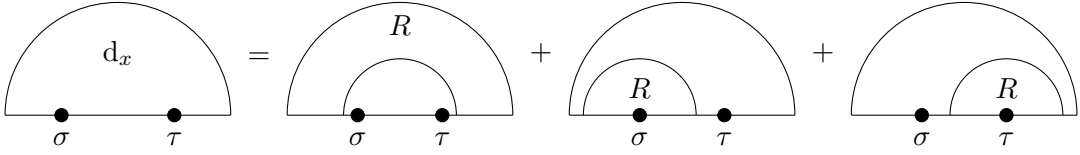 

Recall from \cite{CMW4} that the commutator $[\boldsymbol{\Omega}^\E_{(0)},\boldsymbol{\Omega}^\E_{(1)}]$ can be expressed by replacing the boundary vertices in the graphs defining $\boldsymbol{\Omega}^\E_{(1)}$ by the graphs appearing in $\boldsymbol{\Omega}^\E_{(0)}$ and vice versa. If we ignore possible arrows arriving at the boundary vertices from outside the graph, this yields precisely the graphs on the right hand side of Figure \ref{fig:proof_Omega_E}: The first term are the graphs of $\boldsymbol{\Omega}^\E_{(0)}$ placed at the boundary vertex of graphs appearing in  $\boldsymbol{\Omega}^\E_{(1)}$, and the second and the third term represent the graphs of $\boldsymbol{\Omega}^\E_{(1)}$ placed at one of the boundary vertices of $\boldsymbol{\Omega}^\E_{(0)}$. Arriving arrows from outside the graph corresponds to taking derivatives of $\sigma$ and $\tau$. 
On the other hand, the left hand side yields precisely $\dd_x\boldsymbol{\Omega}^\E_{(0)}$. \\
The other equations are proven in a similar fashion, using the following relations:
\begin{itemize}
\item{Equation \eqref{deg0} holds, since the non-globalized boundary operator squares to zero (which is in turn a consequence of the Classical Master Equation, see \cite{CMR2} and \cite{CMW4}).
}
\item{Equation \eqref{deg4} holds, since there are no $\E$-field contributions in $\boldsymbol{\Omega}^\E_{(2)}$.
}
\item{Equation \eqref{deg3} corresponds to the \textsf{Bianchi identity}.}
\item{Equation \eqref{deg2} corresponds to the equation $\dd_xA+\frac{1}{2}[A,A]=[F,\enspace]_\star$.}
\end{itemize}
In the last two points, the proof is similar to the proof of the degree 1 case \eqref{deg1}.
\end{proof}

\begin{rem}
\label{rem:globalization}
Note that the fact that the non-globalized BFV operator $\boldsymbol{\Omega}_{\de\Sigma}$ (which depends on the constant background field $x\colon \Sigma\to\mathscr{P}$) gives rise to a family of star products, constructed from the Poisson structure $\mathsf{T}\varphi_x^*\Pi$ on $T_x\mathscr{P}$. Moreover, the fact that it squares to zero corresponds to the associativity of these star products. Similarly the globalized BFV operator contains the data of a connection and its curvature and the fact that it is flat corresponds to the structural equations relating these objects. Hence we naturally recover the construction of the globalized version of Kontsevich's star product as it was discussed in \cite{CF3}. In \cite{CF3} the connection was twisted by a $1$-form $\gls{gamma}$ with values in the deformed jet bundle of $\star$-algebras to obtain a flat connection $\gls{barD_G}$, i.e. we have the following chain (see also Appendix \ref{app:globalization})
\begin{equation}
D_\mathsf{G}\xrightarrow{\Pi}\calD_\mathsf{G}\xrightarrow{\gamma}\overline{\calD}_\mathsf{G}.
\end{equation}
This motivates the introduction of an additional term $\calS_{\Sigma,\gamma}$ in the action to obtain $\Tilde{\boldsymbol{\Omega}}_{\de\Sigma}$ corresponding to the connection $\overline{\calD}_\mathsf{G}$ (see Section \ref{subsec:mod_action}).
\end{rem}

\subsubsection{$\bbX$-representation} In the $\bbX$-representation, one can similarly decompose the boundary operator into form degrees $\Tilde{\boldsymbol{\Omega}}^\mathbb{X}_{\de\Sigma} = \sum_{j=0}^{\dim \mathscr{P}} \boldsymbol{\Omega}^\bbX_{(j)}$, and for every $k=0,\ldots,r$ one obtains the equations 
\begin{equation}
\dd_x\boldsymbol{\Omega}^\bbX_{(k-1)} + \frac12\sum_{i+j=k} [\boldsymbol{\Omega}^\bbX_{(i)},\boldsymbol{\Omega}^\bbX_{(j)}]=0.
\end{equation} 
The form degree zero part is again the fact that the non-globalized boundary operator squares to zero. It would be interesting to investigate whether there is an algebraic structure underlying the equations in the other form degrees, similar to the $\E$-representation. 

\subsection{Modification of the action}
\label{subsec:mod_action}
We modify the classical BV action by using results of \cite{BCM,CF3,CFT} as we also discuss in Appendix \ref{app:PSM}.  Let $\gls{calE}:=\Hat{S}T^*\mathscr{P}[[\varepsilon]]$ for some deformation parameter $\gls{epsilon}$.
Recall from Appendix \ref{app:globalization}, that given $\gls{omega}\in\Omega^2(\mathscr{P},\calE)$ such that $\calD_\mathsf{G}\omega=0$ and $[\omega,\enspace]_\star=0$, we can always find $\gamma\in\Omega^1(\mathscr{P},\calE)$ such that 
\begin{equation}
\overline{F}^\mathscr{P}_\omega:=F^\mathscr{P}+\varepsilon\omega+\mathcal{D}_\mathsf{G}\gamma+\gamma\star\gamma=0.
\end{equation}
This is equivalent to equation \eqref{gen_curv}.\\
According to Remark \ref{rem:globalization}, we now formulate a new ``modified'' action.
\begin{defn}[Modified formal globalized action]
Let $\gamma \in \Omega^1(\mathscr{P},\mathcal{E})$ be a solution of equation $\eqref{gen_curv}$ for $\omega \in \Omega^2(\mathscr{P},\mathcal{E})$ as above (here the formal parameter $\varepsilon$ is given by $(-\I\hbar)/2$).
Then the \textsf{modified formal globalized action} $\mathbb{S}_{\Sigma,x}$ is given by 
\begin{equation}
\label{whole_action}
\gls{modifiedformalglobalizedaction}=\Tilde{\calS}_{\Sigma,x}+\calS_{\Sigma,\gamma}+\calS_{\Sigma,\omega},
\end{equation}
where\footnote{The reason why such counter terms always exists is due to the fact that the cohomology which would provide obstructions is trivial \cite{CF3} (see also Remark \ref{rem:counterterm}).} 
\begin{align}
\label{Sgamma}
\gls{gammaaction}&=\int_{\partial\Sigma}\gamma\left(x;\widehat{\sfX}\right) =\sum_{k\geq 1}(-1)^k\left(\frac{\I\hbar}{2}\right)^k\sum_{I} \dd x^i\int_{\de\Sigma}\gamma^{(k)}_{i,I}(x)\hatX^I,\\
\label{Somega}
\gls{omegaaction}&=\int_{\Sigma}\omega\left(x;\hatX\right)=\sum_{k\geq 1}(-1)^k\left(\frac{\I\hbar}{2}\right)^k\sum_J\dd x^i\land\dd x^j\int_{ \Sigma}\omega^{(k)}_{ij,J}(x)\hatX^J.
\end{align}
\end{defn}
\begin{rem}
Here we integrate the source $1$-form part of $\hatX$ along the boundary, which, since the $\hatX$-fluctuation vanishes on components of the boundary in $\bbX$-representation, implies that for a single boundary with $\bbX$-representation $\calS_{\Sigma,\gamma}$ does not give any contribution to $\Tilde{\boldsymbol{\Omega}}^\bbX_{\de \Sigma}$. Therefore we only need to look at the $\E$-representation.
Moreover, note that $\gamma = O(\hbar)$, i.e. it is already a type of quantum counterterm which is not present classically, so it does not violate the modified Classical Master Equation. 
\end{rem}

\begin{prop}
The BFV boundary operator $\Tilde{\boldsymbol{\Omega}}^{\E,\gamma}_{\de \Sigma}$ for the modified formal globalized action $\mathbb{S}_{\Sigma,x}$ is given by 
\begin{equation}
\Tilde{\boldsymbol{\Omega}}^{\E,\gamma}_{\de \Sigma}=\boldsymbol{\Omega}^\E_{\de\Sigma}+\boldsymbol{\Omega}^\E_{(1)}+\left([\ee^{\frac{\I}{\hbar}[\E]y},\gamma]_\star\right)\Big|_{y=\frac{\delta}{\delta[\E]}},
\end{equation}
where $\star$ denotes again the fiberwise star product on $\calE$ as in \ref{subsec:operator}.
\end{prop}

\begin{proof}

Considering again a degree counting, we get different cases of boundary vertex configurations. For the case $(r=0,k=2)$, we can either have two $\E$-field boundary vertices, one $\E$-field and one $\gamma$ boundary vertex or two $\gamma$ boundary vertices. For the case $(r=1,k=1)$, we can have either one $\E$-field boundary vertex or one $\gamma$ boundary vertex. For the case $(r=2,k=0)$ we have the same contribution as before. In the case $\omega \neq 0$, there is a configuration where $(r=k=0)$, but there is a single $\omega$ vertex. These different diagrams contribute to different terms for the new boundary operator, which are:
\begin{itemize}
\item{$r=0,k=2$ ($\E,\E$ on the boundary): Summing over all these graphs, this corresponds to the term 
\begin{equation}
\label{term1}
\boldsymbol{\Omega}^\E_{\de\Sigma},
\end{equation}
}
\item{$r=0,k=2$ ($\gamma,\gamma$ on the boundary): Summing over all these graphs, this corresponds to 
\begin{equation}
\label{term2}
(\gamma\star\gamma)\left(\frac{\delta}{\delta[\E]}\right),
\end{equation}
}
\item{$r=0,k=2$ ($\E,\gamma$ on the boundary): Summing over all these graphs, this corresponds to 
\begin{equation}
\label{term3}
\left([\ee^{\frac{\I}{\hbar}[\E]y},\gamma]_\star\right)\Big|_{y=\frac{\delta}{\delta[\E]}},
\end{equation}
}
\item{$r=1,k=1$ ($\E$ on the boundary): Summing over all these graphs, this corresponds to the term 
\begin{equation}
\label{term4}
\left(A(R,\mathsf{T}\varphi^*_x\Pi)(\ee^{\frac{\I}{\hbar}[\E]y})\right)\Big|_{y=\frac{\delta}{\delta[\E]}},
\end{equation}
}
\item{$r=1,k=1$ ($\gamma$ on the boundary): Summing over all these graphs, this corresponds to the connection term 
\begin{equation}
\label{term5}
A(R,\mathsf{T}\varphi^*_x\Pi)(\gamma)=A(R,\mathsf{T}\varphi^*_x\Pi)\left(\gamma_i\left(\frac{\delta}{\delta[\E]}\right)\right)\dd x^{i},
\end{equation}
}
\item{$r=2,k=0$ (nothing on the boundary): Summing over all these graphs, this corresponds to the curvature term 
\begin{equation}
\label{term6}
F(R,R,\mathsf{T}\varphi_x^*\Pi),
\end{equation}
}
\item{$r=k=0$ (just one $\omega$ vertex in the bulk): Summing over all graphs this just yields $\omega$.}
\end{itemize}
By Equation \eqref{gen_curv} and \eqref{AppB:connection_G}, we obtain that the terms \eqref{term2}, \eqref{term3}, \eqref{term6} and possibly $\omega$, can be put together as 
\begin{equation}
A(R,\mathsf{T}\varphi^*_x\Pi)(\gamma)-F(R,R,\mathsf{T}\varphi^*_x\Pi)-\gamma\star\gamma-\omega=\dd_x\gamma.
\end{equation}
Hence they do not contribute to the boundary operator, since they cancel the terms in $\dd_x\btpsi_{\Sigma,x}$ in the modified differential Quantum Master Equation, where the full state is defined by the action $\mathbb{S}_{\Sigma,x}$.
\end{proof}

\begin{rem}
By equation \eqref{def_conn} and the fact that $\dd_x\ee^{\frac{\I}{\hbar}\E}=0$, the surviving terms will correspond to 
\begin{equation}
\left(\overline{\calD}_\mathsf{G}\ee^{\frac{\I}{\hbar}[\E]y}\right)\Big|_{y=\frac{\delta}{\delta[\E]}}=\left(\calD_\mathsf{G}\ee^{\frac{\I}{\hbar}[\E]y}\right)\Big|_{y=\frac{\delta}{\delta[\E]}}+\left([\ee^{\frac{\I}{\hbar}[\E]y},\gamma]_\star\right)\Big|_{y=\frac{\delta}{\delta[\E]}},
\end{equation}
where $\calD_\mathsf{G}\ee^{\frac{\I}{\hbar}[\E]y}$ means that we apply $\calD_\mathsf{G}$ to the fiber coordinates $y$ of $\ee^{\frac{\I}{\hbar}[\E]y}$.
Hence the boundary operator is given by 
\begin{equation}
\gls{twistedEOmega}= \boldsymbol{\Omega}^{\E}_{\de\Sigma}+\left(\overline{\calD}_\mathsf{G}\ee^{\frac{\I}{\hbar}[\E]y}\right)\Big|_{y=\frac{\delta}{\delta[\E]}}.
\end{equation}
\end{rem}

\subsubsection{The twisted state}

Using the modified action \eqref{whole_action} one can define a state twisted by $\gamma$ as follows. 
\begin{defn}[Twisted full covariant quantum state]
\label{full_covariant_state}
Let $\Sigma$ be a manifold, possibly with boundary. 
Given a $BF$-like BV-BFV theory $\pi_\Sigma\colon\calF_\Sigma \to \calF^\de_{\de \Sigma}$, a polarization $\calP$ on $\calF^\de_{\de \Sigma}$, a splitting $\calF_\Sigma \cong \mathcal{B}^\calP_{\de \Sigma}\oplus \calV^\calP_\Sigma \oplus \calY'$, and a gauge-fixing Lagrangian $\calL_{\Sigma,x} \subset \calY'$, we define the \textsf{twisted full covariant quantum state} by the formal perturbative expansion of the BV integral
\begin{equation}
\label{full_state1}
\btpsi^\gamma_{\Sigma,x}(\mathbb{X},\E;\mathsf{x},\mathsf{e}):=\int_{\calL_{\Sigma,x}\subset\calY'}\ee^{\frac{\I}{\hbar}\mathbb{S}_{\Sigma,x}[(\Hat{\mathsf{X}},\Hat{\boldsymbol{\eta}})]}\in\Omega^\bullet(\Sigma,\Hat{\calH}^\calP_{\Sigma,\textnormal{tot}}),
\end{equation}
using the Feynman rules in Figure \ref{fig:FeynmanRules0} and additionally with the rules for the boundary vertices as in Figure \ref{fig:composite_field_vertices} and Figure \ref{fig:new_vert}.
\end{defn}

The reason to introduce this state will become clear in the next two sections, when we analyze the anomaly arising from alternating boundary conditions. Essentially, the twist localizes the anomaly at the corners (Theorem \ref{thm_corners}), where it can be canceled by changing the boundary operator (Theorem \ref{thm_mdQME_corners}).

\begin{figure}[ht]
\centering
\subfigure[New interaction vertices in the bulk representing $\omega$]{
\centering
\begin{tikzpicture}
\node[vertex, label=below:{$\omega$}] (o) at (0,0) {};
\node[coordinate, label=below:{$i_1$}] at (30:1) {$i_1$}
edge[fermion] (o);
\node[coordinate, label=above:{$i_2$}] at (60:1) {$i_2$}
edge[fermion] (o);
\node[coordinate, label=below:{$i_s$}] at (145:1) {$i_k$}
edge[fermion] (o);
\draw[dotted] (90:0.5) arc (90:130:0.5);

\node[coordinate,label=right:{$\rightsquigarrow \sum_{k\geq 1}(-1)^k\left(\frac{\I\hbar}{2}\right)^k\dd x^i\land\dd x^j\omega^{(k)}_{ij,i_1\cdots i_k}(x)$}] at (2,0) {};
\end{tikzpicture}
}
\hspace{3cm}
\subfigure[New boundary vertices representing $\gamma$]{
\begin{tikzpicture}
\draw (1,-1) -- (3,-1);
\node[vertex, label=below:{$\gamma$}] (e) at (2,-1) {};
\node[coordinate, label=right:{$i_2$}] (b2) at (2,0) {}; 
\node[coordinate,label=right:{$i_1$}] (b3) at (2.7,0) {};
\node[coordinate,label=left:{$i_k$}] (b4) at (1,-0.5) {};
\draw[fermion] (b2) -- (e);
\draw[fermion] (b3) -- (e);
\draw[fermion] (b4) -- (e);
\draw[dotted] (2,-0.3) arc (90:150:0.5);

\node[coordinate,label=right:{$\rightsquigarrow \sum_{k\geq 1}(-1)^k\left(\frac{\I\hbar}{2}\right)^k\dd x^i\gamma^{(k)}_{i,i_1\cdots i_k}(x)$}] at (4,-1) {};
\end{tikzpicture}
\hspace{2cm}
}
\caption{New vertices appearing in the Feynman rules}
\label{fig:new_vert}
\end{figure}
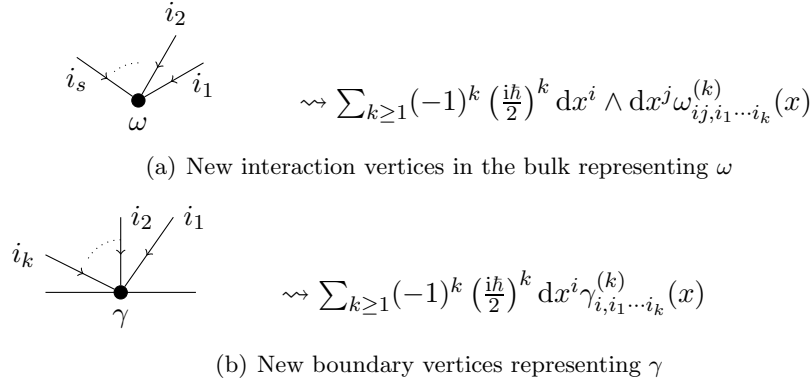

The twisted state is closed with respect to the operator
\begin{equation}
\label{nabla_gamma}
\nabla_{\mathsf{G}}^\gamma:=\left(\dd_x-\I\hbar\Delta_{{\calV}_\Sigma}+\frac{\I}{\hbar}\Tilde{\boldsymbol{\Omega}}^{\E,\gamma}_{\de \Sigma}\right).
\end{equation}

This is a consequence of Theorem \ref{thm_corners} below. 

\subsubsection{Flatness}
The following Proposition tells us that the twisted quantum Grothendieck BFV operator still remains a differential, i.e. squares to zero for $\Tilde{\boldsymbol{\Omega}}^{\E,\gamma}_{\de \Sigma}$.
\begin{prop}\label{prop:flatness_twisted}
The operator 
\begin{equation}
\nabla_{\mathsf{G}}^\gamma =\left(\dd_x-\I\hbar\Delta_{{\calV}_\Sigma}+\frac{\I}{\hbar}\Tilde{\boldsymbol{\Omega}}^{\E,\gamma}_{\de \Sigma}\right) 
\end{equation}
on $\Hat{\calH}_{\Sigma,\textnormal{tot}}$ squares to zero.
\end{prop}

\begin{proof}
Note that the flatness condition of $\qtconn^{\gamma}$, is equivalent to
\begin{equation}
\dd_x\Tilde{\boldsymbol{\Omega}}^{\E,\gamma}_{\de \Sigma}+ \frac{1}{2}\left[\Tilde{\boldsymbol{\Omega}}^{\E,\gamma}_{\de \Sigma}, \Tilde{\boldsymbol{\Omega}}^{\E,\gamma}_{\de \Sigma}\right] = 0. 
\end{equation}
Separating the equation by form degree in $\mathscr{P}$ this is equivalent to 
\begin{align}
\frac{1}{2}\left[\boldsymbol{\Omega}^{\E}_{\de\Sigma}, \boldsymbol{\Omega}^{\E}_{\de\Sigma}\right] &= 0 \label{eq:Omegagammaflatness1} \\
\dd_x\boldsymbol{\Omega}^{\E}_{\de\Sigma} + \left[\boldsymbol{\Omega}^{\E}_{\de\Sigma},  \left(\overline{\calD}_\mathsf{G}\ee^{\frac{\I}{\hbar}[\E]y}\right)\Big|_{y=\frac{\delta}{\delta[\E]}}\right] &= 0 \label{eq:Omegagammaflatness2}\\
\dd_x\left(\left(\overline{\calD}_\mathsf{G}\ee^{\frac{\I}{\hbar}[\E]y}\right)\Big|_{y=\frac{\delta}{\delta[\E]}}\right) + \left[\left(\overline{\calD}_\mathsf{G}\ee^{\frac{\I}{\hbar}[\E]y}\right)\Big|_{y=\frac{\delta}{\delta[\E]}}, \left(\overline{\calD}_\mathsf{G}\ee^{\frac{\I}{\hbar}[\E]y}\right)\Big|_{y=\frac{\delta}{\delta[\E]}}\right] &= 0. \label{eq:Omegagammaflatness3}
\end{align}
Equation \eqref{eq:Omegagammaflatness1} is just saying that the standard BFV boundary operator squares to zero. Equation \eqref{eq:Omegagammaflatness3} is true because $\overline{\calD}_\mathsf{G}$ is a flat connection. Equation \eqref{eq:Omegagammaflatness2} means that $\boldsymbol{\Omega}^{\E}_{\de\Sigma}$ is a $\overline{\calD}_\mathsf{G}$-closed section. This comes from the fact that the coefficients of $\boldsymbol{\Omega}^{\E}_{\de\Sigma}$ are the same as in the star product. 
\end{proof}

\section{Alternating boundary conditions and the modified differential Quantum Master Equation}
\label{sec:ABC}
\subsection{Consistent boundary conditions}
In \cite{CF1} it was shown that the perturbative expansion of the QFT given from of the Poisson Sigma Model on the disk coincides with Kontsevich's star product, where we expand around the gauge equivalent classical solutions of the given Euler--Lagrange equations, which are $\mathsf{X}=x=const.$, $\boldsymbol{\eta}=0$ (recall Subsection \ref{subsec:pert_quant} and see also Appendix \ref{app:PSM}). The boundary conditions on the disk $D$ are exactly set such that $\boldsymbol{\eta}\vert_{\partial D}=0$ in order to be consistent with these types of solutions. 

\subsection{Construction of boundary conditions}
In \cite{CMW4}, the globalization construction was only considered for boundaries with a single polarization. We want to extend the methods developed in the previous sections following \cite{CMW4} to describe deformation quantization of the relational symplectic groupoid \cite{CF5,CC1,CC2} extending what we did in \cite{CMW2} in the case of a constant Poisson structure. This requires that we perform the BV-BFV quantization in the presence of ``alternating'' boundary conditions, which we can formulate for any worldsheet $\Sigma$: Let $\partial\Sigma=\bigsqcup_\ell\partial\Sigma^{(\ell)}$ and consider a partition into two distinguished components for every connected component $\partial\Sigma^{(\ell)}$ of the boundary given by $\partial\Sigma^{(\ell)}=\partial_0\Sigma^{(\ell)}\sqcup\partial_\calP\Sigma^{(\ell)}$. Each $\partial\Sigma^{(\ell)}$ is given as a disjoint union of an even number of intervals $I_1^{(\ell)},\ldots,I_n^{(\ell)}$, such that $\partial_0\Sigma^{(\ell)}=\bigsqcup_{\substack{1\leq j \leq n\\\text{$j$ odd}}}I_j^{(\ell)}$ and $\partial_\calP\Sigma^{(\ell)}=\bigsqcup_{\substack{1\leq j \leq n\\\text{$j$ even}}}I_j^{(\ell)}$.
 Now the alternating condition is that on components of $\partial_0\Sigma^{(\ell)}$ we set $\hateta= 0$, and on components of  $\partial_\calP\Sigma^{(\ell)}$  we choose some polarization $\mathcal{P}_j$ for each $I_j$, and consider the corresponding boundary fields. We think of the endpoints of the intervals as ``corners''. Moreover, we denote by $\de_1\Sigma$ the components of $\de_\calP\Sigma$ with the $\frac{\delta}{\delta\E}$-polarization and by $\de_2\Sigma$ the components of $\de_\calP\Sigma$ with the $\frac{\delta}{\delta\mathbb{X}}$-polarization.

\begin{rem}
The choice of polarization imposes boundary conditions on the fluctuations. The boundary conditions corresponding to our polarizations for split AKSZ theories are some generalization of Dirichlet and Neumann conditions. Note that, even if fixing a field (to zero, in the case of a fluctuation) on the boundary looks like a Dirichlet boundary condition, it may also be thought as a Neumann one, for our theory is of first order.
\end{rem}

\begin{figure}[ht]
\begingroup%
  \makeatletter%
  \providecommand\color[2][]{%
    \errmessage{(Inkscape) Color is used for the text in Inkscape, but the package 'color.sty' is not loaded}%
    \renewcommand\color[2][]{}%
  }%
  \providecommand\transparent[1]{%
    \errmessage{(Inkscape) Transparency is used (non-zero) for the text in Inkscape, but the package 'transparent.sty' is not loaded}%
    \renewcommand\transparent[1]{}%
  }%
  \providecommand\rotatebox[2]{#2}%
  \ifx\svgwidth\undefined%
    \setlength{\unitlength}{258.04146255bp}%
    \ifx\svgscale\undefined%
      \relax%
    \else%
      \setlength{\unitlength}{\unitlength * \real{\svgscale}}%
    \fi%
  \else%
    \setlength{\unitlength}{\svgwidth}%
  \fi%
  \global\let\svgwidth\undefined%
  \global\let\svgscale\undefined%
  \makeatother%
  \begin{picture}(1,0.51358826)%
    \put(0,0){\includegraphics[width=\unitlength]{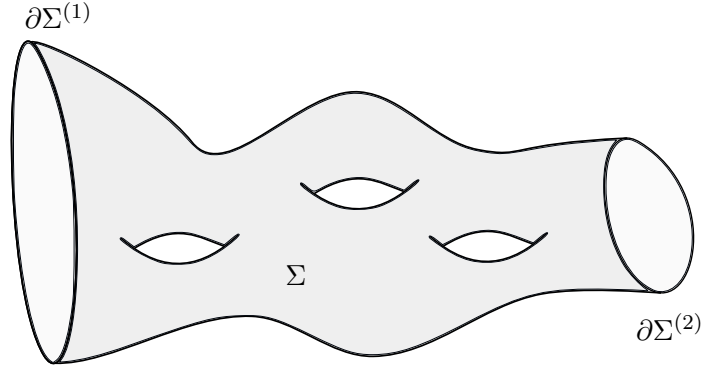}}%
    \put(0.02042569,0.4954046){\color[rgb]{0,0,0}\makebox(0,0)[lb]{\smash{$\de\Sigma^{(1)}$}}}%
    \put(0.91353114,0.04012224){\color[rgb]{0,0,0}\makebox(0,0)[lb]{\smash{$\de\Sigma^{(2)}$}}}%
    \put(0.40137113,0.12070307){\color[rgb]{0,0,0}\makebox(0,0)[lb]{\smash{$\Sigma$}}}%
  \end{picture}%
\endgroup%
%
%
%
\caption{Example of a worldsheet manifold $\Sigma$ with genus $g=3$ and two disjoint boundary components $\partial\Sigma^{(1)}$ and $\partial \Sigma^{(2)}$.}
\label{surface}
\end{figure}

\begin{figure}[ht]
\centering
\tikzset{
particle/.style={thick,draw=black},
particle2/.style={thick,draw=black, postaction={decorate},
    decoration={markings,mark=at position .9 with {\arrow[black]{triangle 45}}}},
gluon/.style={decorate, draw=black,
    decoration={coil,aspect=0}}
 }
\begin{tikzpicture}[x=0.04\textwidth, y=0.04\textwidth]
\node[](1) at (0,0){};
\node[](2) at (5,0){};
\node[](3) at (4,-2){};
\node[](8) at (1,2){};
\draw[fill=black] (3) circle (.07cm);
\node[](5) at (1,-2){};
\draw[fill=black] (5) circle (.07cm);
\node[](7) at (4,2){};
\draw[fill=black] (7) circle (.07cm);
\draw[fill=black] (8) circle (.07cm);
\node[](m) at (2.5,0){$\partial\Sigma^{(\ell)}$};
\node[](9) at (0,0){};
\draw[fill=black] (9) circle (.07cm);
\node[](10) at (5,0){};
\draw[fill=black] (10) circle (.07cm);
\node[](2) at (5,0){};
\node[](I_1) at (5.25,1.5){$I_1$};
\node[](I_2) at (5.25,-1.5){$I_2$};
\node[](I_3) at (2.5,-3){$I_3$};
\node[](I_4) at (-0.25,-1.5){$I_4$};
\node[](I_5) at (-0.25,1.5){$I_5$};
\node[](I_6) at (2.5,3){$I_6$};
\semiloop[particle]{1}{2}{0};
\semiloop[particle]{2}{1}{180};
\end{tikzpicture}
\caption{Example of a boundary component of $\Sigma$ as in Figure \ref{surface}, where the boundary $\partial \Sigma^{(\ell)}$ is split into $n=6$ disjoint components, i.e. $\partial \Sigma^{(\ell)}=\bigsqcup_{1\leq j\leq 6} I^{(\ell)}_j$ with $\partial_0\Sigma^{(\ell)}=\bigsqcup_{j\in\{1,3,5\}}I^{(\ell)}_j$ and $\partial_\calP\Sigma^{(\ell)}=\bigsqcup_{j\in\{2,4,6\}}I^{(\ell)}_j$. On $\de_0\Sigma^{(\ell)}$ we set $\hateta=0$. On $I^{(\ell)}_2,I^{(\ell)}_4,I^{(\ell)}_6$ we choose polarizations and take the corresponding boundary conditions. 
}
\end{figure}
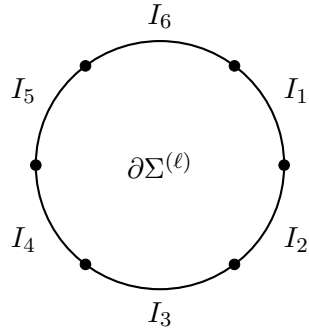

\subsection{Curvature Anomaly} 
Unlike in the constant case discussed in \cite{CMW2}, upon quantization\footnote{Note that we are not performing ``extended'' quantization of a manifold with corners in the sense of extended TQFTs, but simply apply BV-BFV quantization where we allow boundary conditions to change along connected components of the boundary.} the modified differential Quantum Master Equation fails to be satisfied. This effect arises from the curvature of the deformed Grothendieck connection.

\begin{prop}
\label{prop_curv}
Consider the full state $\btpsi_{\Sigma,x}$ defined by $\Tilde{\calS}_{\Sigma,x}$ as in Definition \ref{full_state_2}. Then 
\begin{equation}
\nabla_\mathsf{G}\btpsi_{\Sigma,x}=\exp\left(\frac{\I}{\hbar}\int_{\partial_0\Sigma}F(R,R,\mathsf{T}\varphi^*_x\Pi)(\mathscr{X})\right)\btpsi_{\Sigma,x},
\end{equation}
where we integrate the $\mathsf{X}$-fluctuation $\mathscr{X}$ in $F$ along $\de_0\Sigma$. Here $F$ denotes the curvature of the deformed Grothendieck connection $\calD_\mathsf{G}$ defined in Appendix \ref{app:globalization} and $\nabla_\mathsf{G}$ is the quantum Grothendieck BFV operator defined as in \eqref{AKSZ_mdQME}.
\end{prop}

\begin{proof}
If we try to prove the modified differential Quantum Master Equation as in \cite{CMW4}, when integrating over the boundary of the compactified configuration space there are strata where a bulk graph collapses at a point $u \in \de_0\Sigma$, i.e. one of the boundary components where $\hateta=0$. The degree count as we have seen before, shows that we will only end up with graphs without any boundary vertices and precisely two $R$ vertices in the bulk. Summing over all these graphs one obtains the curvature of the Grothendieck connection as in Appendix \ref{app:PSM}. However, since there are no boundary fields on $\de_0 \Sigma$, these terms cannot be cancelled by a term in the BFV boundary operator. 
\end{proof}

\begin{rem}
This can be interpreted as a quantum anomaly, since this problem is not present at the classical level.
To restore the modified differential Quantum Master Equation, we can add additional terms to the action, reminiscent to the addition of counterterms. This will yield new boundary terms, but they can be cancelled by adding appropriate terms to the BFV boundary operator as we have already seen in Subsection \ref{subsec:mod_action}, if we allow for a slight extension of the space of states (see Subsection \ref{subsec:ex_states}). 
\end{rem}


\subsection{New boundary contributions in the proof of the modified differential Quantum Master Equation} 
\label{subsec_boundary}

To cancel this anomaly we add quantum counterterms to the action, specifically, the terms $\calS_{\Sigma,\gamma}$  and $\calS_{\Sigma,\omega}$ defined in \eqref{Sgamma} and \eqref{Somega} respectively. The new terms in the action give rise to additional vertices. Namely, we now have vertices of arbitrary valence on components of the boundary where $\hatX \neq 0$, i.e. on the $\hateta=0$ boundary components and the components of $\de_\calP\Sigma$ in $\E$-representation. At such a vertex we place the corresponding derivative of $\gamma$ in the formal directions. 
Also, there are new bulk vertices labeled by $\omega$, which are similarly labeled by derivatives of $\omega$ in the formal directions. 

 Let $\gls{scrC}$ denote the set of all corner points of $\Sigma$. There are two types of corners: Let $\mathscr{C}_2\subset\mathscr{C}$ denote the subset containing those corner points which connect a $\frac{\delta}{\delta\mathbb{X}}$-polarized connected component (i.e. a component in $\mathbb{E}$-representation) of $\partial_0\Sigma$ with a connected component of $\partial_\calP\Sigma$ and let $\mathscr{C}_1\subset\mathscr{C}$ denote the subset containing those corner points which connect a $\frac{\delta}{\delta\mathbb{E}}$-polarized connected component of $\partial_0\Sigma$ with a connected component of $\partial_\calP\Sigma$. 

\begin{figure}[ht]
\subfigure[A corner in $\mathscr{C}_2$]{
\centering
\begin{tikzpicture}
\node[coordinate] (bdry1) at (0,2){};
\node[left] at (bdry1.west){$\widehat{\boldsymbol{\eta}}=0$};
\node[vertex] (corner) at (0,0){};
\node[below left] at (corner.west) {$C$};
\node[coordinate, below] (bdry2) at (2,0){};
\node[below] at (bdry2.south){$\mathbb{E}$}; 
\draw[thick] (bdry1.center) -- (corner) -- (bdry2.center);
\end{tikzpicture}
}
\hspace{2cm}
\subfigure[A corner in $\mathscr{C}_1$]{
\begin{tikzpicture}
\node[coordinate] (bdry1) at (0,-2){};
\node[left] at (bdry1.west){$\widehat{\boldsymbol{\eta}}=0$};
\node[vertex] (corner) at (0,0){};
\node[below left] at (corner.west) {$C$};
\node[coordinate, below] (bdry2) at (2,0){};
\node[below] at (bdry2.south){$\mathbb{X}$}; 
\draw[thick] (bdry1.center) -- (corner) -- (bdry2.center);
\end{tikzpicture}
}
\caption{The two types of corners.}
\end{figure}
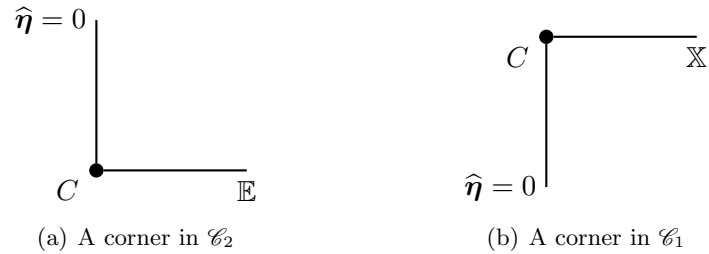

\begin{defn}[Twisted quantum Grothendieck BFV operator]\label{defn_twisted_conn}
We define the \textsf{twisted quantum Grothendieck BFV operator} by
\begin{equation}
\label{twisted_conn}
\gls{nabla^gamma_G}:=\dd_x-\I\hbar\Delta_{\calV_\Sigma}+\frac{\I}{\hbar}\underbrace{\left(\Tilde{\boldsymbol{\Omega}}_{\de_1\Sigma}^{\mathbb{X}}+\Tilde{\boldsymbol{\Omega}}_{\de_2\Sigma}^{\E,\gamma}\right)}_{=:\Tilde{\boldsymbol{\Omega}}_{\de\Sigma}^{\calP,\gamma}}.
\end{equation}
\end{defn}
\begin{rem}\label{rem_flatness_twisted}
The twisted quantum Grothendieck BFV operator is also a coboundary operator. This follows from Proposition \ref{prop:flatness_twisted} and the fact that $
\Tilde{\boldsymbol{\Omega}}_{\de_1\Sigma}^{\mathbb{X}}$ also squares to zero. 
\end{rem}

We can now state the main theorem of this section.
\begin{thm}
\label{thm_corners}
Consider the twisted full state $\btpsi_{\Sigma,x}^\gamma$ defined in Definition \ref{full_covariant_state} and the twisted quantum Grothendieck BFV operator $\nabla^\gamma_\mathsf{G}$ defined in Definition \ref{defn_twisted_conn}.
Then
\begin{equation}
\label{corner}
\nabla_{\mathsf{G}}^\gamma\btpsi_{\Sigma,x}^\gamma=\sum_{C\in\mathscr{C}_1}T(C)\btpsi_{\Sigma,x}^\gamma,
\end{equation}
where $T(C)$ are functionals on $\calB^\calP_{\de\Sigma}$ with values in $\Omega^1(\mathscr{P})$, depending only on the values of the fields at the corner point $C$.
\end{thm}

\begin{rem}
In particular, $T(C)$ are non-regular functionals in the sense of \ref{regular_fun}. In Section \ref{sec:mdQME}, we discuss an extension of the space of operators and states, which will allow us to rewrite equation \eqref{corner} as a closedness condition with respect to a differential on this extended space.
\end{rem}

\subsection{Proof of Theorem \ref{thm_corners}}
If we try to proceed with the proof of the modified differential Quantum Master Equation as in \cite{CMW4}, we get terms where a part of a graph collapses on $\de_0\Sigma$, i.e. the part of the boundary where $\Hat{\boldsymbol{\eta}}= 0$. We will now analyze these terms more closely. Let $\Gamma' \subset \Gamma$ be a subgraph that collapses on a point of the boundary, and denote by $\Gamma/\Gamma'$ the resulting graph. Suppose $\Gamma'$ has $n$ bulk and $k$ boundary vertices on $\de_0\Sigma$. Then the dimension of the corresponding boundary stratum is $2n+k-2$ as we have seen before. The contribution of the graph is non-vanishing only if the form degree of $\omega_{\Gamma'}$ is also $2n+k-2$. The bulk vertices correspond to either $\Pi$ or $R$, the former has two outgoing arrows, the latter only one. If one of these arrows points out of $\Gamma'$, then $\omega_{\Gamma/\Gamma'} = 0$, since it contains a propagator with the tail evaluated on the $\Hat{\boldsymbol{\eta}}=0$ boundary component. Hence all these arrows must point to another vertex in $\Gamma'$. Suppose there are $m$ vertices with two outgoing arrows and $r$ vertices with one outgoing arrow. Then we must have the following system of equations:
\begin{align}
\label{system1}
2n+k-2&=2m +r \\
n&=m+r,
\end{align}
which is equivalent to $r = 2-k$ ($m$ is arbitrary, and $n = m+r$). Since $r \geq 0$, we conclude that $k$ is either $0,1$, or $2$. Let us analyze these possibilities in more detail. 
\subsubsection{Terms with $k=0$}
\label{sec:k0}
In these terms there are no boundary vertices. They are also present if we do not add $\mathcal{S}_{\Sigma,\gamma}$ to the action. We have $r = 2-k = 2$, so these terms are given by graphs with $R$ at two vertices. Summing over all these terms yields the curvature of the Grothendieck connection, $F$ (again, see Appendix \ref{app:PSM} for details). 

This is what spoils the modified differential Quantum Master Equation, since we cannot cancel it with terms in the BFV boundary operator, which can only cancel the boundary contributions on boundary components with free boundary fields. We are thus forced to add other terms to the action to cancel the appearance. 
\subsubsection{Terms with $k=1$}
In these terms there is one boundary vertex labeled by $\gamma$, and one bulk vertex labeled by the vector field $R$. If we sum over all such graphs, we get 
\begin{equation}
A(R,\mathsf{T}\varphi_x^*\Pi)(\gamma) = A(R,\mathsf{T}\varphi_x^*\Pi)(\gamma_i)\dd x^i
\end{equation}
by the Definition of $A$ as in Appendix \ref{app:PSM}.  
\subsubsection{Terms with $k=2$}  
In these terms there are two boundary vertices labeled by $\gamma$, and no vertices labeled by the vector field $R$. If we sum over all such terms, we get precisely the star product $\gamma \star \gamma$.  
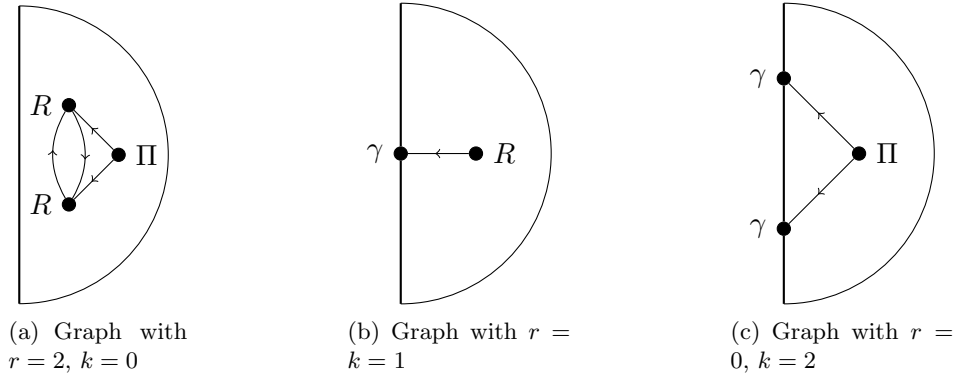
\begin{figure}[ht]
\centering
\subfigure[Graph with $r=2$, $k=0$]{
\begin{tikzpicture}[x=0.04\textwidth, y=0.04\textwidth]
\node[vertex] (R1) at (1,1) {};
\node[left] at (R1.west) {$R$};
\node[vertex] (R2) at (1,-1) {};
\node[left] at (R2.west) {$R$};
\node[vertex] (pi) at (2,0) {};
\node[right] at (pi.east) {$\Pi$};
\node[coordinate] (top) at (0,3) {}; 
\node[coordinate] (bottom) at (0,-3){};
\semiloop[black]{top}{bottom}{-90}[];
\draw[thick,black] (top)  -- (bottom);
\draw[fermion] (pi) -- (R1);
\draw[fermion] (pi) -- (R2);
\draw[fermion] (R2) to [bend left](R1);
\draw[fermion] (R1) to [bend left](R2);
\end{tikzpicture} 
}\hspace{2cm}
\subfigure[Graph with $r=k=1$ ]{
\begin{tikzpicture}
\node[vertex] (center) at (0,0) {};
\node[left] at (center.west) {$\gamma$}; 
\node[vertex] (R) at (1,0) {};
\node[right] at (R.east) {$R$};
\node[coordinate] (top) at (0,2) {}; 
\node[coordinate] (bottom) at (0,-2){};
\semiloop[black]{top}{bottom}{-90}[];
\draw[thick,black] (top) -- (center) -- (bottom);
\draw[fermion] (R) -- (center);
\end{tikzpicture} 
}\hspace{2cm}
\subfigure[Graph with $r=0$, $k=2$]{
\begin{tikzpicture}
\node[vertex] (bdry1) at (0,1) {};
\node[left] at (bdry1.west) {$\gamma$}; 
\node[vertex] (bdry2) at (0,-1) {};
\node[left] at (bdry2.west) {$\gamma$};
\node[vertex] (pi) at (1,0) {};
\node[right] at (pi.east) {$\Pi$};
\node[coordinate] (top) at (0,2) {}; 
\node[coordinate] (bottom) at (0,-2){};
\semiloop[black]{top}{bottom}{-90}[];
\draw[thick,black] (top) -- (bdry1) -- (bdry2) -- (bottom);
\draw[fermion] (pi) -- (bdry1);
\draw[fermion] (pi) -- (bdry2);
\end{tikzpicture} 
}
\caption{Different contributions at the boundary}
\end{figure}

\subsubsection{New contributions at the corners}
Introducing alternating boundary conditions means that the compactification of the configuration space changes. Namely, there are new boundary strata corresponding to the collapse of vertices at one of the corners. Such a collapse can be modeled on a configuration of points on the upper right quadrant, with a choice of boundary conditions on both sides. Here there is no translation symmetry, so the dimension of the boundary stratum is different. Adding $\calS_{\Sigma,\gamma}$ to the action cancels the anomaly that comes from allowing for alternating boundary conditions. However, it results in new boundary contributions that come from graphs collapsing at the corners, as we will show presently. The propagator still vanishes when its tail is evaluated at one of the corners (this can be checked from the explicit formula for the propagator in Appendix \ref{app:prop}). For this reason, as above if some subgraph $\Gamma'$ of a graph $\Gamma$ collapses at a corner, the contribution is only non-vanishing if no arrows leave $\Gamma'$. Let us start at a corner $C$ in $\mathscr{C}_2$. Then we cannot have propagators ending at the $\frac{\delta}{\delta\mathbb{X}}$-polarized boundary, since otherwise we need to evaluate the $\mathbb{E}$-field at the corner point, which is equal to zero because of its boundary condition. So, any subgraph collapsing at $C$ can only have bulk vertices, say $n = m+r$ of them, where $m$ denotes the number of interaction and $r$ the number of $R$ vertices, and vertices and $\de_\calP\Sigma$, say $k$ of them. Counting the dimensions we arrive at the following system of equations: 
\begin{align}
\label{eq:system2}
2n+k-1&=2m +r \\
n&=m+r,
\end{align}
which has the solutions $k=0,r=1$ and $k=1,r=0$, with $m$ arbitrary. However, at these corners, graphs with bulk vertices do not contribute, this is the statement of the following lemma. 
\begin{lem}
If $\Gamma'$ is a subgraph of $\Gamma$ containing bulk points, then the integral of $\omega_{\Gamma}$, defined as in \eqref{eq:omega_gamma}, over the boundary face of $\mathsf{C}_{\Gamma}$ where $\Gamma'$ collapses at a corner $C \in \mathscr{C}_2$ vanishes.
\end{lem}
\begin{proof} The point is that at these corners the boundary conditions are the same on both sides, so we can map the configuration to a configuration of points on the upper half plane, where we use the usual Kontsevich propagator, but without taking the quotient with respect to translations along the real axis. Instead we fix the image of the corner point to be a given point, e.g. $0$. See also Figure \ref{fig:mapping1}. Now, observe that configurations with one bulk point evaluate to 0: These are either $k=0,m=0,r=1$, but this case is ruled out because there are no tadpoles, or $k=1,m=1,r=0$, but this is 0 because graphs cannot double edges. For more than two bulk points, note that the Kontsevich propagator depends on the the real parts of the points in the configuration only through their differences. Hence the product of propgators that is to be integrated has no component in the real part of the center of mass of the configuration, so integrating along this direction yields 0. 
\end{proof}
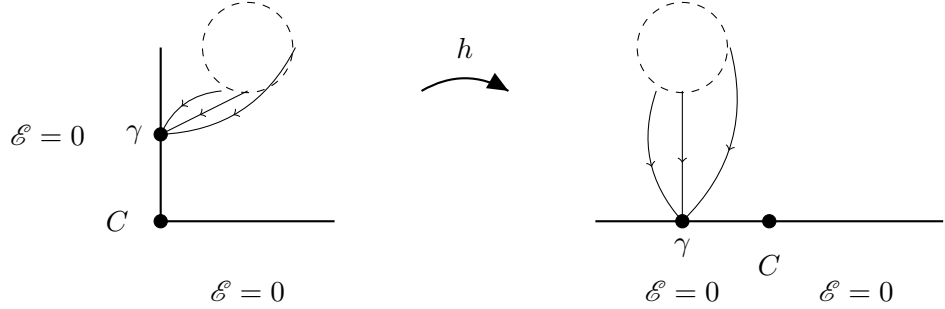
\begin{figure}[ht]
\centering
\tikzset{
particle/.style={thick,draw=black},
particle2/.style={thick,draw=blue},
avector/.style={thick,draw=black, postaction={decorate},
    decoration={markings,mark=at position 1 with {\arrow[black]{triangle 45}}}},
    avector3/.style={thick,draw=black, postaction={decorate},
    decoration={markings,mark=at position 0.5 with {\arrow[black]{triangle 45}}}},
gluon/.style={decorate, draw=black,
    decoration={coil,aspect=0}}
 }
\begin{tikzpicture}[x=0.07\textwidth, y=0.07\textwidth]
\node[](1) at (0,0){};
\node[](3) at (-0.8,-0.3){};
\node[](4) at (0.8,-0.3){};
\node[](g1) at (6,-1.8){$\gamma$};
\node[](g2) at (-0.3,-0.5){$\gamma$};
\node[vertex] (v1) at (0,-1.5){};
\node[vertex] (v2) at (7,-1.5){};
\node[](corner) at (-0.5,-1.5){$C$};
\node[](corner) at (7,-2){$C$};
\node[] (v3) at (6,0.5){};
\node[] (v4) at (1,0.5){};
\node[vertex] (g11) at (6,-1.5){};
\node[vertex] (g22) at (0,-0.5){};
\node[] (d1) at (-1,-2){};
\node[] (s1) at (-2.5,-1.5){};
\node[] (s2) at (2.5,-1.5){};
\draw[dashed,fill=black] (v1) circle (0.05cm);
\draw[dashed,fill=black] (v2) circle (0.05cm);
\draw[dashed,fill=white] (v3) circle (0.6cm);
\draw[dashed,fill=white] (v4) circle (0.6cm);
\draw[dashed,fill=black] (g11) circle (0.05cm);
\draw[dashed,fill=black] (g22) circle (0.05cm);
\node[] (f) at (3.5,0.5){$h$};
\node[] (E1) at (-1.3,-0.5){$\mathscr{E}=0$};
\node[] (E2) at (1,-2.3){$\mathscr{E}=0$};
\node[] (E3) at (6,-2.3){$\mathscr{E}=0$};
\node[] (E4) at (8,-2.3){$\mathscr{E}=0$};
\draw[avector] (3,0) to [bend left](4,0);
\draw[fermion] (5.7,0) to [bend right](6,-1.5);
\draw[fermion] (0.7,0) to [bend right](0,-0.5);
\draw[fermion] (6,0)--(6,-1.5);
\draw[fermion] (1,0)--(0,-0.5);
\draw[fermion] (6.55,0.5) to [bend left](6,-1.5);
\draw[fermion] (1.55,0.5) to [bend left](0,-0.5);
\draw[particle] (0,-1.5)--(0,0.5);
\draw[particle] (0,-1.5)--(2,-1.5);
\draw[particle] (5,-1.5)--(9,-1.5);
\end{tikzpicture}

\caption{Here $h$ represents the mapping of the corner with the interior to the upper half plane, where the corner point is mapped to zero (with the same boundary conditions). The dashed circle represents some graph in the bulk with vertices corresponding to the Poisson structure $\Pi$ and the globalization term $R$, with some outgoing arrows deriving $\gamma$ on the boundary. In particular, the map $h$ is given by $z\mapsto z^2$ on the upper half plane.}
\label{fig:mapping1}
\end{figure}

This means the only possibly nonzero contributions are those with $k=1,n=0$, i.e. subgraphs $\Gamma'$ consisting of a single $\gamma$ vertex - possibly with any number of inward leaves - approaching the corner. This vertex can either lie on the $\de_0\Sigma$ or $\de_\calP\Sigma$ component and the corresponding boundary faces have opposite orientation. Hence all terms cancel out: there are no extra contributions from corners in $\mathscr{C}_2$. \\
Next let us turn to corners $C \in \mathscr{C}_1$. Here the boundary conditions change, so the propagator does not have translation symmetry along the axis. However, by continuity, now it vanishes when either the head or the tail are evaluated at the point of collapse. This implies that a subgraph collapsing at $C$ can have neither inward nor outward leaves, i.e. only entire connected components of graphs can collapse at corner $C \in \mathscr{C}_1$. Counting dimensions as above, we see that there are again the two possibilities $r=0,k=1$ and $r=1,k=0$, with $m$ arbitrary; in addition now we can have an arbitrary number $b$ of vertices at the boundary with $\mathbb{X}$-representation. 

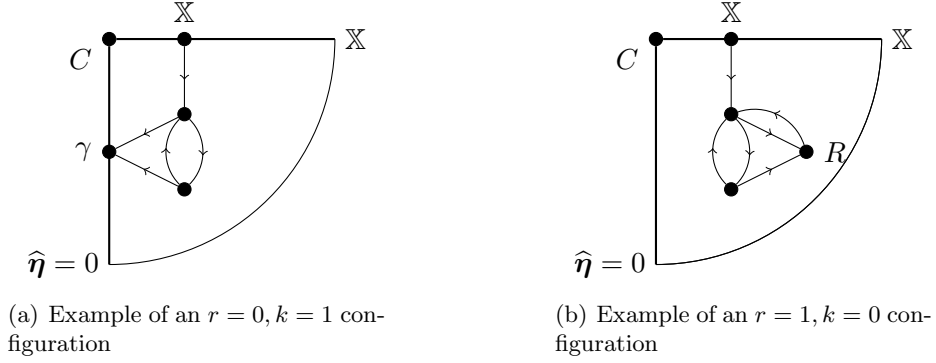
\begin{figure}[ht]

\subfigure[Example of an $r=0,k=1$ configuration]{
\begin{tikzpicture}
\node[coordinate] (bdry1) at (0,-3){};
\node[left] at (bdry1.west){$\widehat{\boldsymbol{\eta}}=0$};
\node[vertex] (corner) at (0,0){};
\node[below left] at (corner.west) {$C$};
\node[coordinate, below] (bdry2) at (3,0){};
\node[right] at (bdry2.east){$\mathbb{X}$}; 
\draw[thick] (bdry1.center) -- (corner) -- (bdry2.center);


\node[vertex] (bdry3) at (1,0) {};
\node[coordinate, label=above:{$\mathbb{X}$}] at (bdry3.north) {};
\node[vertex] (bulk1)  at (1,-1) {}; 
\node[vertex] (bulk2) at (1,-2) {}; 
\node[vertex] (bdry4) at (0,-1.5) {};
\node[coordinate, label=left:{$\gamma$}] at (bdry4.west) {};
\draw[fermion] (bdry3) -- (bulk1); 
\draw[fermion] (bulk1) -- (bdry4);
\draw[fermion] (bulk2) -- (bdry4);
\draw[fermion, bend angle=45, bend left] (bulk1) to (bulk2);
\draw[fermion, bend angle=45, bend left] (bulk2) to (bulk1);

\draw (270:3) arc (270:360:3);
\end{tikzpicture}
}
\hspace{2cm}
\subfigure[Example of an $r=1,k=0$ configuration]{
\begin{tikzpicture}
\node[coordinate] (bdry1) at (0,-3){};
\node[left] at (bdry1.west){$\widehat{\boldsymbol{\eta}}=0$};
\node[vertex] (corner) at (0,0){};
\node[below left] at (corner.west) {$C$};
\node[coordinate, below] (bdry2) at (3,0){};
\node[right] at (bdry2.east){$\mathbb{X}$}; 
\draw[thick] (bdry1.center) -- (corner) -- (bdry2.center);
\draw (270:3) arc (270:360:3);

\node[vertex] (bdry3) at (1,0) {};
\node[coordinate, label=above:{$\mathbb{X}$}] at (bdry3.north) {};
\node[vertex] (bulk1)  at (1,-1) {}; 
\node[vertex] (bulk2) at (1,-2) {}; 
\node[vertex] (bulk3) at (2,-1.5) {};
\node[coordinate, label=right:{$R$}] at (bulk3.east) {};
\draw[fermion] (bdry3) -- (bulk1); 
\draw[fermion] (bulk1) -- (bulk3);
\draw[fermion] (bulk2) -- (bulk3);
\draw[fermion, bend angle=45, bend right] (bulk3) to (bulk1);
\draw[fermion, bend angle=45, bend left] (bulk1) to (bulk2);
\draw[fermion, bend angle=45, bend left] (bulk2) to (bulk1);
\draw (270:3) arc (270:360:3);
\end{tikzpicture}
}
\caption{Possibilities for graphs collapsing at $C \in
\mathscr{C}_1$.}
\label{fig:corners}
\end{figure}

Since only connected components of a graph can collapse, the corresponding  action on the state is a multiplication operator $T(C)$ that multiplies states with a functional of the values of $\mathbb{X}$ at corners in $\mathscr{C}_1$, given by summing over all possible boundary contributions. Since $\gamma$ and $R$ are both $1$-forms on $\mathscr{P}$, $T(C)$ takes values in $1$-forms on $\mathscr{P}$.
This is not a regular functional as in \ref{regular_fun}, as it contains evaluation of fields on the corners. This completes the proof of Theorem \ref{thm_corners}.

\section{The modified differential Quantum Master Equation for the globalized Poisson Sigma Model with alternating boundary conditions}
\label{sec:mdQME}
We have seen that the modified differential Quantum Master Equation fails if we impose alternating boundary conditions as in Proposition \ref{prop_curv} and Theorem \ref{thm_corners}. Hence we need to extend the quantum Grothendieck BFV operator on an extended space of operators and states such that the modified differential Quantum Master Equation holds for the extended connection. The plan is to promote the corner terms $T(C)$ to multiplication operators on the state space. This requires the extensions of the state space to include functionals which evaluate fields at the corners.

\subsection{Extension of states}
\label{subsec:ex_states}
There are two different terms in $T(C)$, namely the one where we have a single $\gamma$ on the boundary approaching the corner and no vector field $R$, or no boundary vertex on the $\hateta=0$ component and one single vector field $R$ included in the graph from the bulk (see also Figure \ref{fig:corners}). To interpret them as multiplication operators we have to enlarge the space of states to allow for functionals evaluating boundary fields at corners.


\begin{defn}[Space of corner states]
For $C\in\mathscr{C}_1$, we define the space of corner states by 
\begin{equation}
\gls{calH_C}:=\left\{F\colon \calB^\calP_{\partial \Sigma}\to \mathbb{C}[[\hbar]]\Big| F(\mathbb{X})=\sum_{J} B_J[\mathbb{X}^{J}(C)],\textnormal{ where } B_J\in\C[[\hbar]]\right\}
\end{equation}
\end{defn}

\begin{defn}[Extended state space]\label{defn_ext_state_space}
We define the extended state space by
\begin{equation}
\gls{ext_state_space}:=\Hat{\calH}_{\de\Sigma,x}\otimes \bigotimes_{C\in\mathscr{C}_1}\calH_C.
\end{equation}
Moreover, the total space is given by $\gls{tot_ext_state_space}:=\bigsqcup_{x\in\mathscr{P}}\Hat{\calH}_{\de\Sigma,x}^\mathscr{C}$.
\end{defn}


Now we can define a state to be given as a nonhomogeneous differential form on $\mathscr{P}$ with values in $\Hat{\calH}_{\de\Sigma,\textnormal{tot}}^\mathscr{C}$, i.e. an element of $\Omega^\bullet(\mathscr{P},\Hat{\calH}_{\de\Sigma,\textnormal{tot}}^\mathscr{C})$.


\subsection{Extension of operators}
\label{subsec:ex_op}
Recall from \cite{CMR2} that the algebra of the operators is generated by $\Omega_0^\textnormal{princ}$, which is the standard quantization of $\calS_{0,\Sigma}$, and simple operators, which are of the form
\begin{equation}
\int_{\partial\Sigma}L_{I^1\dotsm I^r}^J[\mathbb{X}^{I^1}]\dotsm [\mathbb{X}^{I^r}]\frac{\delta^{\vert J\vert}}{\delta\mathbb{X}^J},
\end{equation}
where $L_{I^1\dotsm I^r}^J\in\Omega^\bullet(\de\Sigma)$ are some coefficients. Note that we can also have a similar expression for $\E$. We want to extend this space by the multiplication operators coming from the corners as described above. The space of operators is extended by the multiplication operators that appear in the case of corners. The algebra of boundary operators acts on the algebra of corner operators by commutators. E.g. $\partial_k\Pi^{ij}\mathbb{X}^{k}\frac{\delta}{\delta[\mathbb{X}^{i}\mathbb{X}^j]}$ is a boundary operator and $[\mathbb{X}^{i}\mathbb{X}^j](C)\partial_i\gamma\partial_j\gamma$ is a corner operator, with $C\in\mathscr{C}_1$. Then the commutator is given by 
\begin{equation}
\left[\partial_k\Pi^{ij}\mathbb{X}^{k}\frac{\delta}{\delta[\mathbb{X}^{i}\mathbb{X}^j]},[\mathbb{X}^{i}\mathbb{X}^j](C)\partial_i\gamma\partial_j\gamma\right]=\partial_k\Pi^{ij}\mathbb{X}^k(C)\partial_i\gamma\partial_j\gamma.
\end{equation}

The extended space now consists of operators taking a state in $\Omega^\bullet(\mathscr{P},\Hat{\calH}^\mathscr{C}_{\de\Sigma,\textnormal{tot}})$ and multiplying it with an element in $\Omega^\bullet(\mathscr{P},\calH_C)$.

\subsection{modified differential Quantum Master Equation and Flatness} 
Now we are able to define the extended operator as follows. Let $\gls{Omega_scrC}:=-\sum_{C\in\mathscr{C}_1}T(C)$, where $T(C)$ is as in Theorem \ref{thm_corners}. The new operator $\Tilde{\nabla}^\gamma_\mathsf{G}$ is then defined by 
\begin{equation}
\label{eq:conntilde}
\gls{tildenabla^gamma_G}:=\dd_x-\I\hbar\Delta_{\calV_{\Sigma}}+\frac{\I}{\hbar}\left(\Tilde{\boldsymbol{\Omega}}_{\de\Sigma}^{\calP,\gamma}+\boldsymbol{\Omega}_\mathscr{C}\right).
\end{equation}

\subsubsection{The modified differential Quantum Master Equation}
We have the following theorem.
\begin{thm}[modified differential Quantum Master Equation for alternating boundary conditions]\label{thm_mdQME_corners}
Let $\Tilde{\nabla}^\gamma_\mathsf{G}$ be given as before, and consider the twisted full state $\btpsi_{\Sigma,x}^\gamma$. Then 
\begin{equation}
\label{mdQME_alt_boundary}
\Tilde{\nabla}^\gamma_\mathsf{G}\btpsi_{\Sigma,x}^\gamma=0
\end{equation}
\end{thm}

\begin{proof}
This follows immediatley from Theorem \ref{thm_corners}.
\end{proof}


\subsubsection{Flatness}
We have the following theorem.
\begin{thm}
\label{thm:flatness}
The operator $\Tilde{\nabla}^\gamma_\mathsf{G}$ is a coboundary operator, i.e. $(\Tilde{\nabla}^\gamma_{\mathsf{G}})^2=0$.
\end{thm}

\begin{proof}
The flatness condition is equivalent to the fact that $\boldsymbol{\Omega}^{\text{ext}}=\Tilde{\boldsymbol{\Omega}}_{\de\Sigma}^{\calP,\gamma}+\boldsymbol{\Omega}_\mathscr{C}$ is a Maurer--Cartan element of the differential graded Lie algebra of differential forms with values in $\End(\Hat{\calH}^\mathscr{C}_{\de\Sigma,\textnormal{tot}})$. Hence the proof of Theorem \ref{thm:flatness} is given by the Proposition \ref{prop:maurer-cartan}.
\end{proof}

\begin{prop}
\label{prop:maurer-cartan}
$\dd_x\boldsymbol{\Omega}^{\textnormal{ext}}+\frac{1}{2}\left[\boldsymbol{\Omega}^{\textnormal{ext}},\boldsymbol{\Omega}^{\textnormal{ext}}\right]=0$.
\end{prop}

\begin{proof}
First of all note that $\dd_x\Tilde{\boldsymbol{\Omega}}_{\de\Sigma}^{\calP,\gamma}+\frac{1}{2}\left[\Tilde{\boldsymbol{\Omega}}_{\de\Sigma}^{\calP,\gamma},\Tilde{\boldsymbol{\Omega}}_{\de\Sigma}^{\calP,\gamma}\right]=0$. This means we only need to prove
\begin{equation}
\dd_x\boldsymbol{\Omega}_{\mathscr{C}}+\frac{1}{2}\left[\boldsymbol{\Omega}_{\mathscr{C}},\boldsymbol{\Omega}_{\mathscr{C}}\right]+\left[\boldsymbol{\Omega}_{\mathscr{C}},\Tilde{\boldsymbol{\Omega}}_{\de\Sigma}^{\calP,\gamma}\right]=0.
\end{equation}
We can show this similarly to \cite{CMR2, CMW4}. Namely, since $\Tilde{\boldsymbol{\Omega}}_{\de\Sigma}^{\calP,\gamma}$ and $\boldsymbol{\Omega}_{\mathscr{C}}$ are given as sum of integrals over the boundary of the configuration space of collapsing graphs, we can use Stokes' Theorem: 
\begin{equation}
\dd_x \boldsymbol{\Omega}_{\mathscr{C}} = \dd_x \sum_{\Gamma' \leq \Gamma} \int_{\mathsf{C}^\mathscr{C}_{\Gamma'}(\Sigma)}\sigma_{\Gamma'} = \sum_{\Gamma' \leq \Gamma} \int_{\mathsf{C}^\mathscr{C}_{\Gamma'}(\Sigma)}\dd\sigma_{\Gamma'} - \int_{\de \mathsf{C}^\mathscr{C}_{\Gamma'}(\Sigma)} \sigma_{\Gamma'}
\end{equation}

Here $\mathsf{C}^\mathscr{C}_{\Gamma'}(\Sigma)$ is the configuration space describing the relative position of the vertices of the subgraph $\Gamma$ collapsing to the corner. In the first, the differential can act on the propagators, the boundary fields, or the vertex tensors $\mathsf{T}\varphi_x^*\Pi,\gamma, R$. The restriction of the propagators to this boundary face is closed, see Appendix \ref{app:prop}. If the differential acts on the boundary fields, this yields $[\Omega_0^\textnormal{princ},\boldsymbol{\Omega}_{\mathscr{C}}]$. The differential acting on vertex tensors will be cancelled by boundary terms. Notice that on the boundary faces the dimension counting is different and we can have either two vertices labeled by $R$, one $R$ vertex and one $\gamma$ vertex on the boundary or two $\gamma$ vertices on the boundary. A boundary face of $\mathsf{C}^\mathscr{C}_{\Gamma'}(\Sigma)$ corresponds to a collapse of a subgraph $\Gamma'' \leq \Gamma'$ to a single point. There are four distinct possibilities for that point (see Figure \ref{fig:diff_collapse}): 
\begin{itemize}
\item The point can be in the bulk. If $\Gamma''$ contains more than two vertices then the contribution is zero by a Kontsevich vanishing lemma. 
If it contains exactly two vertices, there is a cancellation similar to the proof of the modified differential Quantum Master Equation using the classical master equation, the fact that vertex tensors are $\dd_x + R$ closed, and that $[R,R]=0$. 
\item The point can be the corner. These terms yield $[\boldsymbol{\Omega}_{\mathscr{C}},\boldsymbol{\Omega}_{\mathscr{C}}]$.
\item The point can be at the boundary with the $\hateta \equiv 0$ boundary condition. In that case there is a cancellation similar to one in the proof of the modified differential Quantum Master Equation in section \ref{subsec_boundary} using the equation $\dd_x\gamma + A(R,\mathsf{T}\varphi^*_x\Pi)(\gamma) + \gamma \star \gamma + F(R,R,\mathsf{T}\varphi^*_x\Pi) = 0.$ 
\item The point can be on the upper boundary, this corresponds to $\left[\Tilde{\boldsymbol{\Omega}}_{\de\Sigma}^{\calP,\gamma},\boldsymbol{\Omega}_{\mathscr{C}}\right]$, the action of the algebra of boundary operators on the algebra of corner operators. 
\end{itemize}

\begin{figure}[ht]

\subfigure[Collapsing for a bulk point]{
\begin{tikzpicture}
\node[coordinate] (bdry1) at (0,-3){};
\node[left] at (bdry1.west){$\widehat{\boldsymbol{\eta}}=0$};
\node[vertex] (corner) at (0,0){};
\node[below left] at (corner.west) {$C$};
\node[coordinate, below] (bdry2) at (3,0){};
\node[right] at (bdry2.east){$\mathbb{X}$}; 
\draw[thick] (bdry1.center) -- (corner) -- (bdry2.center);
\draw (270:3) arc (270:360:3);
\draw (1,-1) circle (.5cm);
\end{tikzpicture}
}
\hspace{2cm}
\subfigure[Collapsing for a point on the $\Hat{\boldsymbol{\eta}}=0$ boundary]{
\begin{tikzpicture}
\node[coordinate] (bdry1) at (0,-3){};
\node[left] at (bdry1.west){$\widehat{\boldsymbol{\eta}}=0$};
\node[vertex] (corner) at (0,0){};
\node[below left] at (corner.west) {$C$};
\node[coordinate, below] (bdry2) at (3,0){};
\node[right] at (bdry2.east){$\mathbb{X}$}; 
\draw[thick] (bdry1.center) -- (corner) -- (bdry2.center);
\draw (270:3) arc (270:360:3);
\draw (0,-2.5) arc (-90:90:1cm);
\end{tikzpicture}
}
\hspace{2cm}
\subfigure[Collapsing for a corner point]{
\begin{tikzpicture}
\node[coordinate] (bdry1) at (0,-3){};
\node[left] at (bdry1.west){$\widehat{\boldsymbol{\eta}}=0$};
\node[vertex] (corner) at (0,0){};
\node[below left] at (corner.west) {$C$};
\node[coordinate, below] (bdry2) at (3,0){};
\node[right] at (bdry2.east){$\mathbb{X}$}; 
\draw[thick] (bdry1.center) -- (corner) -- (bdry2.center);
\draw (270:3) arc (270:360:3);
\draw (0,-1) arc (-90:0:1cm);
\end{tikzpicture}
}
\hspace{2cm}
\subfigure[Collapsing for a point on the polarized boundary]{
\begin{tikzpicture}
\node[coordinate] (bdry1) at (0,-3){};
\node[left] at (bdry1.west){$\widehat{\boldsymbol{\eta}}=0$};
\node[vertex] (corner) at (0,0){};
\node[below left] at (corner.west) {$C$};
\node[coordinate, below] (bdry2) at (3,0){};
\node[right] at (bdry2.east){$\mathbb{X}$}; 
\draw[thick] (bdry1.center) -- (corner) -- (bdry2.center);
\draw (270:3) arc (270:360:3);
\draw (0.5,0) arc (-180:0:1cm);
\end{tikzpicture}
}
\caption{Illustration of the different cases for the collapsing}
\label{fig:diff_collapse}
\end{figure}
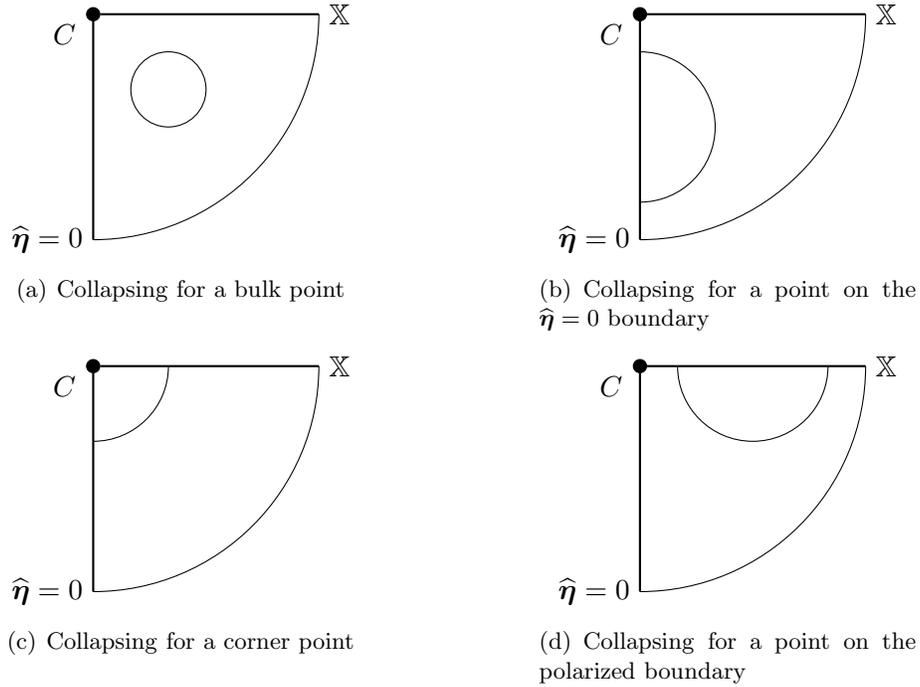

\end{proof}

\begin{rem}
The failure of the (modifed) differential Quantum Master Equation and its resolution is somehow similar to what happens in the Landau--Ginzburg model \cite{KL,BHLS,Laz}. Namely, the classical boundary conditions turn out not to be compatible with quantization. The resolution consists in coupling the bulk theory with a boundary theory with action $\calS_{\Sigma,\gamma}$.  
\end{rem}

\section{Outlook}
\label{sec:Outlook}
\subsection{Relational Symplectic Groupoid}

\subsubsection{Kontsevich's star product}
One can construct the Moyal product \cite{Moy} (deformation quantization) as the gluing of canonical relations as it was shown in \cite{CMW2}. It still remains to show that one can also use the gluing of the relational symplectic groupoid to construct a globalized version of Kontsevich's star product using the gluing formulas of the BV-BFV formalism. One can use the results of this paper to deal with the $L_3$ worldsheet structure, which is given as in Figure \ref{RSG1} with mixed boundary conditions.

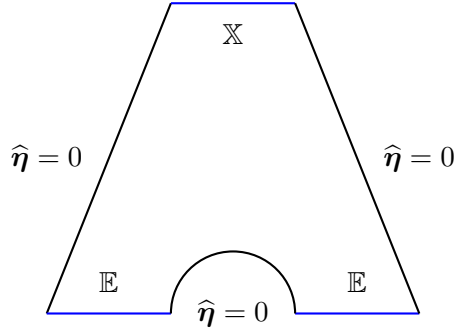
\begin{figure}[ht]
\centering
\tikzset{
particle/.style={thick,draw=black},
particle2/.style={thick,draw=blue},
avector/.style={thick,draw=black, postaction={decorate},
    decoration={markings,mark=at position 1 with {\arrow[black]{triangle 45}}}},
gluon/.style={decorate, draw=black,
    decoration={coil,aspect=0}}
 }
\begin{tikzpicture}[x=0.05\textwidth, y=0.05\textwidth]
\draw[particle] (0,0)--(2,5) node[above]{};
\draw[particle2] (2,5)--(4,5) 
node[right]{};
\draw[particle2] (0,0)--(2,0) node[right]{};
\node[](u1) at (2,0){};
\node[](u2) at (4,0){};
\draw[particle2] (4,0)--(6,0)
node[right]{};
\draw[particle] (6,0)--(4,5) node[right]{};
\node[](2) at (0.7,0){};
\node[](3) at (1.4,0){};
\node[](eta1) at (0,2.5){$\widehat{\boldsymbol{\eta}}=0$};
\node[](eta2) at (6,2.5){$\widehat{\boldsymbol{\eta}}=0$};
\node[](eta3) at (3,0){$\widehat{\boldsymbol{\eta}}=0$};
\node[](eta1) at (1,0.5){$\E$};
\node[](e2) at (5,0.5){$\E$};
\node[](x) at (3,4.5){$\mathbb{X}$};
\semiloop[particle]{u1}{u2}{0};
\end{tikzpicture}
\caption{The canonical relation $L_3$ with its boundary structure. Here we have two $\frac{\delta}{\delta\mathbb{X}}$-polarized boundaries (the lower) and one $\frac{\delta}{\delta\E}$-polarized boundary (the upper), which would correspond to $\partial_2 L_3$ and the $\widehat{\boldsymbol{\eta}}=0$ boundaries which are components of $\partial_1 L_3$.}
\label{RSG1}
\end{figure}

\subsubsection{Relational symplectic groupoid with handles}
Another interesting aspect would be to consider the relational symplectic groupoid with handles. That is, one considers canonical relations $L_3$ with non vanishing genus. Since our theory is topological, we are able to move the handle in arbitrary directions, which means that one has to understand what happens when a hole will approach an observable for the gluing of the disk in \cite{CMW4}. Moreover, one has to check what kind of structures appear for associativity.

\subsubsection{Generalization of Kontsevich's star product}
Kontsevich's star product arises from the computation of expectation values of observables in the Poisson Sigma Model for a genus zero wordsheet surface. As in string theory, one expects that we should sum over all genera. Since a particular gluing of the relational symplectic groupoid gives rise to Kontsevich's star product, one can relate this structure to the relational symplectic groupoid construction with handles. 

\vspace{0.5cm}
We will return on these questions in a forthcoming paper.

\subsection{Manifolds with corners}
The methods developed in this paper can be useful to give a description for the the quantization of manifolds with corners. Here the corners arose from the structure of mixed boundary conditions, but in principle the methods that we develop might be adapted to the general case. Another paper in this direction is \cite{IM}.

\subsection{Globalization of other theories} 
AKSZ theories have a particularly nice subset of classical solutions, the space of constant maps. This subset admits for a natural globalization, as was shown in \cite{CMW4}. It would be interesting to see whether the methods we used carry over to more complicated moduli spaces of classical solutions. E.g. in Chern-Simons theory, this subset is just the trivial connection, since the body of the target in that case is just a point, but one would like to take non-trivial connections into account as well.  



\begin{appendix}
\section{Configuration spaces and their compactifications}\label{app:Conf}

To define the quantum state, we need to recall the notion of configuration spaces and their compactification as in \cite{AS2,FulMacPh} due to Fulton--MacPherson and Axelrod--Singer. 
\subsection{FMAS-compactification}
We start with the definition of the configuration space. 
\begin{defn}
Let $M$ be a manifold and $S$ a finite set. The \textsf{open configuration space} of $S$ in $M$ is defined as
\begin{equation}
\mathsf{Conf}_S(M) := \{\iota\colon S \hookrightarrow M |\iota \hspace{0.2cm}\text{injection}\}
\end{equation}
\end{defn}
Elements of $\textsf{Conf}_{S}(M)$ are called $S$-configurations. To give an explicit definition of the compactification that can be extended to manifolds with boundaries and corners, we introduce the concept of \textsf{collapsed configurations}. Intuitively, a collapsed $S$-configuration is the result of a collapse of a subset of the points in the $S$-configuration. However, we remember the relative configuration of the points before the collapse by directions in the tangent space. This is a configuration in the tangent space that is well-defined only up to translations and scaling. The difficulty is that one can imagine a limiting configuration where two points collapse first together and then with a third (see Figure \ref{fig:collapsing_conf}).
This explains the recursive nature of the following definition. Recall that if $X$ is a vector space, then $X\times \R_{>0}$ acts on $X$ by translations and scaling. 
\begin{defn}[Collapsed configuration in $M$]
Let $M$ be a manifold, $S$ a finite set and $\mathfrak{P} = \{S_1,\ldots,S_k\}$ be a partition of $S$. A \textsf{$\mathfrak{P}$-collapsed configuration in $M$} is a $k$-tuple $(p_{\sigma},c_{\sigma})$ such that 
$((p_{\sigma},c_{\sigma}))_{\sigma = 1}^k$ satisfies 
\begin{enumerate}
\item $p_{\sigma} \in M$ and $p_{\sigma} \neq p_{\sigma'}$, for $\sigma \neq \sigma'$, 
\item $c_{\sigma} \in \Tilde{\mathsf{C}}_{S_{\sigma}}(T_{p_{\sigma}}M)$, where for $|S| = 1$, 
$\Tilde{\mathsf{C}}_S(X) := \{pt\}$  and for $|S| \geq 2$
\begin{equation}
\Tilde{\mathsf{C}}_S(X) := \coprod_{\substack{\mathfrak{P}=\{S_1,\ldots,S_k\} \\ S = \sqcup_\sigma S_\sigma, k\geq 2}}\left\lbrace \left(x_\sigma, c_\sigma \right)_{1\leq \sigma \leq k}\ \bigg|\ (x_\sigma, c_\sigma)\ \mathfrak{P}\text{-collapsed $S$-configuration in $X$}\right\rbrace \bigg/(X \times \R_{>0})
\end{equation}
\end{enumerate}

Here, $\varphi \in X \times \R_{>0}$ acts on $(x_\sigma,c_\sigma)$ by $(x_\sigma,c_\sigma) \mapsto (\varphi(x_\sigma), \dr\varphi_{x_\sigma}c_\sigma)$. 
\end{defn}
Intuitively, given a partition $\mathfrak{P}= \{S_1,\ldots,S_k\}$, a $k$-tuple $(p_\sigma,c_\sigma)$ describes the collapse of the points in $S_\sigma$ to $p_\sigma$. $c_\sigma$ remembers the relative configuration of the collapsing points. This relative configuration can itself be the result of a collapse of some points. 

\begin{defn}[FMAS compactification]
The \textsf{compactified configuration space} $\mathsf{C}_S(M)$ of $S$ in $M$ is given by 
\begin{equation}
\mathsf{C}_S(M) := \coprod_{\substack{S_1,\ldots,S_k \\ S = \sqcup_\sigma S_\sigma}}\left\lbrace (p_\sigma, c_\sigma )_{1\leq \sigma \leq k}\ \bigg|\ (p_\sigma, c_\sigma)\ \mathfrak{P}\text{-collapsed $S$-configuration in $M$}\right\rbrace.
\end{equation}
%
\end{defn}

\subsection{Boundary strata}
A precise description of the combinatorics of the stratification can be found in \cite{FulMacPh}, where it is also shown that $\mathsf{C}_S(M)$ is a manifold with corners and is compact if $M$ is compact. For us, only strata in low codimensions are interesting. Let $S=\{s_1,\ldots,s_k\}$. 
The stratum of codimension 0 corresponds to the partition $\mathfrak{P}= \{\{s_1\},\ldots,\{s_k\}\}$. For $\ell>1$, strata of codimension 1 correspond to the collapse of exactly one subset $S'=\{s_1,\ldots,s_\ell\} \subset S$ with no further collapses, i.e a partition $\mathfrak{P}=\{\{s_1,\ldots,s_\ell\},\{s_{\ell+1}\},\ldots,\{s_k\}\}$ and configuration $(p_\sigma,c_\sigma)$ with $c_\sigma$ in the component of $\Tilde{\mathsf{C}}_{S'}(X)$ given by the partition $\mathfrak{P} =  \{\{s_1\},\ldots,\{s_\ell\}\}$. This boundary stratum will be denoted by $\de_{S'}\mathsf{C}_S(M)$, in particular, we have
\begin{equation}
\de\mathsf{C}_S(M) = \coprod_{S' \subset S}\de_{S'}\mathsf{C}_S(M).
\end{equation} 
There is a natural fibration $\de_{S'}\mathsf{C}_S(M) \to \mathsf{C}_{S \setminus S'\cup \{pt\}}(M)$ whose fiber is $\Tilde{\mathsf{C}}_S(\R^{\dim M})$. Finally, we note that if $|S| = 2$, then $\mathsf{C}_{S}(M) \cong Bl_{\overline{\Delta}}(M \times M)$, the \textsf{differential-geometric blow-up} of the diagonal $\overline{\Delta} \subset M \times M$, and $\Tilde{\mathsf{C}}_S(X) \cong S^{\dim X-1}$. See Figure \ref{fig:collapsing_conf} for an example of a configuration of points and coresponding boundary strata. 

\begin{figure}[ht]
\begingroup%
  \makeatletter%
  \providecommand\color[2][]{%
    \errmessage{(Inkscape) Color is used for the text in Inkscape, but the package 'color.sty' is not loaded}%
    \renewcommand\color[2][]{}%
  }%
  \providecommand\transparent[1]{%
    \errmessage{(Inkscape) Transparency is used (non-zero) for the text in Inkscape, but the package 'transparent.sty' is not loaded}%
    \renewcommand\transparent[1]{}%
  }%
  \providecommand\rotatebox[2]{#2}%
  \ifx\svgwidth\undefined%
    \setlength{\unitlength}{177.26514562bp}%
    \ifx\svgscale\undefined%
      \relax%
    \else%
      \setlength{\unitlength}{\unitlength * \real{\svgscale}}%
    \fi%
  \else%
    \setlength{\unitlength}{\svgwidth}%
  \fi%
  \global\let\svgwidth\undefined%
  \global\let\svgscale\undefined%
  \makeatother%
  \begin{picture}(1,1.22602258)%
    \put(0,0){\includegraphics[width=\unitlength]{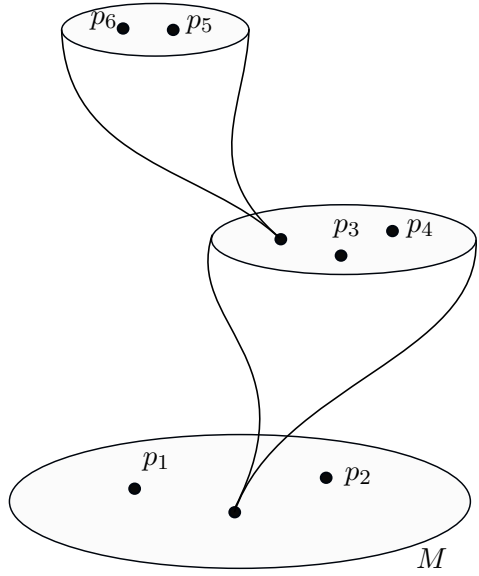}}%
    \put(0.28504216,0.22615565){\color[rgb]{0,0,0}\makebox(0,0)[lb]{\smash{$p_1$}}}%
    \put(7.45141911,-5.32081804){\color[rgb]{0,0,0}\makebox(0,0)[lt]{\begin{minipage}{2.35581904\unitlength}\raggedright \end{minipage}}}%
    \put(0.69011264,0.72199957){\color[rgb]{0,0,0}\makebox(0,0)[lb]{\smash{$p_3$}}}%
    \put(0.7144357,0.19405763){\color[rgb]{0,0,0}\makebox(0,0)[lb]{\smash{$p_2$}}}%
    \put(0.8659863,0.0058981){\color[rgb]{0,0,0}\makebox(0,0)[lb]{\smash{$M$}}}%
    \put(0.84598042,0.72266449){\color[rgb]{0,0,0}\makebox(0,0)[lb]{\smash{$p_4$}}}%
    \put(0.37765289,1.15553808){\color[rgb]{0,0,0}\makebox(0,0)[lb]{\smash{$p_5$}}}%
    \put(0.17364444,1.16429723){\color[rgb]{0,0,0}\makebox(0,0)[lb]{\smash{$p_6$}}}%
  \end{picture}%
\endgroup%

\caption{An element of $\mathsf{C}_S(M)$.}\label{fig:collapsing_conf}
\end{figure}

\subsection{Configuration spaces for manifolds with boundary}
We proceed to recall the definition of a compactified configuration space for manifolds with boundary. Let $M$ be a compact manifold with boundary $\de M$. Recall that for a manifold $M$ with boundary $\de M$, at points $p \in \de M$ there is a well-defined notion of inward and outward half-space in $T_pM$. If $H \subset X$ is a half-space, then $\de H \subset X$ is a hyperplane. $\de H \times \R_{>0}$ acts on $H$ by translations and scaling. 
%
%
\begin{defn}[Configuration spaces for manifolds with boundary]
Let $M$ be a manifold with boundary $\de M$. For $S,T$ finite sets, we define the \textsf{open configuration space} by
\begin{equation}
\mathsf{Conf}_{S,T}(M,\de M) := \{(\iota,\iota')\colon S \times T \hookrightarrow M \times \de M\}
\end{equation}
\end{defn}
\begin{defn}[Collapsed configuration on manifolds with boundary]
Let $(M, \de M)$ be a manifold with boundary. Let $S, T$ be finite sets and $\mathfrak{P}=\{S_1, \ldots, S_k\}$ a partition of $S \sqcup T$.  Then, a $\mathfrak{P}$-collapsed $(S,T)$-configuration in $M$ is a $k$-tuple of pairs $(p_\sigma,c_\sigma)$ such that
\begin{enumerate}
\item $p_\sigma \in M$ and $p_\sigma\neq p_{\sigma'}$, for all $\sigma\neq \sigma'$, 
\item $S_\sigma \cap T \neq \varnothing \Rightarrow p_\sigma \in \de M$,
\item $$c_\sigma \in \begin{cases} 
\Tilde{\mathsf{C}}_{S_\sigma}(T_{p_\sigma}M)  & p_\sigma \in M \setminus \de M \\  \Tilde{\mathsf{C}}_{S \cap S_\sigma, T \cap S_\sigma}(\mathbb{H}(T_{p_\sigma}M)) & p_\sigma \in \de M\end{cases}$$
\end{enumerate}
where $\mathbb{H}(T_{p_\sigma}M) \subset T_{p_\sigma}M$ denotes the inward half-space in $T_{p_\sigma}M$. Here, for a vector space $X$ and a half-space $H \subset X$, $\Tilde{\mathsf{C}}_{\varnothing,\{pt\}}(H) := \Tilde{\mathsf{C}}_{\{pt\},\varnothing}(H) := \{pt\}$, and for $\vert S \sqcup T\vert \geq 2$, 
\begin{equation}
\Tilde{\mathsf{C}}_{S,T}(H) := \coprod_{\substack{\mathfrak{P}=\{S_1,\ldots,S_k\}\\ S\sqcup T = \sqcup_\sigma S_\sigma, k \geq 2 }}\left\lbrace (v_\sigma,c_\sigma)\ \bigg|\ (v_\sigma, c_\sigma)\ \mathfrak{P}\text{-collapsed $(S,T)$-configuration in $H$} \right\rbrace \bigg/ (\de H \times \R_{>0}) 
\end{equation}

\end{defn}
\begin{defn}[FMAS compactification for manifolds with boundary]
We define the  \textsf{compactification} $\mathsf{C}_{S,T}(M,\de M)$ of $\mathsf{Conf}_{S,T}(M,\de M)$ by 
\begin{equation}
\mathsf{C}_{S,T}(M,\de M) = 
\coprod_{\substack{ \mathfrak{P} =\{S_1,\ldots,S_k\} \\ S\sqcup T= \sqcup_\sigma S_\sigma}}\left\lbrace \left(p_\sigma , c_\sigma\right)_{1\leq \sigma \leq k}\ \bigg|\ (p_\sigma,c_\sigma)\  \mathfrak{P}\text{-collapsed $(S,T)$-configuration}\right\rbrace
\end{equation}

\end{defn}
 Again, this is a manifold with corners and is compact if $M$ is compact. We proceed to describe the strata of low codimension. Let $U= \{u_1,\ldots,u_k\}, V = \{v_1,\ldots,v_k\}.$ The codimension 0 stratum again is given by the partition $\mathfrak{P} = \{\{u_1\},\ldots,\{u_k\},\{v_1\},\ldots, \{v_\ell\}\}.$ Let us describe the strata of codimension 1. We denote by $\de^\textnormal{I}_S\mathsf{C}_{U,V}(M,\de M)$ a boundary stratum where a subset $S \subset U$ collapses in the bulk, described in the same way as above.  On manifolds with boundary, there are new boundary strata in the compactified configuration space given by the collapse of a subset of points to a point in the boundary. Concretely, given a subset $S=\{u_1,\ldots,u_{k'},v_1,\ldots,v_{\ell'}\} \subset U \sqcup V$, there is a boundary stratum $\de^{\textnormal{II}}_S\mathsf{C}_{U,V}(M,\de M)$ corresponding to the partition $\mathfrak{P}=\{S,\{u_{k'+1}\},\ldots, \{u_k\},\{v_{\ell'+1}\},\ldots,\{v_\ell\}\}$ and collapsed configurations $(p_\sigma,c_\sigma)$ with $p_{\sigma} \in \de M$ and $c_\sigma$ corresponding to the partition $\mathfrak{P'} = \{\{u_1\},\ldots,\{u_k\},\{v_1\},\ldots,\{v_\ell\}\}$. The boundary decomposes as 
\begin{equation}
\de \mathsf{C}_{U,V}(M,\de M) = \coprod_{S \subseteq U} \de^\textnormal{I}_S\mathsf{C}_{U,V}(M,\de M) \amalg\coprod_{S \subseteq U \sqcup V} \de^{\textnormal{II}}_S\mathsf{C}_{U,V}(M,\de M)
\end{equation} 

\subsection{Configuration spaces for manifolds with corners}
Finally, we consider a manifold $M$ with boundary $\de M$ and corners $\de\de M$. Note that around points in corners $p \in \de\de M$ there is a notion of inward quadrant $\bbQ(T_pM) \subset T_pM$. It can be defined e.g. in coordinates, since the transition functions have to preserve both boundaries and corners. If $Q \subset X$ is any quadrant, its boundary is the union of two half-hyperplanes whose intersection is a $(\dim X - 2)$-dimensional subspace $W$. This subspace acts on $Q$ by translations. Again, $\R_{>0}$ acts on $Q$ by scaling. Note that in this case, $\Tilde{\mathsf{C}}_{\{pt\}}(Q) \cong I$, where $I$ is an interval. Hence the definition of collapsed configurations should be adapted to this case. We want to compactify the open configuration spaces 
\begin{equation}
\mathsf{Conf}_{S,T,U}^\mathscr{C}(M,\de M,\de\de M)
\end{equation} 
where $M$ is a manifold with corners. We proceed to define collapsed configurations as above: 
\begin{defn}[Collapsed configurations for manifolds with corners] 
Let $(M,\de M, \de\de M)$ be a manifold with corners. Let $S,T,U$ be finite sets and $\mathfrak{P} = (S_1,\ldots,S_k)$ be a partition of $S \sqcup T\sqcup U$. Then a $\mathfrak{P}$-collapsed $(S,T,U)$-configuration in $M$ is a $k$-tuple of pairs $(p_\sigma,c_\sigma)$ such that 
\begin{enumerate}
\item $p_\sigma \in M$ and $p_\sigma\neq p_{\sigma'}$, for all $\sigma\neq \sigma'$, 
\item $S_\sigma \cap T \neq \varnothing \Rightarrow p_\sigma \in \de M$,
\item $S_\sigma\cap U\neq \varnothing \Rightarrow p_\sigma\in\de\de M$,
\item $$c_\sigma \in \begin{cases} 
\Tilde{\mathsf{C}}^\mathscr{C}_{S_\sigma}(T_{p_\sigma}M)  & p_\sigma \in M \setminus \de M \\  \Tilde{\mathsf{C}}^\mathscr{C}_{S \cap S_\sigma, T \cap S_\sigma}(\mathbb{H}(T_{p_\sigma}M)) & p_\sigma \in \de M\setminus \de \de M\\  \Tilde{\mathsf{C}}^\mathscr{C}_{S \cap S_\sigma, T \cap S_\sigma,U\cap S_\sigma}(\bbQ(T_{p_\sigma}M))& p_\sigma \in \de\de M\end{cases}$$
\end{enumerate}
where, for $Y$ a quadrant of $X$, we have 
$\Tilde{\mathsf{C}}^\mathscr{C}_{S,\varnothing,\varnothing}(Y) = \Tilde{\mathsf{C}}^\mathscr{C}_{\varnothing,T, \varnothing}(Y) = \{pt\}$, 
$ \Tilde{\mathsf{C}}^\mathscr{C}_{\varnothing,\varnothing,\{pt\}}(Y)\cong I$, and for $|S \sqcup T \sqcup U| \geq 2$ we define
\begin{equation}
\Tilde{\mathsf{C}}^\mathscr{C}_{S,T,U}(Y) := \coprod_{\substack{\mathfrak{P}=\{S_1,\ldots,S_k\}\\ S\sqcup T\sqcup U = \sqcup_\sigma S_\sigma, k \geq 2 }}\left\lbrace (y_\sigma,c_\sigma)\ \bigg|\ (y_\sigma, c_\sigma)\ \mathfrak{P}\text{-collapsed $(S,T,U)$-configuration in $Y$} \right\rbrace \bigg/ (\de Y \times \R_{>0}) 
\end{equation}
\end{defn}
This compactified configuration space has three types of boundary strata: Strata where a set of bulk points collapses in the bulk (called Type I strata), strata where a subset of bulk and boundary points collapses at the boundary (called Type II strata), and strata where a subset of all points collapses to a corner point (called Type III strata): 
\begin{multline}
\de\mathsf{C}_{S,T,U}(M,\de M, \de \de M) = \coprod_{S' \subseteq S} \de^\textnormal{I}_{S'} \mathsf{C}_{S,T,U}(M,\de M, \de \de M) \\ \amalg\coprod_{S' \subseteq S \sqcup T} \de_{S'}^{\textnormal{II}} \mathsf{C}_{S,T,U}(M,\de M, \de \de M) \amalg\coprod_{S' \subseteq S \sqcup T \sqcup U} \de_{S'}^{\textnormal{III}}\mathsf{C}_{S,T,U}(M,\de M, \de \de M) 
\end{multline}
\begin{rem} 
At this point, one can generalize the definitions above to that of compactifications of configuration spaces on stratified manifolds, with strata of any codimension. This is required for the extension of perturbative quantization to fully extended theories. 
\end{rem}
\begin{notation}
For a manifold $M$ without boundary, we also denote the compactified configuration space of $n$ points  $\mathsf{C}_{[n]}(M)$ on $M$ by $\mathsf{C}_n(M)$ (here $[n] = \{1,\ldots, n\}$). Moreover, for a manifold $M$ with boundary, we denote the compactified configuration space $\mathsf{C}_{[n],[m]}(M)$ of $n$ points on the bulk of $M$ and $m$ points on the boundary $\de M$ of $M$ by $\mathsf{C}_{[n],[m]}(M,\de M)$. We will also write $\mathsf{C}_\Gamma(M)$ for $\mathsf{C}_{[n],[m]}(M,\de M)$, if $\Gamma$ is a graph with $n+m$ vertices, $n$ vertices in the bulk of $M$ and $m$ vertices on $\de M$. Moreover, we will write $\mathsf{C}^\mathscr{C}_{n,m}(M)$ (or $\mathsf{C}^\mathscr{C}_\Gamma(M)$) for $\mathsf{C}^\mathscr{C}_{[n],[m],\varnothing}(M, \de M, \de\de M)$, if $M$ is a manifold with corners.
\end{notation}

\section{Deformation quantization and the Poisson Sigma Model}\label{app:PSM}
In this section we recollect some aspects of Kontsevich's star product \cite{K, CI, CKTB}, its globalization construction \cite{CFT,CF3,BCM,D}, and recall the relation with the Poisson Sigma Model \cite{CF1,CF2}. 

\subsection{Kontsevich's formality map on $\R^d$}
\label{app:kontsevich}
Kontsevich's formality map is an $L_{\infty}$ (quasi-iso)morphism from multivector fields $T_{poly}\R^d:=\Gamma\left(\bigwedge^\bullet T\R^d\right)$ to multidifferential operators $D_{poly}^\bullet\R^d$ on $\R^d$. As such it consists of a family of maps 

\begin{align}
\begin{split}
\calU_n\colon \Gamma\left(\bigwedge^{k_1}T\R^d\right)\oplus\dotsm \oplus\Gamma\left(\bigwedge^{k_n} T\R^d\right)&\to D^\bullet_{poly}\R^d\\
(\xi_1,\ldots,\xi_n)&\mapsto\calU_n(\xi_1,\ldots,\xi_n):=\sum_{\Gamma\in \mathcal{G}_{n,\ell}}w_\Gamma B_{\Gamma,\xi_1,\ldots,\xi_n}, 
\end{split}
\end{align}
where $\mathcal{G}_{n,\ell}$ is the set of graphs with $n+\ell$ numbered vertices, with $\ell:=2-2n+\sum_{i=1}^nk_i$, such that the $j$th vertex for $1\leq j\leq n$ emanates exactly $k_j$ arrows (without short loops). Here $k_i$ represents the degree of the multivector field $\xi_i$. Note that $\calU_n(\xi_1,\ldots,\xi_n)$ acts on $\ell$ functions. Here $B_{\Gamma,\xi_1,\ldots,\xi_n}$ are multidifferential operators, depending a graph $\Gamma$ and also on the vector fields $\xi_1,\ldots,\xi_n$, and the $w_\gamma$ are weights corresponding to a graph $\Gamma$ as in \cite{K}. For a vector field $\xi$ (i.e. $\xi$ is of degree $1$) and a bivector field $\Pi$ (i.e. $\Pi$ is of degree $2$) we can define 
\begin{align}
P(\Pi)&:= \sum_{j=0}^{\infty} \frac{\varepsilon^j}{j!}\calU_j(\Pi,\ldots,\Pi), \\
A(\xi,\Pi)&:=\sum_{j=0}^{\infty}\frac{\varepsilon^j}{j!}\calU_{j+1}(\xi,\Pi,....,\Pi),\label{eq:defn_A}\\
F(\xi_1,\xi_2,\Pi)&:=\sum_{j=0}^\infty \frac{\varepsilon^j}{j!}\calU_{j+2}(\xi_1,\xi_2,\Pi,\ldots,\Pi).  \label{eq:defn_F}
\end{align}
We have chosen the letters in this way, because later we will think of $P$ to be Kontsevich's star product for $\Pi$ a given Poisson tensor, $A$ as a connection 1-form and $F$ as its curvature. Let us take a look at some of the graphs appearing for some chosen multivector fields. For example, for a bivector field $\Pi$, we get that the term $U_{1}(\Pi)$ corresponds to the first graph of Figure \ref{formality}, whereas for a multivector field $\mathcal{V}$ of degree $r$ we get for $\calU_1(\mathcal{V})$ the second graph of Figure \ref{formality}. Let now $\xi$ be a vector field. Note that the number $\ell$ for $\calU_n(\xi,\Pi,\ldots,\Pi)$ will always be $1$ for every $n$, which implies that $A(\xi,\Pi)$ takes a smooth map $f$ as an argument.

\begin{figure}[ht]
\centering
\subfigure[Graph corresponding to a bivector field $\Pi$]{
\tikzset{
particle/.style={thick,draw=black},
particle2/.style={thick,draw=blue},
avector/.style={thick,draw=black, postaction={decorate},
    decoration={markings,mark=at position 1 with {\arrow[black]{triangle 45}}}},
gluon/.style={decorate, draw=black,
    decoration={coil,aspect=0}}
 }
\begin{tikzpicture}[x=0.07\textwidth, y=0.07\textwidth]
\node[](1) at (0,0){};
\node[](2) at (0.3,0.7){};
\node[](3) at (-0.8,-0.3){};
\node[](4) at (0.8,-0.3){};
\node[] (v) at (0,0.5){$\Pi$};
\node[] (d1) at (-1,-2){$f_1$};
\node[] (d2) at (1,-2){$f_2$};
\draw[fermion] (0,0)--(-1,-1.5);
\draw[fermion] (0,0)--(1,-1.5);
\draw[particle] (-2,-1.5)--(2,-1.5);
\node[vertex](v) at (0,0){};
\node[vertex](vert1) at (-1,-1.5){};
\node[vertex](vert2) at (1,-1.5){};
\end{tikzpicture}
}
\quad
\subfigure[Graph corresponding to a multivector field $\mathcal{V}$ of degree $r$, where $f_1,\ldots,f_r\in C^\infty(M)$]{
\tikzset{
particle/.style={thick,draw=black},
particle2/.style={thick,draw=blue},
avector/.style={thick,draw=black, postaction={decorate},
    decoration={markings,mark=at position 1 with {\arrow[black]{triangle 45}}}},
gluon/.style={decorate, draw=black,
    decoration={coil,aspect=0}}
 }
\begin{tikzpicture}[x=0.07\textwidth, y=0.07\textwidth]
\node[](1) at (0,0){};
\node[](2) at (0.3,0.7){};
\node[](3) at (-0.8,-0.3){};
\node[](4) at (0.8,-0.3){};
\draw[fermion] (0,0)--(-1,-1.5);
\node[] (v) at (0,0.5){$\mathcal{V}$};
\node[] (1) at (-1,-2){$f_1$};
\node[] (2) at (-0.5,-2){$f_2$};
\node[] (3) at (0,-2){$f_3$};
\node[] (4) at (0.5,-2){$f_4$};
\node[] (5) at (1,-2){$f_5$};
\node[] (6) at (1.5,-2){$f_6$};
\node[] (1) at (5,-2){$f_r$};
\draw[fermion] (0,0)--(1,-1.5);
\draw[fermion] (0,0)--(-0.5,-1.5);
\draw[fermion] (0,0)--(0,-1.5);
\draw[fermion] (0,0)--(0.5,-1.5);
\draw[fermion] (0,0)--(1.5,-1.5);
\node[] (p) at (2,-1){$\dotsm$};
\draw[particle] (-2,-1.5)--(6,-1.5);
\draw[fermion] (0,0)--(5,-1.5);
\node[vertex](vert) at (0,0){};
\node[vertex](vert1) at (-1,-1.5){};
\node[vertex](vert2) at (-0.5,-1.5){};
\node[vertex](vert3) at (0,-1.5){};
\node[vertex](vert4) at (0.5,-1.5){};
\node[vertex](vert5) at (1.5,-1.5){};
\node[vertex](vert6) at (1,-1.5){};
\node[vertex](vert7) at (5,-1.5){};
\end{tikzpicture}
}
\caption{The graphs $\calU_1(\Pi)$ and $\calU_1(\mathcal{V})$.}
\label{formality}
\end{figure}
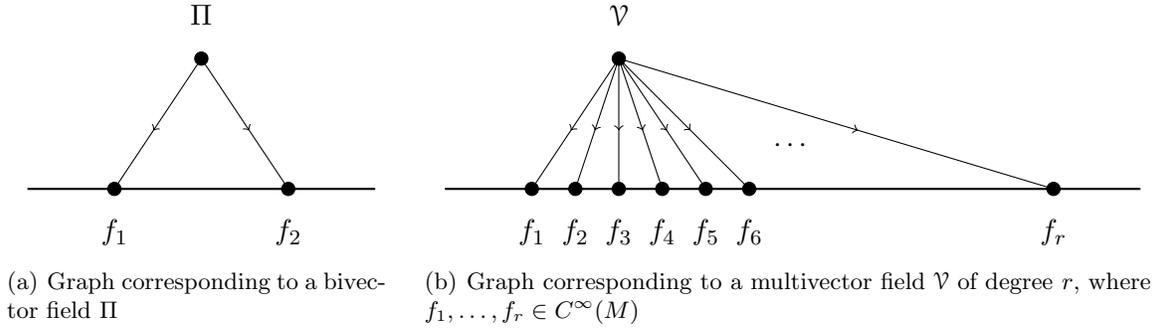
We want to look at graphs appearing for higher terms in $A$. We can, e.g., consider the $n=3$ term, i.e. $\calU_3(\xi,\Pi,\Pi)$. Some example of graphs in $\mathcal{G}_{3,1}$, which are taken in account for the sum, are given in Figure \ref{formality_A}.

\begin{figure}[ht]
\centering
\subfigure[]{
\tikzset{
particle/.style={thick,draw=black},
particle2/.style={thick,draw=blue},
avector/.style={thick,draw=black, postaction={decorate},
    decoration={markings,mark=at position 1 with {\arrow[black]{triangle 45}}}},
gluon/.style={decorate, draw=black,
    decoration={coil,aspect=0}}
 }
\begin{tikzpicture}[x=0.05\textwidth, y=0.05\textwidth]
\node[](1) at (0,0){};
\node[](2) at (0.3,0.7){};
\node[](3) at (-0.8,-0.3){};
\node[](4) at (0.8,-0.3){};
\node[] (v1) at (-1.5,0){$\xi$};
\node[] (v2) at (1,1.5){$\Pi$};
\node[] (v3) at (2.5,-0.5){$\Pi$};
\node[] (d1) at (-1,-2){$f$};
\draw[fermion] (-1,0)--(-1,-1.5);
\draw[fermion] (1,1)--(-1,-1.5);
\draw[fermion] (1,1) to (-1,0);
\draw[fermion] (2,-0.5)--(-1,-1.5);
\draw[fermion] (2,-0.5) to (1,1);
\draw[particle] (-2,-1.5)--(2,-1.5);
\node[vertex](vert1) at (-1,0){};
\node[vertex](vert2) at (1,1){};
\node[vertex](vert3) at (2,-0.5){};
\node[vertex](vert4) at (-1,-1.5){};
\end{tikzpicture}
}
\quad
\subfigure[]{
\tikzset{
particle/.style={thick,draw=black},
particle2/.style={thick,draw=blue},
avector/.style={thick,draw=black, postaction={decorate},
    decoration={markings,mark=at position 1 with {\arrow[black]{triangle 45}}}},
gluon/.style={decorate, draw=black,
    decoration={coil,aspect=0}}
 }
\begin{tikzpicture}[x=0.05\textwidth, y=0.05\textwidth]
\node[](1) at (0,0){};
\node[](2) at (0.3,0.7){$i_1$};
\node[](3) at (-0.8,-0.3){};
\node[](3) at (0.7,0.1){$i_3$};
\node[](3) at (-0.2,1.3){$i_5$};
\node[](3) at (1.8,0.8){$i_2$};
\node[](3) at (0.7,-1.2){$i_4$};
\node[](4) at (0.8,-0.3){};
\node[] (v1) at (-1.5,0){$\xi$};
\node[] (v2) at (1,1.5){$\Pi$};
\node[] (v3) at (2.5,-0.5){$\Pi$};
\node[] (d1) at (-1,-2){$f$};
\draw[fermion] (-1,0) to [bend left](1,1);
\draw[fermion] (2,-0.5)--(-1,0);
\draw[fermion] (1,1) to [bend left](-1,0);
\draw[fermion] (2,-0.5)--(-1,-1.5);
\draw[fermion] (1,1)--(2,-0.5);
\draw[particle] (-2,-1.5)--(2,-1.5);
\node[vertex](vert1) at (-1,0){};
\node[vertex](vert2) at (2,-0.5){};
\node[vertex](vert3) at (1,1){};
\node[vertex](vert4) at (-1,-1.5){};
\end{tikzpicture}
}
\subfigure[]{
\tikzset{
particle/.style={thick,draw=black},
particle2/.style={thick,draw=blue},
avector/.style={thick,draw=black, postaction={decorate},
    decoration={markings,mark=at position 1 with {\arrow[black]{triangle 45}}}},
gluon/.style={decorate, draw=black,
    decoration={coil,aspect=0}}
 }
\begin{tikzpicture}[x=0.05\textwidth, y=0.05\textwidth]
\node[](1) at (0,0){};
\node[](2) at (0.3,0.7){};
\node[](3) at (-0.8,-0.3){};
\node[](4) at (0.8,-0.3){};
\node[] (v1) at (-1.5,0){$\xi$};
\node[] (v2) at (1,1.5){$\Pi$};
\node[] (v3) at (2.5,-0.5){$\Pi$};
\node[] (d1) at (-1,-2){$f$};
\draw[fermion] (-1,0)--(-1,-1.5);
\draw[fermion] (1,1) to [bend left](-1,0);
\draw[fermion] (1,1) to [bend right](-1,0);
\draw[fermion] (2,-0.5) to [bend left](-1,0);
\draw[fermion] (2,-0.5) to [bend right](1,1);
\draw[particle] (-2,-1.5)--(2,-1.5);
\node[vertex](vert1) at (-1,0){};
\node[vertex](vert2) at (1,1){};
\node[vertex](vert3) at (2,-0.5){};
\node[vertex](vert4) at (-1,-1.5){};
\end{tikzpicture}
}
\caption{Example of graphs in $\mathcal{G}_{3,1}$.}
\label{formality_A}
\end{figure}
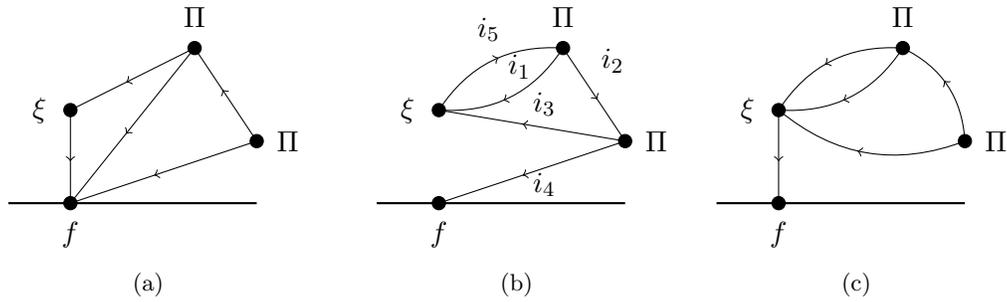

We can also explicitly say what the differential operator given by a graph will be. E.g. for the graph as in \ref{formality_A} (b) we get 
\begin{equation}
\partial_{i_1}\partial_{i_3}\xi^{i_5}\partial_{i_2}\partial_{i_2}\Pi^{i_3i_4}\partial_{i_5}\Pi^{i_1i_2}\partial_{i_4}(f).
\end{equation}
By definition of $F$, for every $n$ we get that $\ell=0$, i.e. the image of $\mathcal{U}_n$ will be a differential operator of degree zero, which is a smooth function. Some examples for graphs in $\mathcal{G}_{3,0}$ are given in Figure \ref{formality_F}.

\begin{figure}[ht]
\centering
\subfigure[]{
\tikzset{
particle/.style={thick,draw=black},
particle2/.style={thick,draw=blue},
avector/.style={thick,draw=black, postaction={decorate},
    decoration={markings,mark=at position 1 with {\arrow[black]{triangle 45}}}},
gluon/.style={decorate, draw=black,
    decoration={coil,aspect=0}}
 }
\begin{tikzpicture}[x=0.05\textwidth, y=0.05\textwidth]
\node[](1) at (0,0){};
\node[](2) at (0.3,0.7){};
\node[](3) at (-0.8,-0.3){};
\node[](4) at (0.8,-0.3){};
\node[] (v1) at (-1.5,0){$\xi_1$};
\node[] (v2) at (1,1.5){$\xi_2$};
\node[] (v3) at (2.5,-0.5){$\Pi$};
\node[] (d1) at (-1,-2){};
\draw[fermion] (-1,0) to [bend right](2,-0.5);
\draw[fermion] (1,1) to (-1,0);
\draw[fermion] (2,-0.5)--(-1,0);
\draw[fermion] (2,-0.5) to (1,1);
\draw[particle] (-2,-1.5)--(2,-1.5);
\node[vertex](vert1) at (1,1){};
\node[vertex](vert2) at (2,-0.5){};
\node[vertex](vert4) at (-1,0){};
\end{tikzpicture}
}
\quad
\subfigure[]{
\tikzset{
particle/.style={thick,draw=black},
particle2/.style={thick,draw=blue},
avector/.style={thick,draw=black, postaction={decorate},
    decoration={markings,mark=at position 1 with {\arrow[black]{triangle 45}}}},
gluon/.style={decorate, draw=black,
    decoration={coil,aspect=0}}
 }
\begin{tikzpicture}[x=0.05\textwidth, y=0.05\textwidth]
\node[](1) at (0,0){};
\node[](2) at (0.3,0.7){};
\node[](3) at (-0.8,-0.3){};
\node[](4) at (0.8,-0.3){};
\node[] (v1) at (-1.5,0){$\xi_1$};
\node[] (v2) at (1,1.5){$\xi_2$};
\node[] (v3) at (2.5,-0.5){$\Pi$};
\node[] (d1) at (-1,-2){};
\draw[fermion] (-1,0) to (2,-0.5);
\draw[fermion] (1,1) to (-1,0);
\draw[fermion] (2,-0.5) to [bend left](1,1);
\draw[fermion] (2,-0.5) to [bend right](1,1);
\draw[particle] (-2,-1.5)--(2,-1.5);
\node[vertex](vert1) at (1,1){};
\node[vertex](vert2) at (2,-0.5){};
\node[vertex](vert4) at (-1,0){};
\end{tikzpicture}
}
\quad
\subfigure[]{
\tikzset{
particle/.style={thick,draw=black},
particle2/.style={thick,draw=blue},
avector/.style={thick,draw=black, postaction={decorate},
    decoration={markings,mark=at position 1 with {\arrow[black]{triangle 45}}}},
gluon/.style={decorate, draw=black,
    decoration={coil,aspect=0}}
 }
\begin{tikzpicture}[x=0.05\textwidth, y=0.05\textwidth]
\node[](1) at (0,0){};
\node[](2) at (0.3,0.7){};
\node[](3) at (-0.8,-0.3){};
\node[](4) at (0.8,-0.3){};
\node[] (v1) at (-1.5,0){$\xi_1$};
\node[] (v2) at (1,1.5){$\xi_2$};
\node[] (v3) at (2.5,-0.5){$\Pi$};
\node[] (d1) at (-1,-2){};
\draw[fermion] (1,1) to [bend right](-1,0);
\draw[fermion] (-1,0) to [bend right](1,1);
\draw[fermion] (2,-0.5)--(-1,0);
\draw[fermion] (2,-0.5) to (1,1);
\draw[particle] (-2,-1.5)--(2,-1.5);
\node[vertex](vert1) at (1,1){};
\node[vertex](vert2) at (2,-0.5){};
\node[vertex](vert4) at (-1,0){};
\end{tikzpicture}
}
\caption{Example of graphs in $\mathcal{G}_{3,0}$.}
\label{formality_F}
\end{figure}
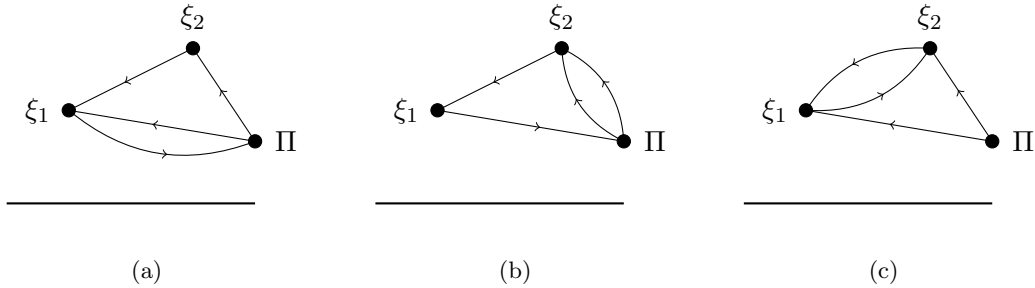

\subsection{Notions of formal geometry}
\label{app:formal_geometry}
We recall the most important notions of formal geometry as in \cite{GK,B} following the presentation as in \cite{CF3} and \cite{BCM}. For a smooth manifold $\mathscr{P}$ we can consider a formal exponential map $\varphi$ on $\mathscr{P}$, such that for $x\in\mathscr{P}$ we have $\varphi_x\colon T_x\mathscr{P}\to  \mathscr{P}$, and we define a vector field $R\in \Gamma(T^*\mathscr{P}\otimes T\mathscr{P}\otimes \widehat{S}T^*\mathscr{P})$, which is a 1-form with values in derivations of $\widehat{S}T^*\mathscr{P}$. Here $\widehat{S}$ denotes the completed symmetric algebra. In local coordinates we have $R=R_i\dd x^{i}$ with
\begin{equation}
R_i(x;y) = \left(\left(\frac{\partial \varphi_x}{\partial y}\right)^{-1}\right)^k_j\frac{\partial\varphi_x^{j}}{\partial x^{i}}\frac{\de }{\de y^k} =: Y^k_i(x;y)\frac{\de}{\de y^k}
\end{equation}
Then we can define the classical Grothendieck connection $D_\mathsf{G}:=\dd +R$, which is flat. For a vector field $\xi=\xi\frac{\partial}{\partial x^{i}}$ we have $D^\xi_\mathsf{G}=\xi+\widehat{\xi}$, where 
\begin{equation}
\label{eq:xi_hat}
\widehat{\xi}(x;y)=\iota_\xi R(x;y)=\xi^{i}Y_i^k(x;y)\frac{\partial}{\partial y^k}. 
\end{equation}

\subsection{Globalization}
\label{app:globalization}
Now let us describe how to generalize the above procedure to an arbitrary Poisson manifold $(\mathscr{P},\Pi)$. Namely, let $x \in \mathscr{P}$, and $\varphi$ a formal exponential map on $\mathscr{P}$. Then $\mathsf{T}\varphi_x^*\Pi$, the Taylor expansion of $\Pi$ around $x$ defined using $\varphi$, is a Poisson tensor on $\Hat{S}T^*_x\mathscr{P}$. Any choice of coordinates on $T_xM$ now allows us to identify $\Hat{S}T^*_x\mathscr{P} \cong \R[[y_1,\ldots,y_d]]$ and define Kontsevich's star product $P(\mathsf{T}\varphi^*_x\Pi)$. See \cite{CFT} for a discussion of the equivariance of this construction in the choice of coordinates. In this way we get a new bundle $\mathcal{E}:= \Hat{S}T^*\mathscr{P}[[\varepsilon]]$ of $\star$-algebras. One can use the Grothendieck connection defined in \ref{app:formal_geometry} to give a description of a subalgebra $\mathcal{A}\subset \Gamma(\mathcal{E})$ which is a deformation quantization of $C^\infty(\mathscr{P})$ seen as a subalgebra of $\Gamma(\mathcal{E})$. Formally we have
\begin{equation}
\diagram
\centering
\Gamma(\calE)\supset C^\infty(\mathscr{P})&\rTo^{\textsf{Deformation Quantization}}&\mathcal{A}\subset \Gamma(\mathcal{E}).\\
\enddiagram
\end{equation}
The algebra $\mathcal{A}$ is given by closed sections under a deformation of the Grothendieck connection, which is defined in two steps:
For a tangent vector $\xi\in T_x\mathscr{P}$, we let
\begin{equation}
\label{loc_def_conn}
\mathcal{D}_\mathsf{G}^\xi:=\xi+A\left(\widehat{\xi},\mathsf{T}\varphi^*_x\Pi\right)= D_\mathsf{G}^\xi+O(\varepsilon),
\end{equation}
where again we denote by $\mathsf{T}\varphi^*_x\Pi$ the Poisson tensor $\Pi$ lifted to a formal neighborhood and $\Hat{\xi}$ is defined as in \eqref{eq:xi_hat}. One can write 
\begin{equation}
\label{AppB:connection_G}
\mathcal{D}_{\mathsf{G}} = \dd + A(R,\mathsf{T}\varphi^*\Pi) 
\end{equation}
interpreting $A(R,\mathsf{T}\varphi^*\Pi)$ as a one-form valued in differential operators on $\mathcal{E}$. At some point $x \in \mathscr{P}$, in coordinates $x^i$ around $x$, it is given by 
\begin{equation} 
A(R,\mathsf{T}\varphi_x^*\Pi) = \dd x^i A(R_i(x;y),\mathsf{T}\varphi^*_x\Pi) =  \dd x^i A\left(Y_i^k(x;y)\frac{\partial}{\partial y^k},\mathsf{T}\varphi^*_x\Pi\right).   
\end{equation}
One can then show \cite{CFT} that $\mathcal{D}_\mathsf{G}$ is a globally defined connection on $\Gamma(\mathcal{E})$, a derivation, and that $(\mathcal{D}_\mathsf{G})^2$ is an inner derivation, i.e. 
\begin{equation}
(\mathcal{D}_\mathsf{G})^2\sigma=[F^\mathscr{P},\sigma]_\star:=F^\mathscr{P}\star\sigma-\sigma\star F^\mathscr{P}, 
\end{equation}
for any $\sigma\in \Gamma(\mathcal{E})$, where $\gls{F^P}$ is the \textsf{Weyl curvature} tensor of $\mathcal{D}_\mathsf{G}$ given by $F^\mathscr{P}(\xi_1,\xi_2):=F(\widehat{\xi}_1,\widehat{\xi}_2,\mathsf{T}\varphi^*\Pi)$, where $\xi_1,\xi_2\in T_x\mathscr{P}$ are two tangent vectors on $\mathscr{P}$. More, precisely, $F^\mathscr{P}$ is a 2-form valued in sections of $\mathcal{E}$ which in local coordinates can be expressed as 
\begin{equation}
F^\mathscr{P}_x = \dd x^i \wedge \dd x^j F(R_i(x;y),R_j(x;y),\mathsf{T}\varphi^*_x\Pi).
\end{equation}
For the Weyl tensor we get $\mathcal{D}_\mathsf{G}F^\mathscr{P}=0$. The task is to modify the globalized connection $\mathcal{D}_\mathsf{G}$ slightly more, so that it becomes flat but still remaining a derivation. One can set\footnote{For any two $\mathcal{E}$-valued $1$-forms $\gamma=\gamma_i\dd x^{i},\sigma=\sigma_j\dd x^{j}\in\Omega^1(\mathscr{P},\mathcal{E})$ one defines their star product by $\gamma\star\sigma:=(\gamma_i\star\sigma_j)\dd x^{i}\land \dd x^{j}$} 
\begin{equation}
\label{def_conn}
\overline{{\mathcal{D}}}_\mathsf{G}:=\mathcal{D}_\mathsf{G}+[\gamma,\enspace]_\star,
\end{equation}
and observe that for any $1$-form $\gamma\in\Omega^1(\mathscr{P},\mathcal{E})$ this connection is a derivation. Moreover, its Weyl curvature tensor is then given by
\begin{equation}
\label{flatness}
\gls{barF^P}=F^\mathscr{P}+\mathcal{D}_\mathsf{G}\gamma+\gamma\star\gamma. 
\end{equation}
We call \eqref{loc_def_conn} the \textsf{deformed Grothendieck connection} and \eqref{def_conn} the \textsf{modified deformed Grothendieck connection}. One then needs to find $\gamma\in \Omega^1(\mathscr{P},\mathcal{E})$ such that $\overline{F}^\mathscr{P}=0$, which implies that $(\overline{\mathcal{D}}_\mathsf{G})^2=0$, so that $\overline{{\mathcal{D}}}_\mathsf{G}$-closed sections will form the algebra $\mathcal{A}$ as a deformation quantization of $C^\infty(\mathscr{P})$. If we compute $(\overline{{\mathcal{D}}}_\mathsf{G})^2$ explicitly, by using \eqref{def_conn} we get 
\begin{equation}
\label{square_conn}
(\overline{{\mathcal{D}}}_\mathsf{G})^2=\underbrace{(\mathcal{D}_\mathsf{G})^2}_{:=[F^\mathscr{P},\enspace]_\star}+\mathcal{D}_\mathsf{G}[\gamma,\enspace]_\star+[\gamma,[\gamma,\enspace]_\star]_\star.
\end{equation}
More precisely, $\gamma$ has to satisfy 
\begin{equation}
F^\mathscr{P}+\mathcal{D}_\mathsf{G}\gamma+\gamma\star\gamma=0.
\end{equation}
The existence of such a $\gamma$ was shown in \cite{CF3,CFT} by homological perturbation theory. 
One can actually construct $\gamma$ to be a solution of the more general equation given by 
\begin{equation}
\label{gen_curv}
\overline{F}^\mathscr{P}_\omega=F^\mathscr{P}+\varepsilon\omega+\mathcal{D}_\mathsf{G}\gamma+\gamma\star\gamma=0,
\end{equation}
where $\omega\in\Omega^2(\mathscr{P},\calE)$ such that $\mathcal{D}_{\mathsf{G}}\omega=0$ and $[\omega,\enspace]_\star=0$ \cite{CFT}. 

Now we want to focus on some special cases.
We want to look at two important examples of Poisson structures. 

\subsubsection{Constant Poisson structure} 
The situation of a constant Poisson structure is a first example to think about. Let $(\mathscr{P},\Pi)$ be a Poisson manifold with constant Poisson structure $\Pi$ and $\xi\in T_x\mathscr{P}$ for $x\in \mathscr{P}$ be a fixed tangent vector. By the definition of $A$, and the fact that each vertex has only one outgoing and no incoming arrow, we get $A(\widehat{\xi},\mathsf{T}\varphi^*\Pi)=\widehat{\xi}$, which leads to the fact that 
\begin{equation}
\mathcal{D}^\xi_\mathsf{G}=(\xi+\widehat{\xi}) = D_\mathsf{G}^{\xi}.
\end{equation}
Therefore we get $(\mathcal{D}_\mathsf{G})^2=0$ and thus $F^\mathscr{P}=0$. We can then choose $\gamma=0$.

\subsubsection{Linear Poisson structure} 
Let now $(\mathscr{P}=\mathfrak{g}^*,\Pi)$ be a Poisson manifold with linear Poisson structure $\Pi(x)=\Pi_k^{ij}x^k\frac{\partial}{\partial x^{i}}\land \frac{\partial}{\partial x^j}$, where $\Pi^{ij}_k$ represent the structure constants of a Lie algebra $\mathfrak{g}$, and $\xi\in T_x\mathscr{P}$ for $x\in\mathscr{P}$ be a fixed tangent vector. As in the constant case, we observe that $A(\widehat{\xi},\mathsf{T}\varphi^*\Pi)=\widehat{\xi}$, which is the case since the integral of a bulk vertex with one incoming and one outgoing arrow is zero, and since there is at most one incoming arrow for each vertex. Again we may choose $\gamma=0$.

\subsection{Connection to the Poisson Sigma Model} 
\label{app:conn_PSM}
In \cite{CF1} and \cite{CF2} it was shown that Kontsevich's formality map on $\R^d$ can be intepreted as the perturbative computation of  expectation values of observables of the Poisson Sigma Model on the upper half plane (or respectively the disk) with values in $\R^d$. The graphs which appear in the construction of Kontsevich's star product on Poisson manifolds \cite{K} are given on the upper half plane, where they can collapse, according to the boundary of the configuration space, on the boundary of the upper half plane. This means that the graphs that appear in the Poisson Sigma Model are exactly the graphs that appear for Kontsevich's star product. More precisely, if one considers the disk $D$ in $\R^2$ and the classical action of the Poisson Sigma Model on $D$ given by $S_D[(X,\eta)]=\int_D\left(\langle \eta,\dd X\rangle+\frac{1}{2}\langle\Pi(X),\eta\land\eta\rangle\right)$, we can asymptotically write Kontsevich's star product for two smooth maps $f$ and $g$ as a perturbative expansion of the following path integral:
\begin{equation}
f\star g(x)=\int_{X(\infty)=x}f(X(0))g(X(1))\ee^{\frac{\I}{\hbar}S_D[(X,\eta)]},
\end{equation}
where $0,1,\infty$ represent some marked points on the boundary of $D$. Note that $x\in \Map(D,\R^d)$ is a constant map, i.e. the we get a local representation of the star product. If one considers a general Poisson manifold $(\mathscr{P},\Pi)$, one can consider the constant map $x\in\Map(D,\mathscr{P})$ as a point sitting in $\mathscr{P}$ giving a local product on each fiber. As already described in \ref{app:globalization}, one can then algebraically construct the star product on all of $\mathscr{P}$.

\section{On the Propagator}
\label{app:prop}
We have an explicit propagator for the Poisson Sigma Model, i.e. using the superfields of it, on a disk with alternating boundary conditions, which was computed in \cite{CF6}, in \cite{CT1} and, in full generality, in \cite{AF}. 

\subsection{Construction of the branes}
Consider an $n$-sided polygon $P_n=u(\mathbb{H}^+)$ where $u:\mathbb{H}^+\to P_n$ is a suitable homeomorphism between the compactified complex upper half plane $\mathbb{H}^+$ and $P_n$, depending on the number of the branes considered.
Let $G_{S_i}$, be the relevant superpropagators for the Poisson Sigma Model with $n$ branes defined by constraints $C_j=\{x^{\mu_j}=0\mid \mu_j\in I_j\}$ (also called \textsf{branes}) and index sets $S_1=I_1^C\cap I_2\cap I_3^C\cap\dotsm\cap I_n$, $S_2=I_1\cap I_2^C\cap I_3\cap\dotsm \cap I_n^C$ for $n$ even, and $S_1=I_1^C\cap I_2\cap I_3^C\cap\dotsm\cap I_n^C$, $S_2=I_1\cap I_2^C\cap\dotsm \cap I_n$ for $n$ odd, which are called \textsf{relevant}. It turns out that the $C_i\subset \mathscr{P}$ are coisotropic submanifolds of $\mathscr{P}$ \cite{CF6}. 
\subsection{Constructing integral kernels}
The integral kernels $\theta(Q,P)_{S_i}:=-\frac{\I}{\hbar}\langle \widehat{\mathsf{X}}^\bullet(Q)\widehat{\boldsymbol{\eta}}_\bullet(P)\rangle$ for the two brane case are given by:
\begin{align}
\theta(Q;P)_{S_i}&=\frac{1}{2\pi}\dd \arg\frac{(u-v)(\bar u-v)}{(\bar u+v)(u+v)},\\
\theta(Q,P)_{S_2}&=\frac{1}{2\pi}\dd\arg\frac{(u-v)(\bar u+v)}{(\bar u-v)(u+v)},
\end{align}
where $P_2:=u(\mathbb{H}^+)$ with $u(z)=\sqrt{z}$, $v:=u(w)$, $\dd=\dd_u+\dd_v$. We identify $(P,Q)$ with the couple $(u,v)$. Consider e.g. $P_2$ to be the worldsheet disk $\Sigma$ with boundary $\partial\Sigma=\bigsqcup_{1\leq j\leq 6}J_j$ (we denote the intervals here by $J$ instead of $I$ such that there is no confusion with the index sets) and the branes $C_1=\{x^{\mu_1}=0\mid \mu_1\in I_1=\{1,\ldots,n\}\}$ and $C_2=\{x^{\mu_2}=0\mid\mu_2\in I_2=\varnothing\}$, which correspond to the boundary conditions of $\partial_1\Sigma$ and $\partial_2^{\text{tot}}\Sigma$ respectively. The components $\partial_1\Sigma$ and $\partial_2^{\text{tot}}\Sigma$ are such that $\partial\Sigma=\partial_1\Sigma\sqcup\partial_2^{\text{tot}}\Sigma$, where $\partial_1\Sigma$ is chosen to be some $J_1$ endowed with the $\frac{\delta}{\delta\E}$-polarization and $\partial_2^\text{tot}\Sigma=\bigsqcup_{2\leq j\leq 6} J_j$ such that $J_j$ is endowed with the $\frac{\delta}{\delta\mathbb{X}}$-polarization and with the boundary condition $\hateta\equiv0$ for $j$ odd and even respectively. Now we get $S_1=I_1^C\cap I_2=\varnothing$ and $S_2=I_1\cap I_2^C=\{1,\ldots,n\}$. Now $P_2$ is defined by $P_2=u(\mathbb{H}^+)$, where $u$ is the map $z\mapsto \sqrt{z}$. Points $(P,Q)\in P_2\times P_2$ are represented respectively by a pair of complex numbers $(u,v)$ in the first quadrant, with $u=u(z)$, $v=u(w)$ for all $(z,w)\in\mathbb{H}^+\times\mathbb{H}^+$. The boundary $\partial_1 P_2$ (corresponding to $\partial_1\Sigma$) is given by the positive imaginary axis, while $\partial_2P_2$ (corresponding to $\partial_2^{\text{tot}}\Sigma$) is given by the positive real axis. 

\subsection{Construction of superpropagators}
The boundary conditions imposed by the index sets $S_i$ are $\theta(v,u\in\partial_1P_2)_{S_1}=\theta(u\in\partial_2P_2,u)_{S_1}=0$, $\theta(v,u\in \partial_2P_2)_{S_2}=\theta(v\in\partial_1P_2,u)_{S_2}=0$. Let 
\begin{align}
\label{mirror1}
\psi(u,v)_{S_1}&=\arg\frac{(u-v)(\bar u-v)}{(\bar u+v)(u+v)},\\
\label{mirror2}
\psi(u,v)_{S_2}&=\arg\frac{(u-v)(\bar u+v)}{(\bar u-v)(u+v)},
\end{align}
which satisfy the same boundary conditions as $\theta(v,u)_{S_i}$. Now for vanishing cohomology, we get the following Theorem.
\begin{thm}
The integral kernels for the superpropagators $G_{S_i}$ in presence of two branes are given by 
\begin{equation}
\label{superpropagators}
\theta(v,u)_{S_i}=\frac{1}{2\pi}\dd\psi(u,v)_{S_i},
\end{equation}
with \textsf{angle maps} \eqref{mirror1} and \eqref{mirror2}. The integral kernels satisfy the additional boundary conditions $\theta(v,u)_{S_1}=\theta(v,\bar u)=\theta(-\bar v,u)_{S_1}$, $\theta(v,u)_{S_2}=\theta(v,-\bar u)_{S_2}=\theta(\bar v,u)_{S_2}$, i.e. every boundary component of $P_2$ is labeled by a boundary condition for both the variables $(u,v)$. By construction $\theta(v,u)_{S_1}=\theta(u,v)_{S_2}$, $\theta(v,u)_{S_2}=\theta(u,v)_{S_1}$.
\end{thm}
\subsection{Relation to Kontsevich's propagator}
Let $\phi$ be Kontsevich's angle 1-form. Then, one can show that 
\begin{align}
\theta(v,u)_{\A_1}&=\frac{1}{2\pi}\dd\arg\frac{(u-v)(u+v)}{(u+\bar v)(u-\bar v)}=\frac{1}{2\pi}\dd\arg\frac{(z-w)}{(z-\bar w)}=\frac{1}{2\pi}\dd\phi(z,w),\\
\theta(v,u)_{\A_2}&=\frac{1}{2\pi}\dd\arg\frac{(u-v)(u+v)}{(\bar u- v)(\bar u+ v)}=\frac{1}{2\pi}\dd\arg\frac{(z-w)}{(\bar z- w)}=\frac{1}{2\pi}\dd\phi(w,z),
\end{align}
where $\A_1=I_2\cap I_2$ and $\A_2=I_1^C\cap I_2^C$.

\end{appendix}

\printglossary

\printbibliography

\end{document}